\documentclass[aps,pra,twocolumn,superscriptaddress,longbibliography]{revtex4-1}
\usepackage{amsfonts}
\usepackage{mathtools}
\usepackage{amsmath}
\usepackage{amssymb}
\usepackage{lipsum}
\usepackage{graphicx}
\usepackage[colorlinks=true,linkcolor=blue,citecolor=red,plainpages=false,pdfpagelabels]
{hyperref}
\usepackage{color}

\allowdisplaybreaks

\providecommand{\U}[1]{\protect\rule{.1in}{.1in}}
\newtheorem{theorem}{Theorem}

\newtheorem{corollary}{Corollary}

\newtheorem{definition}{Definition}

\newtheorem{lemma}{Lemma}

\newtheorem{proposition}{Proposition}
\newtheorem{remark}{Remark}

\newenvironment{proof}[1][Proof]{\noindent\textbf{#1.} }{\ \rule{0.5em}{0.5em}}

\def\Tr{\operatorname{Tr}}

\def\sq{\operatorname{sq}}
\def\SEP{\operatorname{SEP}}
\def\Ent{E}
\def\PPT{\operatorname{PPT}}
\def\supp{\operatorname{supp}}
\def\LOCC{\operatorname{LOCC}}
\def\T{\operatorname{T}}
\def\>{\rangle}
\def\<{\langle}
\def\({\left(}
\def\){\right)}
\def\[{\left[}
\def\]{\right]}
\def\V{\Vert}

\def\id{\operatorname{id}}
\let\oldemptyset\emptyset
\let\emptyset\varnothing

\newcommand{\ket}[1]{\left|{#1}\right\rangle}
\newcommand{\bra}[1]{\left\langle{#1}\right|}
\newcommand{\norm}[1]{\left\Vert{#1}\right\Vert}
\newcommand{\mc}[1]{\mathcal{#1}}
\newcommand{\wt}[1]{\widetilde{#1}}

\date{\today}

\begin{document}
\widetext
\title{ Entanglement and secret-key-agreement capacities of bipartite~quantum~interactions and read-only~memory~devices}

\author{Siddhartha Das}

\email{sidddas@ulb.ac.be} 

\affiliation{Hearne Institute for Theoretical Physics, Department of Physics and Astronomy, Louisiana State University, Baton Rouge, Louisiana, 70803, USA}
\affiliation{Centre for Quantum Information \& Communication (QuIC), \'{E}cole polytechnique de Bruxelles,   Universit\'{e} libre de Bruxelles, Brussels, B-1050, Belgium}

\author{Stefan B{\"a}uml}\email{stefan.baeuml@icfo.eu}

\affiliation{ICFO-Institut de Ciencies Fotoniques, The Barcelona Institute of Science and Technology, Av. Carl Friedrich Gauss 3, 08860 Castelldefels (Barcelona), Spain.}

\affiliation{QuTech, Delft University of Technology, Lorentzweg 1, 2628 CJ Delft, Netherlands}

\affiliation{NTT Basic Research Laboratories and NTT Research Center for Theoretical Quantum Physics, NTT Corporation, 3-1 Morinosato-Wakamiya, Atsugi, Kanagawa 243-0198, Japan}

\author{Mark M. Wilde}

\email{mwilde@lsu.edu}

\affiliation{Hearne Institute for Theoretical Physics, Department of Physics and Astronomy, Louisiana State University, Baton Rouge, Louisiana, 70803, USA}

\affiliation{Center for Computation and Technology, Louisiana State University, Baton Rouge, Louisiana 70803, USA}

\date{\today}

\begin{abstract}
A bipartite quantum interaction corresponds to the most general quantum interaction that can occur between two quantum systems in the presence of a bath. In this work, we determine bounds on the capacities of bipartite interactions for entanglement generation and secret key agreement between two quantum systems. Our upper bound on the entanglement generation capacity of a bipartite quantum interaction is given by a quantity called the bidirectional max-Rains information. Our upper bound on the secret-key-agreement capacity of a bipartite quantum interaction is given by a related quantity  called the bidirectional max-relative entropy of entanglement. We also derive tighter upper bounds on the capacities of bipartite interactions obeying certain symmetries. Observing that reading of a memory device is a particular kind of bipartite quantum interaction, we leverage our bounds from the bidirectional setting to deliver bounds on the capacity of a task that we introduce, called private reading of a wiretap memory cell. Given a set of point-to-point quantum wiretap channels, the goal of private reading is for an encoder to form codewords from these channels, in order to establish secret key with a party who controls one input and one output of the channels, while a passive eavesdropper has access to one output of the channels. We derive both lower and upper bounds on the private reading capacities of a wiretap memory cell. We then extend these results to determine achievable rates for the generation of entanglement between two distant parties who have coherent access to a controlled point-to-point channel, which is a particular kind of bipartite interaction.
\end{abstract}

\maketitle
\tableofcontents

\section{Introduction}
In general, any two-body quantum system of interest can be in contact with a bath, and part of the composite system may be inaccessible to observers possessing these systems. The effective interaction between given two constituent systems in the presence of the bath is known as a bipartite quantum interaction. It is well known that a closed quantum system evolves according to a unitary transformation \cite{Dbook81,SCbook95}.

Let $U^{\hat{H}}_{A'B'E'\to ABE}$ denote a unitary transformation associated to a Hamiltonian $\hat{H}$, which governs the underlying interaction between a two-body quantum system and a bath. Here $A'B'$ and $E'$ denote system labels for a two-body quantum system of interest and the inaccessible bath, respectively, at an initial time, and $AB$ and $E$ denote system labels for a two-body quantum system of interest and the inaccessible bath, respectively, at a final time when the evolution is complete.
The individual input systems $A'$, $B'$, and $E'$ and the respective output systems $A$, $B$, and $E$ can have different dimensions. Initially, in the absence of an interaction Hamiltonian $\hat{H}$, the bath is taken to be in a pure state and the systems of interest have no correlation with the bath; i.e., the state of the composite system $A'B'E'$ is of the form $\omega_{A'B'}\otimes \ket{0}\!\bra{0}_{E'}$, where $\omega_{A'B'}$ and $\ket{0}\!\bra{0}_{E'}$ are density operators of the systems $A'B'$ and $E'$, respectively. Under the action of the Hamiltonian $\hat{H}$, the state of the composite system transforms as
\begin{equation}\label{eq:bi-u}
\rho_{ABE}=U^{\hat{H}}(\omega_{A'B'}\otimes \ket{0}\!\bra{0}_{E'})(U^{\hat{H}})^\dag.
\end{equation}
Since the system $E$ in \eqref{eq:bi-u} is inaccessible, the evolution of the systems of interest is noisy in general. The noisy evolution of the bipartite system $A'B'$ under the action of Hamiltonian $\hat{H}$ is represented by a completely positive, trace-preserving (CPTP) map \cite{Sti55}, called a bipartite quantum channel:
\begingroup
\allowdisplaybreaks[0]
\begin{multline}\label{eq:bipartite-int}
\mc{N}^{\hat{H}}_{A'B'\to AB}(\omega_{A'B'})=\\
\Tr_{E}\{U^{\hat{H}}(\omega_{A'B'}\otimes \ket{0}\!\bra{0}_{E'})(U^{\hat{H}})^\dag\},
\end{multline}
\endgroup
where system $E$ represents inaccessible degrees of freedom. In particular, when the Hamiltonian $\hat{H}$ is such that there is no interaction between the composite system $A'B'$ and the bath $E'$, and $A'B'\simeq AB$, then $\mc{N}^{\hat{H}}$ corresponds to a bipartite unitary, i.e., $\mc{N}^{\hat{H}}(\cdot)=U^{\hat{H}}_{A'B'\to AB}(\cdot)(U^{\hat{H}}_{A'B'\to AB})^\dag$.  

In an information-theoretic setting, a bipartite quantum channel $\mc{N}_{A'B'\to AB}$ is also called \textit{bidirectional quantum channel} when system pairs $A',A$ and $B',B$ belong to two separate parties (cf.~\cite{BHLS03}).  

Depending on the kind of bipartite quantum interaction, there may be an increase, decrease, or no change in the amount of entanglement \cite{PV07,HHHH09} of a bipartite state after undergoing a bipartite interaction. As entanglement is one of the fundamental and intriguing quantum phenomena~\cite{EPR35,S35}, determining the entangling abilities of bipartite quantum interactions is pertinent. 

In this work, we focus on two different information-processing tasks relevant for bipartite quantum interactions, the first being entanglement distillation~\cite{BBPS96,BBP+97,Rai99} and the second secret key agreement~\cite{D05,DW05,HHHO05,HHHO09}. Entanglement distillation is the task of generating a maximally entangled state, such as the singlet state, when two separated quantum systems undergo a bipartite interaction. Whereas, secret key agreement is the task of extracting maximal classical correlation between two separated systems, such that it is independent of the state of the bath system, which an eavesdropper could possess. Both of these tasks are of practical interest: distilling pure maximally entangled states is useful for fundamental tasks such as teleportation \cite{BBC+93}, super-dense coding \cite{PhysRevLett.69.2881}, and distributed quantum computation, while distilled secret key is useful for private communication when combined with the one-time pad. Thus, it is of interest to know fundamental limitations for these tasks for the design of actual protocols, and this is what our bounds provide.

In an information-theoretic setting, a bipartite interaction between classical systems was first considered in \cite{Sha61} in the context of communication; therein, a bipartite interaction was called a two-way communication channel. In the quantum domain, bipartite unitaries have been widely considered in the context of their entangling ability, applications for interactive communication tasks, and the simulation of bipartite Hamiltonians in distributed quantum computation   \cite{BDEJ95,ZZF00,EJPP00,BRV00,NC00,CLP01,CDKL01,BHLS03,CLV04,JMZL17,DSW17}. These unitaries form the simplest model of non-trivial interactions in many-body quantum systems and have been used as a model of scrambling in the context of quantum chaotic systems~\cite{SS08b,HQRY16,DHW16}, as well as for the internal dynamics of a black hole~\cite{HP07} in the context of the information-loss paradox~\cite{Haw76}. More generally, \cite{CLL06} developed the model of a bipartite interaction or two-way quantum communication channel.  Bounds on the rate of entanglement generation in open quantum systems undergoing time evolution have also been discussed for particular classes of quantum dynamics~\cite{Bra07,DKSW17}.  

The maximum rate at which a particular task can be accomplished by allowing the use of a bipartite interaction a large number of times, is equal to the capacity of the interaction for the task. The entanglement generating capacity quantifies the maximum rate of entanglement that can be generated  from a bipartite interaction. Various capacities of a general bipartite unitary evolution were formalized in \cite{BHLS03}. Later, various capacities of a general two-way channel were discussed in \cite{CLL06}. The entanglement generating capacities  of bipartite unitaries for different communication protocols have been widely discussed in the  literature~\cite{ZZF00,LHL03,BHLS03,HL05,LSW09,WSM17,CY16}. Also, prior to our work here, it was an open question to find a non-trivial, computationally efficient upper bound on the entanglement generating capacity of a bipartite quantum interaction. Another natural direction left open in prior work is to determine other information-processing tasks for bipartite quantum interactions, beyond those discussed previously \cite{BHLS03,CLL06}.

In this paper, we determine bounds on the capacities of bipartite interactions for entanglement generation and secret key agreement. Observing that the read-out task of memory devices is a particular kind of bipartite quantum interaction (cf.~\cite{BRV00,Pir11}), we leverage our bounds from the bidirectional setting to deliver bounds on the capacity of a task that we introduce here, called private reading of a memory cell.  We derive both lower and upper bounds on the capacities of private reading protocols. We then extend these results to determine achievable rates for the generation of entanglement between two distant parties who have coherent access to a controlled point-to-point channel, which is a particular kind of bipartite interaction. 

Private reading is a quantum information-processing task in which 
a classical message from an encoder to a reader is delivered in a \emph{read-only} memory device. The message is encoded in such a way that a reader can reliably decode it, while a passive eavesdropper recovers no information about it. This protocol can be used for secret key agreement between two trusted parties. A physical model of a read-only memory device involves encoding the classical message using a \emph{memory cell}, which is a set of point-to-point quantum wiretap channels. Note that a point-to-point quantum wiretap channel is a channel that takes one input and produces two outputs. The reading task is restricted to information-storage devices that are read-only, such as a CD-ROM. One feature of a read-only memory device is that a message is stored for a fairly long duration if it is kept safe from tampering. One can read information from these devices many times without the eavesdropper learning about the encoded message. 

The strong converse bounds on the bidirectional quantum and private capacities of bidirectional channels presented in this work have also been stated, in abbreviated form and without proofs, in our companion paper \cite{bauml2018fundamental}. There we also compute the bounds on the bidirectional quantum capacity for several examples. In the current paper, we present a more comprehensive discussion of the results, including proofs and derivations, as well as a detailed overview of the underlying concepts. The present article also includes additional results on private reading, namely the computation of the non-adaptive private reading capacity of a wiretap memory cell presented in Theorem~\ref{thm:n-a-priv-read}, an alternative converse bound on the non-adaptive private reading capacity of an isometric memory cell presented in Proposition~\ref{prop:EsqBound}, and the study of entanglement generation from a coherent memory cell or controlled isometry, presented in Section~\ref{sec:coh-read}.

The organization of our paper is as follows. We set  notation and review basic definitions in Section~\ref{sec:review}. In Section~\ref{sec:ent-dist}, we derive a strong converse upper bound on the rate at which entanglement can be distilled from a bipartite quantum interaction. This bound is given by an information quantity that we call the bidirectional max-Rains information $R^{2\to 2}_{\max}({\mc{N}})$ of a bidirectional channel $\mc{N}$. The bidirectional max-Rains information is the solution to a semi-definite program and is thus efficiently computable. In Section~\ref{sec:priv-key}, we derive a strong converse upper bound on the rate at which a secret key can be distilled from a bipartite quantum interaction. This bound is given by a related information quantity that we call the bidirectional max-relative entropy of entanglement $E^{2\to 2}_{\max}(\mc{N})$ of a bidirectional channel $\mc{N}$. In Section~\ref{sec:ent-mes-sim}, we derive upper bounds on the entanglement generation and secret key agreement capacities of bidirectional PPT- and teleportation-simulable channels, respectively. Our upper bounds on the capacities of such channels depend only on the entanglement of the resource states with which these bidirectional channels can be simulated. In Section~\ref{sec:priv-read}, we introduce a protocol called private reading, whose goal is to generate a secret key between an encoder and a reader. We derive both lower and upper bounds on the private reading capacities. In Section~\ref{sec:coh-read}, we introduce a protocol whose goal is to generate entanglement between two parties who have coherent access to a memory cell, and we give a lower bound on the entanglement generation capacity in this setting. Finally, we  conclude in Section~\ref{sec:dis} with a summary and some open directions.

\section{Preliminaries}
\label{sec:review}

We begin by establishing some notation and reviewing  definitions needed in the rest of the paper. 

\subsection{States, channels, isometries, separable states, and positive partial transpose}

Let $\mc{B}(\mc{H})$ denote the 
algebra of bounded linear operators acting on a Hilbert space $\mc{H}$. Throughout this paper, we restrict our development to finite-dimensional Hilbert spaces. The
subset of $\mc{B}(\mc{H})$ 
containing all positive semi-definite operators is denoted by $\mc{B}_+(\mc{H})$. We denote the identity operator as $I$ and the identity superoperator as $\id$. The Hilbert space 
of a quantum system $A$ is denoted by $\mc{H}_A$.
The state of a quantum system $A$ is represented by a density operator $\rho_A$, which is a positive semi-definite operator with unit trace.
Let $\mc{D}(\mc{H}_A)$ denote the set of density operators, i.e., all elements $\rho_A\in \mc{B}_+(\mc{H}_A)$ such that $\Tr\{\rho_A\}=1$. The Hilbert space for a composite system $LA$ is denoted as $\mc{H}_{LA}$ where $\mc{H}_{LA}=\mc{H}_L\otimes\mc{H}_A$. The density operator of a composite system $LA$ is defined as $\rho_{LA}\in \mc{D}(\mc{H}_{LA})$, and the partial trace over $A$ gives the reduced density operator for system $L$, i.e., $\Tr_A\{\rho_{LA}\}=\rho_L$ such that $\rho_L\in \mc{D}(\mc{H}_L)$. The notation $A^n:= A_1A_2\cdots A_n$ indicates a composite system consisting of $n$ subsystems, each of which is isomorphic to the Hilbert space $\mc{H}_A$. A pure state $\psi_A$ of a system $A$ is a rank-one density operator, and we write it as $\psi_A=|\psi\>\<\psi|_A$
for $|\psi\>_A$ a unit vector in $ \mc{H}_A$. A purification of a density operator $\rho_A$ is a pure state $\psi^\rho_{EA}$
such that $\Tr_E\{\psi^\rho_{EA}\}=\rho_A$, where $E$ is called the purifying system.
The maximally mixed state is denoted by
$\pi_A := I_A / \dim(\mathcal{H}_A) \in\mc{D}\(\mc{H}_A\)$. The fidelity of $\tau,\sigma\in\mc{B}_+(\mc{H})$ is defined as $F(\tau,\sigma)=\norm{\sqrt{\tau}\sqrt{\sigma}}_1^2$ \cite{U76},  with the trace norm $\norm{X}_1=\Tr\sqrt{X^\dagger X}$ for $X\in\mc{B}(\mc{H})$.

The adjoint $\mc{M}^\dagger:\mc{B}(\mc{H}_B)\to\mc{B}(\mc{H}_A)$ of a linear map $\mc{M}:\mc{B}(\mc{H}_A)\to\mc{B}(\mc{H}_B)$ is the unique linear map such that
	\begin{equation}
	\label{eq-adjoint}
 \<Y_B,\mc{M}(X_A)\> =\<\mc{M}^\dag(Y_B),X_A\>,
	\end{equation}
	for all $X_A\in\mc{B}(\mc{H}_A)$ and $Y_B\in\mc{B}(\mc{H}_B)$,
	where $\<C,D\>=\Tr\{C^\dag D\}$ is the Hilbert-Schmidt inner product. An isometry $U:\mc{H}\to\mc{H}'$ is a
linear map such that $U^{\dag}U=I_{\mathcal{H}}$. 

The evolution of a quantum state is described by a quantum channel. A quantum channel $\mc{M}_{A\to B}$ is a completely positive, trace-preserving (CPTP) map $\mc{M}:\mc{B}_+(\mc{H}_A)\to \mc{B}_+(\mc{H}_B)$. A memory cell $\{\mc{M}^x\}_{x\in\mc{X}}$ is defined as a set of quantum channels $\mc{M}^x$, for all $x\in\mc{X}$, where $\mc{X}$ is a finite alphabet, and $\mc{M}^{x}:\mc{B}_+(\mc{H}_{A})\to\mc{B}_+(\mc{H}_{B})$.

Let $U^\mc{M}_{A\to BE}$ denote an isometric extension of a quantum channel $\mc{M}_{A\to B}$, which by definition means that for all $\rho_A\in \mc{D}\(\mc{H}_A\)$,
\begin{equation}
\Tr_E\left\{U^\mc{M}_{A\to BE}\rho_A\left(U^\mc{M}_{A\to BE}\right)^\dagger\right\}=\mathcal{M}_{A\to B}(\rho_A) ,
\end{equation}
along with the following conditions
for $U^\mc{M}$ to be 
an isometry: 
\begin{equation}
(U^\mc{M})^\dagger U^\mc{M}=I_{A}.
\label{eq:isometry-condition}
\end{equation}
As a consequence of \eqref{eq:isometry-condition}, we conclude that $U^\mc{M}(U^\mc{M})^\dagger=\Pi_{BE}$, where $\Pi_{BE}$ is a projection onto a subspace of the Hilbert space $\mc{H}_{BE}$. A complementary channel $\widehat{\mc{M}}_{A\to E}$ of $\mc{M}_{A\to B}$ is defined as
\begin{equation}
\widehat{\mc{M}}_{A\to E}(\rho_A):=\Tr_{B}\left\{U^\mc{M}_{A\to BE}\rho_A(U^\mc{M}_{A\to BE})^\dag\right\},
\end{equation}
for all $\rho_A\in \mc{D}\(\mc{H}_A\)$.

The Choi isomorphism represents a well known duality between channels and states. Let $\mc{M}_{A\to B}$ be a quantum channel, and let $\left|\Upsilon\right>_{L:A}$ denote the following maximally entangled vector:
\begin{equation}
|\Upsilon\>_{L:A}\coloneqq \sum_{i}|i\>_L|i\>_A ,
\end{equation}
where $\dim(\mc{H}_L)=\dim(\mc{H}_A)$, and $\{|i\>_L\}_i$ and $\{|i\>_A\}_i$ are fixed orthonormal bases. We extend this notation to multiple parties with a given bipartite cut as
\begin{equation}
|\Upsilon\>_{L_AL_B:AB}\coloneqq |\Upsilon\>_{L_A:A}\otimes |\Upsilon\>_{L_B:B}.
\end{equation}
The maximally entangled state $\Phi_{LA}$ is denoted as
\begin{equation}
\Phi_{LA}:=\frac{1}{|A|}\ket{\Upsilon}\!\bra{\Upsilon}_{LA},
\end{equation}
where $|A|=\dim(\mc{H}_A)$.
 The Choi operator for a channel $\mc{M}_{A\to B}$ is defined as
\begin{equation}
J^\mc{M}_{LB}:=(\id_L\otimes\mc{M}_{A\to B})\(|\Upsilon\>\<\Upsilon|_{LA}\),
\end{equation}
where $\id_L$ denotes the identity map on $L$. For $A'\simeq A$, the following identity holds
\begin{equation}\label{eq:choi-sim}
\<\Upsilon|_{A':L}(\rho_{SA'}\otimes J^\mc{M}_{LB})|\Upsilon\>_{A':L}=\mc{M}_{A\to B}(\rho_{SA}),
\end{equation}
where $A'\simeq A$. The above identity can be understood in terms of a post-selected variant \cite{HM04} of the quantum teleportation protocol \cite{BBC+93}. Another identity that holds is
\begin{equation}\label{eq:12}
\<\Upsilon|_{L:A} [Q_{SL}\otimes I_A]
|\Upsilon\>_{L:A}=\Tr_L\{Q_{SL}\},
\end{equation}
for an operator $Q_{SL}\in \mc{B}(\mc{H}_S\otimes\mc{H}_L)$. 

For a fixed basis $\{|i\>_B\}_i$, the partial transpose $\T_B$ on system $B$ is the following map:
\begin{multline}
\(\id_A\otimes \T_B\)(Q_{AB})\\ =\sum_{i,j}\(I_A\otimes |i\>\<j|_B\) Q_{AB}\( I_A\otimes |i\>\<j|_B\),\, \label{eq:PT-1}
\end{multline}
where $Q_{AB}\in\mc{B}(\mc{H}_{A}\otimes\mc{H}_{B})$. 

Furthermore, it holds that 
\begin{equation}\label{eq:Ttrick}
\(Q_{SL}\otimes I_A\)|\Upsilon\>_{L:A}=\(\T_A\(Q_{SA}\)\otimes I_L\)|\Upsilon\>_{L:A}.
\end{equation}

We note that the partial transpose is self-adjoint, i.e., $\T_B=\T^\dag_B$ and is also involutory:
\begin{equation}
\T_B\circ\T_B=I_B.
\end{equation} 
The following identity also holds
\begin{equation}
\T_{L}(\ket{\Upsilon}\!\bra{\Upsilon}_{LA})=\T_{A}(\ket{\Upsilon}\!\bra{\Upsilon}_{LA}).
\label{eq:PT-last}
\end{equation} 

Let $\SEP(A\!:\!B)$ denote the set of all separable states $\sigma_{AB}\in\mc{D}(\mc{H}_A\otimes\mc{H}_B)$, which are states that can be written as
\begin{equation}
\sigma_{AB}=\sum_{x}p(x)\omega^x_A\otimes\tau^x_B,
\end{equation}
where $p(x)$ is a probability distribution, $\omega^x_A \in \mc{D}(\mc{H}_A)$, and $\tau^x_B\in\mc{D}(\mc{H}_B)$ for all $x$. This set
is closed under the action of the partial transpose maps $\T_A$ and $\T_B$ \cite{HHH96,Per96}. Generalizing the set of separable states, we define the set $\PPT (A\!:\!B)$ of all bipartite states $\rho_{AB}$ that remain positive after the action of the partial transpose $\T_B$. A state $\rho_{AB}\in\PPT(A\!:\!B)$ is also called a PPT (positive under partial transpose) state. We can define an even more general set of positive semi-definite operators \cite{AdMVW02} as follows:
\begin{equation}
\PPT'(A\!:\!B)\coloneqq \{\sigma_{AB}:\ \sigma_{AB}\geq 0\land \norm{\T_B(\sigma_{AB})}_1\leq 1\}. 
\end{equation} 
We then have the containments $\SEP\subset \PPT\subset \PPT' $. A bipartite quantum channel $\mc{P}_{A'B'\to AB}$ is a completely PPT-preserving channel if the map $\T_{B}\circ\mc{P}_{A'B'\to AB}\circ\T_{B'}$ is a quantum channel \cite{Rai99,Rai01} (see also \cite{CVGG17}). A bipartite quantum channel $\mc{P}_{A'B'\to AB}$ is completely PPT-preserving if and only if its Choi state is a PPT state \cite{Rai01}, i.e.,
\begin{equation}
\frac{J^{\mc{P}}_{L_AL_B:AB}}{ |L_A L_B|}\in \PPT(L_A A\!:\!BL_B),
\end{equation}
where
\begin{equation}
\frac{J^{\mc{P}}_{L_AL_B:AB}}{ |L_A L_B|} =  \mc{P}_{A'B'\to AB}(\Phi_{L_AA'}\otimes\Phi_{B'L_B}).
\end{equation}
Any local operations and classical communication (LOCC) channel is a completely PPT-preserving channel \cite{Rai99,Rai01}. For a formal definition of LOCC channels, see \cite{chitambar2014everything}.

\subsection{Channels with symmetry}\label{sec:symmetry}

Consider a finite group $G$. For every $g\in G$, let $g\to U_A(g)$ and $g\to V_B(g)$ be projective unitary representations of $g$ acting on the input space $\mc{H}_A$ and the output space $\mc{H}_B$ of a quantum channel $\mc{M}_{A\to B}$, respectively. A quantum channel $\mc{M}_{A\to B}$ is covariant with respect to these representations if the following relation is satisfied \cite{Hol02,H13book}:
\begin{equation}\label{eq:cov-condition}
\mc{M}_{A\to B}\!\(U_A(g)\rho_A U_A^\dagger(g)\)=V_B(g)\mc{M}_{A\to B}(\rho_A)V_B^\dagger(g),
\end{equation}
for all $ \rho_A\in\mc{D}(\mc{H}_A)$ and $ g\in G$.
\begin{definition}[Covariant channel \cite{H13book}]\label{def:covariant}
A quantum channel is covariant if it is covariant with respect to a group $G$ which has a representation $U(g)$, for all $g\in G$, on $\mc{H}_A$ that is a unitary one-design; i.e., the map  $\frac{1}{|G|}\sum_{g\in G}U(g)(\cdot)U^\dagger(g)$ always outputs the maximally mixed state for all input states. 
\end{definition}

For an isometric channel $\mc{U}^\mc{M}_{A\to BE}$ extending the above channel $\mc{M}_{A\to B}$, there exists a unitary representation $W_E(g)$ acting on the environment Hilbert space $\mc{H}_E$ \cite{H13book}, such that
for all $g\in G$,
\begin{multline}\label{eq:iso-covariant}
\mc{U}^\mc{M}_{A\to BE}\!\({U_A(g)\rho_AU^\dagger_A(g)}\)
= \\ \(V_B(g)\otimes W_E(g)\)\(\mc{U}^\mc{M}_{A\to BE}\(\rho_A\)\)\(V^\dagger_B(g)\otimes W^\dagger_E(g)\).
\end{multline}
We restate this as the following lemma:
\begin{lemma}[\cite{H13book}]\label{thm:cov-hol}
Suppose that a channel $\mathcal{M}_{A\rightarrow B}$ is covariant  with respect to a group $G$. 
For an isometric extension $U_{A\rightarrow BE}^{\mathcal{M}}$ of
$\mathcal{M}_{A\rightarrow B}$, there is a set of unitaries $\{W_{E}^{g}\}_{g\in G}$ such
that the following covariance holds for all $g \in G$:
\begin{equation}
U_{A\rightarrow BE}^{\mathcal{M}}U_{A}^{g}=\left(  V_{B}^{g}\otimes W_{E}%
^{g}\right)  U_{A\rightarrow BE}^{\mathcal{M}}.
\end{equation}
\end{lemma}
For convenience, we provide a proof of this interesting lemma in Appendix~\ref{app:cov-lemma}.

\begin{definition}[Teleportation-simulable \cite{BDSW96,HHH99}]\label{def:tel-sim}
A channel $\mc{M}_{A\to B}$ is teleportation-simulable with associated resource state $\omega_{L_AB}$ if there exists an LOCC channel $\mc{L}_{L_AAB\to B}$, such that for all input states $\rho_{A}\in\mc{D}\(\mc{H}_{A}\)$, the following equality holds
\begin{equation}
\mc{M}_{A\to B}\(\rho_A\)=\mc{L}_{L_AA B\to B}\(\rho_{A}\otimes\omega_{L_AB}\).
\label{eq:TP-simul}
\end{equation}
(A particular example of an LOCC channel is  a generalized teleportation protocol \cite{Wer01}).
\end{definition}

One can find the defining equation \eqref{eq:TP-simul} explicitly stated as \cite[Eq.~(11)]{HHH99}.
 All covariant channels, as given in  Definition~\ref{def:covariant}, are teleportation-simulable with respect to the resource state $\mathcal{M}_{A\to B}(\Phi_{L_AA})$~\cite{CDP09}.

\begin{definition}[PPT-simulable \cite{KW17}]
A channel $\mc{M}_{A\to B}$ is PPT-simulable with associated resource state $\omega_{L_AB}$ if there exists a completely PPT-preserving channel  $\mc{P}_{L_AAB\to B}$ (acting on systems $L_AA:B$ and where the transposition map is with respect to the system $B$) such that for all input states $\rho_{A}\in\mc{D}\(\mc{H}_{A}\)$, the following equality holds
\begin{equation}
\mc{M}_{A\to B}\(\rho_A\)=\mc{P}_{L_AA B\to B}\(\rho_{A}\otimes\omega_{L_AB}\).
\end{equation}
\end{definition}

\begin{definition}[Jointly covariant memory cell \cite{DW17}]\label{def:cov-cell}
A set $\overline{\mc{M}}_{\mc{X}}=\{\mc{M}^x_{A\to B}\}_{x\in\mc{X}}$ of quantum channels is  jointly covariant if there exists a group $G$ such that for all $x\in\mc{X}$, the channel $\mc{M}^x$ is a covariant channel with respect to the group $G$ (cf., Definition~\ref{def:covariant}).
\end{definition}

\begin{remark}[\cite{DW17}]
Any jointly covariant memory cell $\overline{\mc{M}}_{\mc{X}}=\{\mc{M}^x_{A\to B}\}_{x}$ is jointly teleportation-simulable with respect to the set $\{\mc{M}^x_{A\to B}(\Phi_{L_AA})\}_{x}$ of resource states.
\end{remark}

\subsection{Bipartite interactions and controlled channels}\label{sec:rev-control-channels}

Let us consider a bipartite quantum interaction between systems $X'$ and $B'$, generated by a Hamiltonian $\hat{H}_{X'B'E'}$, where $E'$ is a bath system. Suppose that the Hamiltonian is time independent, having the following form:
\begin{equation}\label{eq:c-ham}
\hat{H}_{X'B'E'}\coloneqq \sum_{x\in\mc{X}}\ket{x}\!\bra{x}_{X'}\otimes \hat{H}^x_{B'E'},
\end{equation}
where $\{\vert x \rangle\}_{x\in\mc{X}}$ is an orthonormal basis for the Hilbert space of system $X'$ and $\hat{H}^x_{B'E'}$ is a Hamiltonian for the composite system $B'E'$. Then, the evolution of the composite system $X'B'E'$  is given by the following controlled unitary:
\begin{equation}
U_{\hat{H}}(t)\coloneqq\sum_{x\in\mc{X}}\ket{x}\!\bra{x}_{X'}\otimes \exp\!\left(-\frac{\iota}{\hslash}\hat{H}^x_{B'E'}t\right),
\end{equation}
where $t$ denotes time. Suppose that the systems $B'$ and $E'$ are not correlated before the action of Hamiltonian $\hat{H}^x_{B'E'}$ for each $x\in\mc{X}$. Then, the evolution of the system $B'$ under the interaction $\hat{H}^x_{B'E'}$ is given by a quantum channel $\mc{M}^x_{B'\to B}$ for all $x$.

For some distributed quantum computing and information processing tasks where the controlling system $X$ and input system $B'$ are jointly accessible, the following bidirectional channel is relevant:
\begin{equation}\label{eq:bi-ch-mc-1}
\mc{N}_{X'B'\to XB}(\cdot)\coloneqq \sum_{x\in\mc{X}}\ket{x}\!\bra{x}_X\otimes\mc{M}^x_{B'\to B}\(\bra{x}(\cdot)\ket{x}_{X'}\).
\end{equation}
In the above, $X'$ is a controlling system that determines which evolution from the set $\{\mc{M}^x\}_{x\in\mc{X}}$ takes place on input system $B'$. 
In particular, when $X'$ and $B'$ are spatially separated and the input states for the system $X'B'$ are considered to be in product state, the noisy evolution for such constrained interactions is given by the following bidirectional channel:
\begin{multline}\label{eq:bi-ch-mc-2}
\mc{N}_{X'B'\to XB}(\sigma_{X'}\otimes\rho_{B'})\\
 \coloneqq \sum_{x\in\mc{X}}\bra{x}\sigma_{X'}\ket{x}_{X'}\ket{x}\!\bra{x}_X\otimes\mc{M}^x_{B'\to B}(\rho_{B'}).
\end{multline}
This kind of bipartite interaction is in one-to-one correspondence with the notion of a memory cell from the context of quantum reading \cite{BRV00,Pir11}. There, a memory cell is a collection $\{\mc{M}^x_{B'\to B}\}_x$ of quantum channels. One party chooses which channel is applied to another party's input system $B'$ by selecting a classical letter $x$. Clearly, the description in
\eqref{eq:bi-ch-mc-1} is a fully quantum description of this process, and thus we see that quantum reading can be understood as the use of a particular kind of bipartite interaction. 

\subsection{Entropies and information}

The quantum entropy of a density operator $\rho_A$ is defined as \cite{Neu32}
\begin{equation}
S(A)_\rho:= S(\rho_A)= -\Tr[\rho_A\log_2\rho_A].
\end{equation}
The conditional quantum entropy $S(A\vert B)_\rho$ of a density operator $\rho_{AB}$ of a composite system $AB$ is defined as
\begin{equation}
S(A\vert B)_\rho \coloneqq S(AB)_\rho-S(B)_\rho.
\end{equation}
The coherent information $I(A\> B)_{\rho}$ of a density operator $\rho_{AB}$  of a composite system $AB$ is defined as \cite{SN96}
\begin{equation}\label{eq:coh-info}
I(A\rangle B)_{\rho} \coloneqq - S(A\vert B)_\rho = S(B)_{\rho}-S(AB)_{\rho}.
\end{equation}
The quantum relative entropy of two quantum states is a measure of their distinguishability. For $\rho\in\mc{D}(\mc{H})$ and $\sigma\in\mc{B}_+(\mc{H})$, it is defined as~\cite{Ume62} 
\begin{equation}
D(\rho\V \sigma):= \left\{ 
\begin{tabular}{c c}
$\Tr\{\rho[\log_2\rho-\log_2\sigma]\}$, & $\supp(\rho)\subseteq\supp(\sigma)$\\
$+\infty$, &  otherwise.
\end{tabular} 
\right.
\end{equation}
The quantum relative entropy is non-increasing under the action of positive trace-preserving maps \cite{MR15}, which is the statement that $D(\rho\V\sigma)\geq D(\mc{M}(\rho)\V\mc{M}{(\sigma)})$ for any two density operators $\rho$ and $\sigma$ and a positive trace-preserving map $\mc{M}$ (this inequality applies to quantum channels as well \cite{Lin75}, since every completely positive map is also a positive map by definition).

The quantum mutual information $I(L;A)_\rho$ is a measure of correlations between quantum systems $L$ and $A$ in a state $\rho_{LA}$. It is defined as
\begin{align}
I(L;A)_\rho &:=\inf_{\sigma_A \in\mathcal{D}(\mathcal{H}_A)}D(\rho_{LA}\Vert\rho_L\otimes\sigma_A)\\ 
&=S(L)_\rho+S(A)_\rho-S(LA)_\rho.
\end{align}
The conditional quantum mutual information $I(L;A\vert C)_\rho$ of a tripartite density operator $\rho_{LAC}$ is defined as
\begin{align}
I(L;A\vert C)_\rho &:=S(L\vert C)_\rho+S(A\vert C)_\rho-S(LA\vert C)_\rho.
\end{align}
It is known that quantum entropy, quantum mutual information, and conditional quantum  mutual information are all non-negative quantities (see \cite{LR73,LR73b}). 

The following Alicki--Fannes--Winter (AFW) inequality gives uniform continuity bounds for conditional entropy:
\begin{lemma}[\cite{AF04,Win16}]\label{thm:AFW}
Let $\rho_{LA},\sigma_{LA}\in\mc{D}(\mc{H}_{LA})$. Suppose that $\frac{1}{2}\left\Vert \rho_{LA}-\sigma_{LA}\right\Vert_1\leq\varepsilon$, where $\varepsilon\in\[0,1\]$. Then
\begin{equation}
\left\vert S(A|L)_\rho-S(A|L)_\sigma\right\vert \leq 2\varepsilon\log_2\dim(\mc{H}_A)+g(\varepsilon),
\end{equation}
where 
\begin{equation}\label{eq:g-2}
g(\varepsilon)\coloneqq (1+\varepsilon)\log_2(1+\varepsilon)-\varepsilon\log_2\varepsilon,
\end{equation} and $\dim(\mc{H}_A)$ denotes the dimension of the Hilbert space~$\mc{H}_A$.

Suppose that system $L$ is a classical register $X$ such that $\rho_{XA}$ and $\sigma_{XA}$ are classical--quantum (cq) states of the following form:
\begin{align}
\rho_{XA}&=\sum_{x\in\mc{X}}p_X(x)|x\>\<x|_X\otimes\rho^x_A,\\
 \sigma_{XA}& =\sum_{x\in\mc{X}}q_X(x)|x\>\<x|_X\otimes\sigma^x_A,
\end{align}
where $\{|x\>_X\}_{x\in\mc{X}}$ forms an orthonormal basis and for all $ x\in\mc{X},\ \rho^x_A,\sigma^x_A\in\mc{D}(\mc{H}_A)$. Then the following inequalities hold
\begin{align}
\left\vert S(X|A)_\rho-S(X|A)_\sigma\right\vert &\leq \varepsilon\log_2\dim(\mc{H}_X)+g(\varepsilon),\\
\left\vert S(A|X)_\rho-S(A|X)_\sigma\right\vert &\leq \varepsilon\log_2\dim(\mc{H}_A)+g(\varepsilon).
\end{align}
\end{lemma}

\subsection{Generalized divergence and generalized relative entropies}

A quantity is called a generalized divergence \cite{PV10,SW12} if it satisfies the following monotonicity (data-processing) inequality for all density operators $\rho$ and $\sigma$ and quantum channels $\mc{N}$:
\begin{equation}\label{eq:gen-div-mono}
\mathbf{D}(\rho\Vert \sigma)\geq \mathbf{D}(\mathcal{N}(\rho)\Vert \mc{N}(\sigma)).
\end{equation}
As a direct consequence of the above inequality, any generalized divergence satisfies the following two properties for an isometry $U$ and a state~$\tau$ \cite{WWY14}:
\begin{align}
\mathbf{D}(\rho\Vert \sigma) & = \mathbf{D}(U\rho U^\dag\Vert U \sigma U^\dag),\label{eq:gen-div-unitary}\\
\mathbf{D}(\rho\Vert \sigma) & = \mathbf{D}(\rho \otimes \tau \Vert \sigma \otimes \tau).\label{eq:gen-div-prod}
\end{align}
One can define a generalized mutual information for a quantum state $\rho_{RA}$ as
\begin{equation}
I_{\mathbf{D}}(R;A)_\rho :=\inf_{\sigma_A\in\mc{D}(\mc{H}_A)}\mathbf{D}(\rho_{RA}\Vert \rho_R\otimes\sigma_A).
\end{equation}

The sandwiched R\'enyi relative entropy  \cite{MDSFT13,WWY14} is denoted as $\wt{D}_\alpha(\rho\V\sigma)$  and defined for
$\rho\in\mc{D}(\mc{H})$, $\sigma\in\mc{B}_+(\mc{H})$, and  $\forall \alpha\in (0,1)\cup(1,\infty)$ as
\begin{equation}\label{eq:def_sre}
\wt{D}_\alpha(\rho\V \sigma):= \frac{1}{\alpha-1}\log_2 \Tr\left\{\left(\sigma^{\frac{1-\alpha}{2\alpha}}\rho\sigma^{\frac{1-\alpha}{2\alpha}}\right)^\alpha \right\} ,
\end{equation}
but it is set to $+\infty$ for $\alpha\in(1,\infty)$ if $\supp(\rho)\nsubseteq \supp(\sigma)$.
The sandwiched R\'enyi relative entropy obeys the following ``monotonicity in $\alpha$'' inequality \cite{MDSFT13}: for  $\alpha,\beta\in(0,1)\cup(1,\infty)$,
\begin{equation}\label{eq:mono_sre}
\wt{D}_\alpha(\rho\V\sigma)\leq \wt{D}_\beta(\rho\V\sigma) \quad \text{ if }  \quad \alpha\leq \beta.
\end{equation}
The following lemma states that the sandwiched R\'enyi relative entropy $\wt{D}_\alpha(\rho\V\sigma)$ is a particular generalized divergence for certain values of $\alpha$. 
\begin{lemma}[\cite{FL13}]
\label{lem:DP-sandRenyi}
Let $\mc{N}:\mc{B}_+(\mc{H}_A)\to \mc{B}_+(\mc{H}_B)$ be a quantum channel   and let $\rho_A\in\mc{D}(\mc{H}_A)$ and $\sigma_A\in \mc{B}_+(\mc{H}_A)$. Then, for all $ \alpha\in \[1/2,1\)\cup (1,\infty)$
\begin{equation}
\wt{D}_\alpha(\rho\V\sigma)\geq \wt{D}_\alpha(\mc{N}(\rho)\V\mc{N}(\sigma)).
\end{equation}
\end{lemma}

See \cite{Wilde2018a} for an alternative proof of Lemma~\ref{lem:DP-sandRenyi}, and \cite{Bei13} for an even different proof when $\alpha > 1$.

In the limit $\alpha\to 1$, the sandwiched R\'enyi relative entropy $\wt{D}_\alpha(\rho\V\sigma)$ converges to the quantum relative entropy \cite{MDSFT13,WWY14}:
\begin{equation}\label{eq:mono_renyi}
\lim_{\alpha\to 1}\wt{D}_\alpha(\rho\V\sigma):= D_1(\rho\V\sigma)=D(\rho\V\sigma).
\end{equation}
In the limit $\alpha\to \infty$, the sandwiched R\'enyi relative entropy $\wt{D}_\alpha(\rho\V\sigma)$ converges to the max-relative entropy \cite{MDSFT13}, which is defined as \cite{D09,Dat09}
\begin{equation}\label{eq:max-rel}
D_{\max}(\rho\V\sigma)=\inf\{\lambda:\ \rho \leq 2^\lambda\sigma\},
\end{equation}
and if $\supp(\rho)\nsubseteq\supp(\sigma)$ then $D_{\max}(\rho\V\sigma)=\infty$.

Another generalized divergence is the $\varepsilon$-hypothesis-testing divergence \cite{BD10,WR12},  defined as
\begin{multline}
D^\varepsilon_h\!\(\rho\Vert\sigma\) 
\coloneqq\\ -\log_2\inf_{\Lambda}\{\Tr\{\Lambda\sigma\}:\ 0\leq\Lambda\leq I \wedge\Tr\{\Lambda\rho\}\geq 1-\varepsilon\},
\end{multline}
for $\varepsilon\in[0,1]$, $\rho\in\mc{D}(\mc{H})$, and $\sigma\in\mc{B}_+(\mc{H})$.

\subsection{Entanglement measures}
 
Let $\Ent(A;B)_\rho$ denote an entanglement measure \cite{HHHH09} that is evaluated for a bipartite state~$\rho_{AB}$. 
The basic property of an entanglement measure is that it should be an LOCC monotone \cite{HHHH09}, i.e.,  non-increasing under the action of an LOCC channel. 
Given such an entanglement measure, one can define the entanglement $\Ent(\mc{M})$ of a channel $\mc{M}_{A\to B}$ in terms of it by optimizing over all pure, bipartite states that can be input to the channel:
\begin{equation}\label{eq:ent-mes-channel}
\Ent(\mc{M})=\sup_{\psi_{LA}} \Ent(L;B)_\omega,
\end{equation}
where $\omega_{LB}=\mc{M}_{A\to B}(\psi_{LA})$. Due to the properties of an entanglement measure and the well known Schmidt decomposition theorem, it suffices to optimize over pure states $\psi_{LA}$ such that $L\simeq A$ (i.e., one does not achieve a higher value of
$\Ent(\mc{M})$
 by optimizing over mixed states with unbounded reference system $L$). In an information-theoretic setting, the entanglement $\Ent(\mc{M})$ of a channel~$\mc{M}$ characterizes the amount of entanglement that a sender $A$ and receiver $B$ can generate by using the channel if they do not share entanglement prior to its use.

Alternatively, one can consider the amortized entanglement $\Ent_A(\mc{M})$ of a channel $\mc{M}_{A\to B}$ as the following optimization~\cite{KW17} (see also \cite{LHL03,BHLS03,CM17,DDMW17,RKB+17}):
\begin{multline}\label{eq:ent-arm}
\Ent_A(\mc{M})
 \coloneqq \\ \sup_{\rho_{L_AAL_B}} \left[\Ent(L_A;BL_B)_{\tau}-\Ent(L_AA;L_B)_{\rho}\right],
\end{multline}
where $\tau_{L_ABL_B}=\mc{M}_{A\to B}(\rho_{L_AAL_B})$ and $\rho_{L_AAL_B}$ is a state. The supremum is with respect to all states $\rho_{L_AAL_B}$ and the systems $L_A,L_B$ are finite-dimensional but could be arbitrarily large. Thus, in general, $\Ent_A(\mc{M})$ need not be computable. The amortized entanglement quantifies the net amount of entanglement that can be generated by using the channel $\mc{M}_{A\to B}$, if the sender and the receiver are allowed to begin with some initial entanglement in the form of the state $\rho_{L_AAL_B}$. That is, $\Ent(L_AA;L_B)_\rho$ quantifies the entanglement of the initial state $\rho_{L_AAL_B}$, and $\Ent(L_A;BL_B)_{\tau}$ quantifies the entanglement of the final state produced after the action of the channel. 

The Rains relative entropy of a state $\rho_{AB}$ is defined as \cite{Rai01,AdMVW02}
\begin{equation}\label{eq:rains-inf-state}
R(A;B)_\rho\coloneqq \min_{\sigma_{AB}\in \PPT'(A:B)} D (\rho_{AB}\Vert \sigma_{AB}),
\end{equation}
and it is monotone non-increasing under the action of a completely PPT-preserving quantum channel $\mc{P}_{A'B'\to AB}$, i.e.,
\begin{equation}
R(A';B')_\rho\geq R(A;B)_\omega,
\end{equation}
where $\omega_{AB}=\mc{P}_{A'B'\to AB}(\rho_{A'B'})$. The sandwiched Rains relative entropy of a state $\rho_{AB}$ is defined as follows \cite{TWW17}:  
\begin{equation}\label{eq:alpha-rains-inf-state}
\widetilde{R}_{\alpha}(A;B)_\rho\coloneqq \min_{\sigma_{AB}\in \PPT'(A:B)} \widetilde{D}_{\alpha} (\rho_{AB}\Vert \sigma_{AB}).
\end{equation}
The max-Rains relative entropy of a state $\rho_{AB}$ is defined as \cite{WD16b}
\begin{equation}
R_{\max}(A;B)_\rho\coloneqq \min_{\sigma_{AB}\in \PPT'(A:B)} D_{\max} (\rho_{AB}\Vert \sigma_{AB}).
\end{equation}
The max-Rains information of a quantum channel $\mc{M}_{A\to B}$ is defined as \cite{WFD17}
\begin{equation}
R_{\max}(\mc{M})\coloneqq \max_{\phi_{SA}}R_{\max} (S;B)_\omega,
\label{eq:max-Rains-channel}
\end{equation}
where $\omega_{SB}=\mc{M}_{A\to B}(\phi_{SA})$ and $\phi_{SA}$ is a pure state, with $\dim(\mc{H}_S)=\dim(\mc{H}_A)$. The amortized max-Rains information of a channel $\mc{M}_{A\to B}$, denoted as $R_{\max,A}(\mc{M})$, is defined by replacing $\Ent$ in \eqref{eq:ent-arm} with the max-Rains relative entropy $R_{\max}$ \cite{BW17}. It was shown in \cite{BW17} that amortization does not enhance the max-Rains information of an arbitrary point-to-point channel, i.e.,
\begin{equation}
R_{\max,A}(\mc{M})=R_{\max}(\mc{M}).
\end{equation}  

Recently, in \cite[Eq.~(8)]{WD16a} (see also \cite{WFD17}), the max-Rains relative entropy of a state $\rho_{AB}$ was expressed as 
\begin{equation}\label{eq:rains-w}
R_{\max}(A;B)_{\rho}=\log_2 W(A;B)_{\rho}, 
\end{equation}
where $W(A;B)_{\rho}$ is the solution to the following semi-definite program:
\begin{align}
\textnormal{minimize}\ &\ \Tr\{C_{AB}+D_{AB}\}\nonumber\\
\textnormal{subject to}\ &\ C_{AB}, D_{AB}\geq 0,\nonumber\\
   &\ \T_{B} (C_{AB}-D_{AB})\geq \rho_{AB}. \label{eq:rains-state-sdp}
\end{align}
Similarly, in \cite[Eq.~(21)]{WFD17}, the max-Rains information of a quantum channel $\mc{M}_{A\to B}$ was expressed as 
\begin{equation}\label{eq:rains-omega}
R_{\max}(\mc{M})=\log_2 \Gamma (\mc{M}),
\end{equation}
where $\Gamma(\mc{M})$ is the solution to the following semi-definite program:
\begin{align}
\textnormal{minimize}\ &\ \norm{\Tr_B\{V_{SB}+Y_{SB}\}}_{\infty}\nonumber\\
\textnormal{subject to}\ &\ Y_{SB}, V_{SB}\geq 0,\nonumber\\
& \ \T_B(V_{SB}-Y_{SB})\geq J^\mc{M}_{SB}.\label{eq:rains-channel-sdp}
\end{align}

The sandwiched relative entropy of entanglement of a bipartite state $\rho_{AB}$ is defined as \cite{WTB16} 
\begin{equation}\label{eq:rel-ent-state}
\widetilde{E}_{\alpha}(A;B)_{\rho}\coloneqq\min_{\sigma_{AB}\in\SEP(A:B)}\widetilde{D}_{\alpha}(\rho_{AB}\Vert\sigma_{AB}).
\end{equation} 
In the limit $\alpha\to 1$, $\widetilde{E}_{\alpha}(A;B)_{\rho}$ converges to the relative entropy of entanglement \cite{VP98}, i.e.,
\begin{align}\label{eq:rel-ent-state-1}
\lim_{\alpha\to 1}\widetilde{E}_{\alpha}(A;B)_{\rho} &=E_R(A;B)_{\rho}\\
& \coloneqq \min_{\sigma_{AB}\in\SEP(A:B)}D(\rho_{AB}\Vert\sigma_{AB}).
\end{align} The max-relative entropy of entanglement \cite{D09,Dat09} is defined for a bipartite state $\rho_{AB}$ as
\begin{equation}\label{eq:Emax}
E_{\max}(A;B)_{\rho}\coloneqq \min_{\sigma_{AB}\in\SEP(A:B)}D_{\max}(\rho_{AB}\Vert\sigma_{AB}).
\end{equation}  
The max-relative entropy of entanglement $E_{\max}(\mc{M})$ of a channel $\mc{M}_{A\to B}$ is defined as in \eqref{eq:ent-mes-channel}, by replacing $\Ent$ with $E_{\max}$ \cite{CM17}. It was shown in \cite{CM17} that amortization does not increase max-relative entropy of entanglement of a channel $\mc{M}_{A\to B}$, i.e.,
\begin{equation}
E_{\max,A}(\mc{M})=E_{\max}(\mc{M}).
\end{equation}

The squashed entanglement of a state $\rho_{AB}\in\mc{D}(\mc{H}_{AB})$ is defined as \cite{CW04} (see also \cite{Tuc99,Tuc02}):
\begin{multline}
E_{\sq}(A;B)_\rho \coloneqq
\frac{1}{2}\inf_{\omega_{ABE}\in\mc{D}\(\mc{H}_{ABE}\)} \{ I(A;B|E)_\omega:
\\
 \Tr_E\{\omega_{ABE}\}=\rho_{AB} \}.\label{eq:Esq}
\end{multline}
In general, the extension system $E$ is finite-dimensional, but can be arbitrarily large. We can directly infer from the above definition that $E_{\sq}(B;A)_\rho=E_{\sq}(A;B)_\rho$ for any $\rho_{AB}\in\mc{D}(\mc{H}_{AB})$. We can similarly define the squashed entanglement $E_{\sq}(\mc{M})$ of a channel $\mc{M}_{A\to B}$ \cite{TGW14}, and it is known that amortization does not increase the squashed entanglement of a channel \cite{TGW14}:
\begin{equation}
E_{\sq,A}(\mc{M}) = E_{\sq}(\mc{M}).
\end{equation}
For an overview of the various entanglement measures used in this work, see Table~\ref{tb:e-meas}.

\begin{table*}
\begin{tabular}{ c || c| c| c| c| c c c}
  $E$   & $E(\rho_{AB})$ & $E(\mc{M}_{A\to B})$ & $E_A(\mc{M}_{A\to B})$ & $E^{2 \to 2}(\mc{N}_{A'B'\to AB})$ & $E^{2 \to 2}_A(\mc{N}_{A'B'\to AB})$ \\  \hline\hline
  $\tilde{R}_\alpha$   & Eq. \eqref{eq:alpha-rains-inf-state} &  via Eq. \eqref{eq:ent-mes-channel}&  via Eq. \eqref{eq:ent-arm} &  &   \\ \hline
  $R$   & Eq. \eqref{eq:rains-inf-state} &  via Eq. \eqref{eq:ent-mes-channel}&  via Eq. \eqref{eq:ent-arm} &   &   \\ \hline
    $R_{\max}$   & Eq. \eqref{eq:rains-w} &  Eq. \eqref{eq:max-Rains-channel} &  via Eq. \eqref{eq:ent-arm} &  Definition \ref{def:bi-max-rains} &  Eq. \eqref{eq:ent-locc-a} \\ \hline
      $\tilde{E}_\alpha$   &Eq. \eqref{eq:rel-ent-state} &  via Eq. \eqref{eq:ent-mes-channel} &  via Eq. \eqref{eq:ent-arm}&   & \\ \hline 
  $E_R$   & Eq. \eqref{eq:rel-ent-state-1} &  via Eq. \eqref{eq:ent-mes-channel} & via Eq. \eqref{eq:ent-arm} & &  \\ \hline
    $E_{\max}$   &Eq.  \eqref{eq:Emax} &  via Eq. \eqref{eq:ent-mes-channel} &  via Eq. \eqref{eq:ent-arm}  &   Definition \ref{def:bi-max-rel}  & Eq.  \eqref{eq:ent-locc-b} \\\hline
        $E_{\sq}$   & Eq. \eqref{eq:Esq} & via Eq. \eqref{eq:ent-mes-channel}  & via Eq. \eqref{eq:ent-arm} &   & \\
\end{tabular}
\caption{Overview of where one can find the definitions of various entanglement measures for states $\rho_{AB}$, point-to-point channels $\mc{M}_{A\to B}$, bidirectional channels $\mc{N}_{A'B'\to AB}$, and their amortized versions.}
\label{tb:e-meas}
\end{table*}

\subsection{Private states and privacy test}\label{sec:rev-priv-states}

Private states \cite{HHHO05,HHHO09} are an essential notion in any discussion of secret key distillation in quantum information, and we review their basics here.

A tripartite key state $\gamma_{K_AK_BE}$ contains $\log_2 K$ bits of secret key, shared between systems $K_A$ and $K_B$, such that $|K_A|=|K_B|=K$, and protected from an eavesdropper possessing system $E$, if there exists a state $\sigma_E$ and a projective measurement channel $\mc{M}(\cdot)=\sum_{i}\ket{i}\!\bra{i}(\cdot)\ket{i}\!\bra{i}$, where $\{\ket{i}\}_i$ is an orthonormal basis,  such that
\begin{multline}
\(\mc{M}_{K_A}\otimes\mc{M}_{K_B}\)(\gamma_{K_AK_BE})\\
 =\frac{1}{K}\sum_{i=0}^{K-1}\ket{i}\!\bra{i}_{K_A}\otimes\ket{i}\!\bra{i}_{K_B}\otimes\sigma_E.
\end{multline}
The systems $K_A$ and $K_B$ are maximally classically correlated, and the key value is uniformly random and independent of the system $E$. 

A bipartite private state
$\gamma_{S_AK_AK_BS_B}$
 containing $\log_2 K$ bits of secret key   has the following form:
\begin{multline}
\gamma_{S_AK_AK_BS_B} = \\
U^t_{S_AK_AK_BS_B}(\Phi_{K_AK_B}\otimes\theta_{S_AS_B})(U^t_{S_AK_AK_BS_B})^\dag,
\end{multline}
where $\Phi_{K_AK_B}$ is a maximally entangled state of Schmidt rank $K$, $U^t_{S_AK_AK_BS_B}$ is a \textquotedblleft twisting\textquotedblright unitary of the form 
\begin{equation}
U^t_{S_AK_AK_BS_B}\coloneqq\sum_{i,j=0}^{K-1}\ket{i}\!\bra{i}_{K_A}\otimes\ket{j}\!\bra{j}_{K_B}\otimes U^{ij}_{S_AS_B},
\end{equation}
with each $U^{ij}_{S_AS_B}$ a unitary, and $\theta_{S_AS_B}$ is a state. The systems $S_A,S_B$ are called \textquotedblleft shield\textquotedblright systems because they, along with the twisting unitary, can help to protect the key in systems $K_A$ and $K_B$ from any party possessing a purification of $\gamma_{S_AK_AK_BS_B}$.

Bipartite private states and tripartite key states are equivalent \cite{HHHO05,HHHO09}. That is, for $\gamma_{S_AK_AK_BS_B}$ a bipartite private state and $\gamma_{S_AK_AK_BS_BE}$ some purification of it, $\gamma_{K_AK_BE}$ is a tripartite key state. Conversely, for any tripartite key state $\gamma_{K_AK_BE}$ and any purification $\gamma_{S_AK_AK_BS_BE}$ of it, $\gamma_{S_AK_AK_BS_B}$ is a bipartite private state. 

A state $\rho_{K_AK_BE}$ is an $\varepsilon$-approximate tripartite key state if there exists a tripartite key state $\gamma_{K_AK_BE}$ such that 
\begin{equation}
F(\rho_{K_AK_BE},\gamma_{K_AK_BE})\geq 1-\varepsilon,
\end{equation}
where $\varepsilon\in[0,1]$. Similarly, a state $\rho_{S_AK_AK_BS_B}$ is an $\varepsilon$-approximate bipartite private state if there exists a bipartite private state $\gamma_{S_AK_AK_BS_B}$ such that 
\begin{equation}
F(\rho_{S_AK_AK_BS_BE},\gamma_{S_AK_AK_BS_BE})\geq 1-\varepsilon.
\end{equation}

If $\rho_{S_AK_AK_BS_B}$ is an $\varepsilon$-approximate bipartite key state with $K$ key values, then Alice and Bob hold an $\varepsilon$-approximate tripartite key state with $K$ key values, and the converse is true as well \cite{HHHO05,HHHO09}.

A privacy test corresponding to $\gamma_{S_AK_AK_BS_B}$ (a $\gamma$-privacy test) is defined as the following dichotomic measurement \cite{WTB16}:
\begin{equation}
\{\Pi^\gamma_{S_AK_AK_BS_B}, I_{S_AK_AK_BS_B}-\Pi^\gamma_{S_AK_AK_BS_B}\},
\end{equation}
where
\begin{multline}
\Pi^\gamma_{S_AK_AK_BS_B}
\coloneqq \\ U^t_{S_AK_AK_BS_B}(\Phi_{K_AK_B}\otimes I_{S_AS_B})(U^t_{S_AK_AK_BS_B})^\dag
\end{multline}
 and $U^t_{S_AK_AK_BS_B}$ is the twisting unitary discussed earlier. Let $\varepsilon\in[0,1]$ and $\rho_{S_AK_AK_BS_B}$ be an $\varepsilon$-approximate bipartite private state. The probability for $\rho_{S_AK_AK_BS_B}$ to pass the $\gamma$-privacy test is never smaller than $1-\varepsilon$ \cite{WTB16}:
\begin{equation}
\Tr\{\Pi^\gamma_{S_AK_AK_BS_B}\rho_{S_AK_AK_BS_B}\}\geq 1-\varepsilon. 
\end{equation}
For a state $\sigma_{S_AK_AK_BS_B}\in\SEP(S_AK_A\!:\!K_BS_B)$, the probability of passing any $\gamma$-privacy test is never greater than $\frac{1}{K}$ \cite{HHHO09}:
\begin{equation}
\Tr\{\Pi^\gamma_{S_AK_AK_BS_B}\sigma_{S_AK_AK_BS_B}\}\leq \frac{1}{K},
\end{equation}
where $K$ is the number of values that the secret key can take (i.e., $K=\dim(\mc{H}_{K_A})=\dim(\mc{H}_{K_B})$). These two inequalities are foundational for some of the converse bounds established in this paper, as was the case in \cite{HHHO09,WTB16}.


\section{Entanglement distillation from bipartite quantum interactions}\label{sec:ent-dist}

In this section, we define the bidirectional max-Rains information $R^{2\to 2}_{\max}(\mc{N})$ of a bidirectional channel $\mc{N}$ and show that it is not enhanced by amortization.  We also prove that $R^{2\to 2}_{\max}(\mc{N})$ is an upper bound on the amount of entanglement that can be distilled from a bidirectional channel $\mc{N}$. We do so by adapting to the bidirectional setting, the result  from \cite{KW17} discussed below and recent techniques developed in  \cite{CM17,RKB+17,BW17} for point-to-point quantum communication protocols.

Recently, it was shown in \cite{KW17}, connected to related developments in \cite{LHL03,BHLS03,CM17,DDMW17,DW17}, that the amortized entanglement of a point-to-point channel $\mc{M}_{A\to B}$ serves as an upper bound on the entanglement of the final state, say $\omega_{AB}$, generated at the end of an LOCC- or PPT-assisted quantum communication protocol that uses $\mc{M}_{A\to B}$ $n$ times:
\begin{equation}
\Ent(A;B)_{\omega}\leq n\Ent_{A}(\mc{M}).
\end{equation}
Thus, the physical question of determining meaningful upper bounds on the LOCC- or PPT-assisted capacities of point-to-point channel $\mc{M}$ is equivalent to the mathematical question of whether amortization can enhance the entanglement of a given channel, i.e., whether the following equality holds for a given entanglement measure $\Ent$:
\begin{equation}
\Ent_A(\mc{M})\stackrel{?}{=} \Ent(\mc{M}). 
\end{equation}  

\subsection{Bidirectional max-Rains information} 

The following definition generalizes the  max-Rains information from \eqref{eq:max-Rains-channel}, \eqref{eq:rains-omega}, and \eqref{eq:rains-channel-sdp} to the bidirectional setting:

\begin{definition}[Bidirectional max-Rains information]\label{def:bi-max-rains}
The bidirectional max-Rains information of a bidirectional quantum channel $\mc{N}_{A'B'\to AB}$ is defined as
\begin{equation}
R_{\max}^{2\to 2}(\mc{N})\coloneqq \log \Gamma^{2\to2} (\mc{N}), \label{eq:bi-max-Rains-info}
\end{equation}
where $\Gamma^{2\to2}(\mc{N})$ is the solution to the following semi-definite program:
\begin{align}
\textnormal{minimize}\ &\ \norm{\Tr_{AB}\{V_{S_A ABS_B}+Y_{S_A ABS_B}\}}_\infty\nonumber\\
\textnormal{subject to}\ &\ V_{S_A ABS_B},Y_{S_A ABS_B}\geq 0,\nonumber\\
&\ \T_{BS_B}(V_{S_A ABS_B}-Y_{S_A ABS_B})\geq J^\mc{N}_{S_AABS_B},
\label{eq:bi-rains-channel-sdp}
\end{align} 
such that $S_A\simeq A'$, and $S_B\simeq B'$. 
\end{definition}

\begin{remark}
By employing the Lagrange multiplier method,
the bidirectional max-Rains information of a bidirectional channel $\mc{N}_{A'B'\to AB}$ can also be expressed as
\begin{equation}
R^{2\to 2}_{\max}(\mc{N})=\log \Gamma^{2\to2}(\mc{N}),
\end{equation}
where $\Gamma^{2\to2}(\mc{N})$ is solution to the following semi-definite program (SDP):
\begin{align}
&\textnormal{maximize}\qquad \Tr\{J^{\mc{N}}_{S_AABS_B}X_{S_A ABS_B}\}\nonumber\\
&\textnormal{subject to}:\nonumber\\
 &\ X_{S_A ABS_B},\rho_{S_AS_B}\geq 0,\quad \Tr\{\rho_{S_A S_B}\}=1, \nonumber\\
&-\rho_{S_A S_B}\otimes I_{AB}\leq \T_{BS_B}( X_{S_A ABS_B})\leq \rho_{S_A S_B}\otimes I_{AB},\label{eq:bi-rains-channel-sdp-primal}
\end{align}
such that $S_A\simeq A'$, and $S_B\simeq B'$. 
Strong duality holds by employing Slater's condition \cite{Wat15} (see also \cite{WD16a}). Thus, as indicated above, the optimal values of the primal and dual semi-definite programs, i.e., \eqref{eq:bi-rains-channel-sdp-primal} and \eqref{eq:bi-rains-channel-sdp}, respectively, are equal.
\end{remark}

The following proposition constitutes one of our main technical results, and an immediate corollary of it is that the bidirectional max-Rains information of a bidirectional quantum channel is an upper bound on the amortized max-Rains information of the same channel. 

\begin{proposition}\label{prop:rains-tri-ineq}
Let $\rho_{L_A A'B'L_B}$ be a state and let $\mc{N}_{A'B'\to AB}$ be a bidirectional channel. Then 
\begin{multline}
 R_{\max}(L_A A; B L_B)_{\omega} \leq \\
 R_{\max}(L_A A';B' L_B)_{\rho}+ R^{2\to 2}_{\max}(\mc{N}),
\end{multline}
where $\omega_{L_A AB L_B}=\mc{N}_{A'B'\to AB}(\rho_{L_A A'B'L_B})$ and $R^{2\to 2}_{\max}(\mc{N})$ is the bidirectional max-Rains information of $\mc{N}_{A'B'\to AB}$. 
\end{proposition}
\begin{proof}
We adapt the proof steps of \cite[Proposition 1]{BW17} to the bidirectional setting. By removing logarithms and applying \eqref{eq:rains-w} and \eqref{eq:bi-max-Rains-info}, the desired inequality is equivalent to the following one:
\begin{equation}\label{eq:w-omega-ineq}
W(L_A A; B L_B)_{\omega}\leq W(L_A A';B' L_B)_{\rho}\cdot \Gamma^{2\to2}(\mc{N}),
\end{equation}
and so we aim to prove this one. Exploiting the identity in \eqref{eq:rains-state-sdp}, we find that 
\begin{equation}
W(L_A A'; B' L_B)_{\rho}=\min \Tr\{C_{L_AA'B'L_B}+D_{L_AA'B'L_B}\},
\end{equation}
subject to the constraints 
\begin{align}
C_{L_AA'B'L_B},D_{L_AA'B'L_B} &\geq 0,\\
\T_{B'L_{B}}(C_{L_AA'B'L_B}-D_{L_AA'B'L_B})&\geq \rho_{L_AA'B'L_B},
\end{align}
while the definition in \eqref{eq:bi-rains-channel-sdp} gives that 
\begin{equation}
\Gamma^{2\to2}(\mc{N})=\min \norm{\Tr_{AB}\{V_{S_A ABS_B}+Y_{S_A ABS_B}\}}_\infty,
\end{equation}
subject to the constraints 
\begin{align}
V_{S_A ABS_B},Y_{S_A ABS_B} &\geq 0,\\
 \T_{BS_B}(V_{S_A ABS_B}-Y_{S_A ABS_B})&\geq J^\mc{N}_{S_AABS_B}.\label{eq:choi-bi-b}
\end{align}
The identity in \eqref{eq:rains-state-sdp} implies that the left-hand side of \eqref{eq:w-omega-ineq} is equal to 
\begin{equation}
W(L_AA;BL_B)_\omega=\min \Tr\{E_{L_AABL_B}+F_{L_AABL_B}\},
\end{equation}
subject to the constraints
\begin{align}
E_{L_AABL_B},F_{L_AABL_B}&\geq 0,\label{eq:rains-sdp-ef}\\
\mc{N}_{A'B'\to AB}(\rho_{L_AA'B'L_B})&\leq \T_{BL_B}(E_{L_AABL_B}-F_{L_AABL_B})\label{eq:rains-sdp-channel-ef}.
\end{align}

Once we have these SDP formulations, we can now show that the inequality in \eqref{eq:w-omega-ineq} holds by making appropriate choices for $E_{L_AABL_B}$ and $ F_{L_AABL_B}$.  Let $C_{L_AA'B'L_B}$ and $ D_{L_AA'B'L_B}$ be optimal for $W(L_AA';B'L_B)_\rho$, and let $V_{S_AABS_B}$ and $Y_{S_AABS_B}$ be optimal for $\Gamma^{2\to2}(\mc{N})$. Let $\ket{\Upsilon}_{S_AS_B:A'B'}$ be the maximally entangled vector. Choose
\begin{align}
E_{L_AABL_B}&=\bra{\Upsilon}_{S_AS_B:A'B'}C_{L_AA'B'L_B}\otimes V_{S_AABS_B}\nonumber\\ 
&\quad +D_{L_AA'B'L_B}\otimes Y_{S_AABS_B}\ket{\Upsilon}_{S_AS_B:A'B'}\label{eq:E}\\
F_{L_AABL_B}&=\bra{\Upsilon}_{S_AS_B:A'B'}C_{L_AA'B'L_B}\otimes Y_{S_AABS_B}\nonumber\\
&\quad +D_{L_AA'B'L_B}\otimes V_{S_AABS_B}\ket{\Upsilon}_{S_AS_B:A'B'}.\label{eq:F}
\end{align}
Then, we have, $E_{L_AABL_B},F_{L_AABL_B}\geq 0$, because
\begin{equation}
C_{L_AA'B'L_B}, D_{L_AA'B'L_B}, Y_{S_AABS_B}, V_{S_AABS_B}\geq 0.
\end{equation}
Also, consider that
\begin{align}
&E_{L_AABL_B}-F_{L_AABL_B}\nonumber\\
&\ = \bra{\Upsilon}_{S_AS_B:A'B'}(C_{L_AA'B'L_B}-D_{L_AA'B'L_B})\otimes\nonumber\\&\qquad (V_{S_AABS_B}- Y_{S_AABS_B})\ket{\Upsilon}_{S_AS_B:A'B'}\nonumber\\
&\ = \Tr_{S_AA'B'S_B}\{\ket{\Upsilon}\!\bra{\Upsilon}_{S_AS_B:A'B'}(C_{L_AA'B'L_B} \nonumber\\ &\qquad -D_{L_AA'B'L_B})\otimes (V_{S_AABS_B}- Y_{S_AABS_B})\}.
\end{align}
 Then, using the abbreviations $E'\coloneqq~E_{L_AABL_B}$, $ F'\coloneqq~F_{L_AABL_B}$, $C'\coloneqq~C_{L_AA'B'L_B}$,  $D'\coloneqq~D_{L_AA'B'L_B}$, $V'\coloneqq~V_{S_AABS_B}$,  and $Y'\coloneqq~Y_{S_AABS_B}$, we have
 \begin{widetext}
\begin{align}
\T_{BL_B}(E'-F')
& = \T_{BL_B}\!\left[\Tr_{S_AA'B'S_B}\{\ket{\Upsilon}\!\bra{\Upsilon}_{S_AS_B:A'B'}(C'-D')\otimes (V'- Y')\}\right]\\
& = \T_{BL_B}\!\left[\Tr_{S_AA'B'S_B}\{\ket{\Upsilon}\!\bra{\Upsilon}_{S_AS_B:A'B'}(C'-D')\otimes(\T_{S_B}\circ\T_{S_B}) (V'- Y')\}\right]\\
& = \T_{BL_B}\!\left[\Tr_{S_AA'B'S_B}\{\T_{S_B}(\ket{\Upsilon}\!\bra{\Upsilon}_{S_AS_B:A'B'}) (C'-D')\otimes \T_{S_B} (V'- Y')\}\right]\\
& = \T_{BL_B}\!\left[\Tr_{S_AA'B'S_B}\{\ket{\Upsilon}\!\bra{\Upsilon}_{S_AS_B:A'B'} \T_{B'}(C'-D')\otimes \T_{S_B} (V'- Y')\}\right]\\
&=\Tr_{S_AA'B'S_B}\{\ket{\Upsilon}\!\bra{\Upsilon}_{S_AS_B:A'B'} \T_{B'L_B}(C'-D')\otimes \T_{BS_B} (V'- Y')\}\\
& \geq \bra{\Upsilon}_{S_AS_B:AB}\rho_{L_AA'B'L_B}\otimes J^\mc{N}_{S_AABS_B}\ket{\Upsilon}_{S_AS_B:AB}\\
& =\mc{N}_{A'B'\to AB}(\rho_{L_AA'B'L_B}).
\end{align}
\end{widetext}
In the above, we employed properties of the partial transpose reviewed in \eqref{eq:PT-1}--\eqref{eq:PT-last}. In particular, the third equality follows from the fact that $\T_{S_B}^\dagger=\T_{S_B}$. For the fourth equality we have used \eqref{eq:PT-last} to change $\T_{S_B}$ to $\T_{B'}$ and then $\T_{B'}^\dagger=\T_{B'}$. Now, consider that
\begin{align}
&\Tr\{E_{L_AABL_B}+F_{L_AABL_B}\} \nonumber\\
& = \Tr\{\bra{\Upsilon}_{S_AS_B:A'B'}(C_{L_AA'B'L_B}+D_{L_AA'B'L_B})\otimes\nonumber\\
&\qquad (V_{S_AABS_B}+ Y_{S_AABS_B})\ket{\Upsilon}_{S_AS_B:A'B'}\}\nonumber\\
& = \Tr\{(C_{L_AA'B'L_B}+D_{L_AA'B'L_B})\nonumber\\
&\qquad T_{A'B'}(V_{A'ABB'}+ Y_{A'ABB'})\}\nonumber\\
& = \Tr\{(C_{L_AA'B'L_B}+D_{L_AA'B'L_B})\nonumber\\
&\qquad T_{A'B'}(\Tr_{AB}\{V_{A'ABB'}+ Y_{A'ABB'}\})\}\nonumber\\
& \leq \Tr\{(C_{L_AA'B'L_B}+D_{L_AA'B'L_B})\}\nonumber\\
&\qquad \norm{ T_{A'B'}(\Tr_{AB}\{V_{A'ABB'}+ Y_{A'ABB'})\}}_\infty\nonumber\\
& = \Tr\{(C_{L_AA'B'L_B}+D_{L_AA'B'L_B})\}\nonumber\\
&\qquad \norm{ \Tr_{AB}\{V_{A'ABB'}+ Y_{A'ABB'}\}}_\infty\nonumber\\
&=W(L_AA';B'L_B)_\rho\cdot \Gamma^{2\to2}(\mc{N}).
\end{align}

The second equality follows from \eqref{eq:12} and \eqref{eq:Ttrick}. The inequality is a consequence of H\"{o}lder's inequality \cite{Bha97}. The  second-to-last equality follows because the spectrum of a positive semi-definite operator is invariant under the action of a full transpose (note, in this case, $\T_{A'B'}$ is the full transpose as it acts on reduced positive semi-definite operators $V_{A'B'}$ and $Y_{A'B'}$).

Therefore, we can infer that our choices of $E_{L_AABL_B}$ and $ F_{L_AABL_B}$ are feasible for $W(L_AA;BL_B)_\omega$. Since $W(L_AA;BL_B)_\omega$ involves a minimization over all  operators $E_{L_AABL_B}$ and $ F_{L_AABL_B}$ satisfying \eqref{eq:rains-sdp-ef} and \eqref{eq:rains-sdp-channel-ef}, this concludes our proof of \eqref{eq:w-omega-ineq}.
\end{proof}

\begin{remark}
The choices made for $E_{L_AABL_B}$ and $ F_{L_AABL_B}$ in \eqref{eq:E} and \eqref{eq:F}, respectively, can be thought of as bidirectional generalizations of those made in the proof of
\cite[Proposition 1]{BW17} (see also
\cite[Proposition 6]{WFD17}), and they can be understood roughly via \eqref{eq:choi-sim} as a post-selected teleportation of the optimal operators of $W(L_AA';B'L_B)_\rho$ through the optimal operators of $\Gamma^{2\to2}(\mc{N})$, with the optimal operators of $W(L_AA';B'L_B)_{\rho}$ being in correspondence with the Choi operator $J^\mc{N}_{S_AABS_B}$ through \eqref{eq:choi-bi-b}. 
\end{remark}

An immediate corollary of Proposition~\ref{prop:rains-tri-ineq} is the following:
\begin{corollary}\label{cor:rains-tri-ineq}
The amortized max-Rains information  of a bidirectional quantum channel $\mc{N}_{A'B'\to AB}$ is bounded from above by its bidirectional max-Rains information; i.e., the following inequality holds
\begin{equation}
R^{2\to 2}_{\max,A}(\mc{N})\leq R^{2\to 2}_{\max}(\mc{N}),
\label{eq:amortization-ineq-max-Rains}
\end{equation}
where $R^{2\to 2}_{\max, A} (\mc{N})$ is the amortized max-Rains information of a bidirectional channel $\mc{N}$, i.e.,
\begin{multline}\label{eq:ent-locc-a}
R^{2\to 2}_{\max, A} (\mc{N}) \coloneqq \\
\sup_{\rho_{L_AA'B'L_B}} \left[R_{\max}(L_AA;BL_B)_{\sigma}
 -R_{\max}(L_AA';B'L_B)_{\rho}\right],
\end{multline}
where $\rho_{L_AA'B'L_B}\in\mc{D}(\mc{H}_{L_AA'B'L_B})$ and $\sigma_{L_AABL_B}\coloneqq \mc{N}_{A'B'\to AB}(\rho_{L_AA'B'L_B})$.
\end{corollary}
\begin{proof}
The inequality in \eqref{eq:amortization-ineq-max-Rains} is an immediate consequence of Proposition~\ref{prop:rains-tri-ineq}. To see this, let $\rho_{L_AA'B'L_B}$ denote an arbitrary input state. Then from Proposition~\ref{prop:rains-tri-ineq} 
\begin{multline} \label{eq:r-lower-ineq}
R_{\max}(L_AA;BL_B)_\omega-R_{\max}(L_AA';B'L_B)_\rho
\\\leq R^{2\to 2}_{\max}(\mc{N}),
\end{multline}
where $\omega_{L_AABL_B}=\mc{N}_{A'B'\to AB}(\rho_{L_AA'B'L_B})$. As the inequality holds for any state $\rho_{L_AA'B'L_B}$, we conclude the inequality in \eqref{eq:amortization-ineq-max-Rains}.
\end{proof}


\subsection{Application to entanglement generation}

In this section, we discuss the implication of Proposition~\ref{prop:rains-tri-ineq} for PPT-assisted entanglement generation from a bidirectional channel. Suppose that two parties Alice and Bob are connected by a bipartite quantum interaction. Suppose that the systems that Alice and Bob hold are $A'$ and $B'$, respectively. The bipartite quantum interaction between them is represented by a bidirectional quantum channel $\mc{N}_{A'B'\to AB}$, where output systems $A$ and $B$ are in possession of Alice and Bob, respectively. This kind of protocol was considered in \cite{BHLS03} when there is LOCC assistance.  
 
\subsubsection{Protocol for PPT-assisted bidirectional entanglement generation}

\label{sec:ent-dist-protocol}

We now discuss PPT-assisted entanglement generation protocols that make use of a bidirectional quantum channel. We do so by generalizing the point-to-point communication protocol discussed in \cite{KW17} to the bidirectional setting.  

\begin{figure*}
		\centering
		\includegraphics[width=1.0\textwidth]{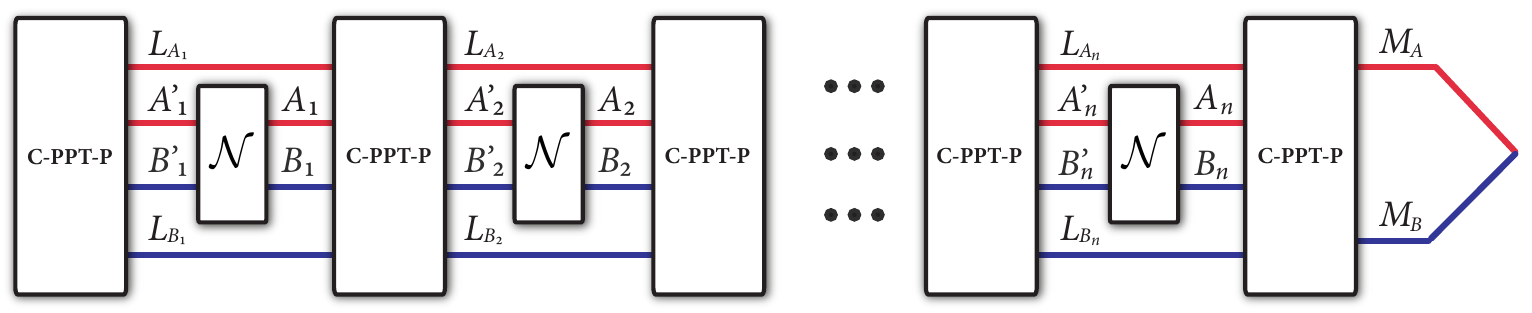}
		\caption{A protocol for PPT-assisted bidirectional quantum communication that employs $n$ uses of a bidirectional quantum channel $\mc{N}$. Every channel use is interleaved by a completely PPT-preserving channel. The goal of such a protocol is to produce an approximate maximally entangled state in the systems $M_A$ and $M_B$, where Alice possesses system $M_A$ and Bob system $M_B$.}\label{fig:bi-q-com}
	\end{figure*}

In a PPT-assisted bidirectional protocol, as depicted in Figure~\ref{fig:bi-q-com}, Alice and Bob are spatially separated and they are allowed to undergo a bipartite quantum interaction $\mc{N}_{A'B'\to AB}$, where for a fixed basis $\{|i\>_B|j\>_{L_B}\}_{i,j}$, the partial transposition $T_{BL_B}$ is considered on systems associated to Bob. Alice holds systems labeled by $A', A$ whereas Bob holds $B',B$. They begin by performing a completely PPT-preserving channel $\mc{P}^{(1)}_{\oldemptyset\to L_{A_1}A_1'B_1'L_{B_1}}$, which leads to a PPT state $\rho^{(1)}_{L_{A_1}A_1'B_1'L_{B_1}}$, where $L_{A_1},L_{B_1}$ are finite-dimensional systems of arbitrary size and $A_1',B_1'$ are input systems to the first channel use. Alice and Bob send systems $A_1'$ and $B_1'$, respectively, through the first channel use, which yields the output state
\begin{equation}
\sigma^{(1)}_{L_{A_1}A_1B_1L_{B_1}}\coloneqq\mc{N}_{A_1'B_1'\to A_1B_1}(\rho^{(1)}_{L_{A_1}A_1'B_1'L_{B_1}}).
\end{equation}
Alice and Bob then perform the completely PPT-preserving channel $\mc{P}^{(2)}_{L_{A_1}A_1B_1L_{B_1}\to L_{A_2}A_2'B_2'L_{B_2}}$, which leads to the state
\begin{equation}
\rho^{(2)}_{L_{A_2}A_2'B_2'L_{B_2}}\coloneqq \mc{P}^{(2)}_{L_{A_1}A_1B_1L_{B_1}\to L_{A_2}A_2'B_2'L_{B_2}}(\sigma^{(1)}_{L_{A_1}A_1B_1L_{B_1}}).
\end{equation}
Both parties then send systems $A_2',B_2'$ through the second channel use $\mc{N}_{A_2'B_2'\to A_2B_2}$, which yields the state
\begin{equation}
\sigma^{(2)}_{L_{A_2}A_2B_2L_{B_2}}\coloneqq \mc{N}_{A_2'B_2'\to A_2B_2}(\rho^{(2)}_{L_{A_2}A_2'B_2'L_{B_2}}).
\end{equation}
They iterate this process such that the protocol makes use of the channel $n$ times. In general, we have the following states for the $i$th use, for $i\in\{2,3,\ldots,n\}$:
\begin{align}
\rho^{(i)}_{L_{A_i}A_i'B_i'L_{B_i}} &\coloneqq \mc{P}^{(i)}(\sigma^{(i-1)}_{L_{A_{i-1}}A_{i-1}B_{i-1}L_{B_{i-1}}}),\\
\sigma^{(i)}_{L_{A_i}A_iB_iL_{B_i}} &\coloneqq \mc{N}_{A_i'B_i'\to A_iB_i}(\rho^{(i)}_{L_{A_i}A_i'B_i'L_{B_i}}),
\end{align}
where $\mc{P}^{(i)}_{L_{A_{i-1}}A_{i-1}B_{i-1}L_{B_{i-1}}\to L_{A_i}A_i'B_i'L_{B_i}}$ is a completely PPT-preserving channel, with  the partial transposition acting on systems $B_{i-1},L_{B_{i-1}}$ associated to Bob. In the final step of the protocol, a completely PPT-preserving channel $\mc{P}^{(n+1)}_{L_{A_{n}}A_{n}B_{n}L_{B_{n}}\to M_AM_B}$ is applied, which generates the final state:
\begin{equation}
\omega_{M_AM_B}\coloneqq \mc{P}^{(n+1)}_{L_{A_{n}}A_{n}B_{n}L_{B_{n}}\to M_AM_B} (\sigma^{(n)}_{L_{A_n}A_n'B_n'L_{B_n}}),
\end{equation}
where $M_A$ and $M_B$ are held by Alice and Bob, respectively. 

The goal of the protocol is for Alice and Bob to distill entanglement in the end; i.e., the final state $\omega_{M_AM_B}$ should be close to a maximally entangled state. For a fixed $n,\ M\in\mathbb{N},\ \varepsilon\in[0,1]$, the original protocol is an $(n,M,\varepsilon)$ protocol if the channel is used $n$ times as discussed above, $|M_A|=|M_B|=M$, and if 
\begin{align}
F(\omega_{M_AM_B},\Phi_{M_AM_B})&=\bra{\Phi}_{M_AM_B}\omega_{M_AM_B}\ket{\Phi}_{AB}\nonumber\\
& \geq 1-\varepsilon,
\end{align}
where $\Phi_{M_AM_B}$ is the maximally entangled state.

A rate $R$ is achievable for PPT-assisted bidirectional entanglement generation if for all $\varepsilon\in(0,1]$, $\delta>0$, and sufficiently large $n$, there exists an $(n,2^{n(R-\delta)},\varepsilon)$ protocol. The PPT-assisted bidirectional quantum capacity of a bidirectional channel $\mc{N}$, denoted as $Q^{2\to 2}_{\PPT}(\mc{N})$, is equal to the supremum of all achievable rates. Whereas, a rate $R$ is a strong converse rate for PPT-assisted bidirectional entanglement generation if for all $\varepsilon\in[0,1)$, $\delta>0$, and sufficiently large $n$, there does not exist an $(n,2^{n(R+\delta)},\varepsilon)$ protocol. The strong converse PPT-assisted bidirectional quantum capacity $\widetilde{Q}^{2\to 2}_{\PPT}(\mc{N})$ is equal to the infimum of all strong converse rates. A bidirectional channel $\mc{N}$ is said to obey the strong converse property for PPT-assisted bidirectional entanglement generation if $Q^{2\to 2}_{\PPT}(\mc{N})=\widetilde{Q}^{2\to 2}_{\PPT}(\mc{N})$. 

We note that every LOCC channel is a completely PPT-preserving channel. Given this, the well-known fact that teleportation \cite{BBC+93} is an LOCC channel, and completely PPT-preserving channels are allowed for free in the above protocol, there is no difference between an $(n,M,\varepsilon)$ entanglement generation protocol and an
$(n,M,\varepsilon)$ quantum communication protocol. Thus, all of the capacities for quantum communication are equal to those for entanglement generation.

Also, one can consider the whole development discussed above for LOCC-assisted bidirectional quantum communication instead of more general PPT-assisted bidirectional quantum communication. All the notions discussed above follow when we restrict the class of assisting completely PPT-preserving channels allowed to be LOCC channels. It follows that the LOCC-assisted bidirectional quantum capacity $Q^{2\to 2}_{\LOCC}(\mc{N})$ and the strong converse LOCC-assisted quantum capacity $\widetilde{Q}^{2\to 2}_{\LOCC}(\mc{N})$ are bounded from above as
\begin{align}
Q^{2\to 2}_{\LOCC}(\mc{N})&\leq Q^{2\to 2}_{\PPT}(\mc{N}),\\
\widetilde{Q}^{2\to 2}_{\LOCC}(\mc{N})& \leq \widetilde{Q}^{2\to 2}_{\PPT}(\mc{N}).
\end{align} 
Also, the capacities of bidirectional quantum communication protocols without any assistance are always less than or equal to the LOCC-assisted bidirectional quantum capacities.

 The following lemma is useful in deriving upper bounds on the bidirectional quantum capacities in the forthcoming sections, and it represents a generalization of the amortization idea to the bidirectional setting (see \cite{BHLS03} in this context).
 
\begin{lemma}\label{thm:ent-ppt-single-letter}
Let $\Ent_{\PPT}(A;B)_{\rho}$ be a bipartite entanglement measure for an arbitrary bipartite state $\rho_{AB}$. Suppose that $\Ent_{\PPT}(A;B)_{\rho}$ vanishes for all $\rho_{AB}\in \PPT(A\!:\!B)$ and is monotone non-increasing under completely PPT-preserving channels. Consider an $(n,M,\varepsilon)$ protocol for PPT-assisted entanglement generation over a bidirectional quantum channel $\mc{N}_{A'B'\to AB}$, as described in Section~\ref{sec:ent-dist-protocol}. Then the following bound holds
\begin{equation}
\Ent_{\PPT}(M_A;M_B)_\omega\leq n \Ent_{\PPT, A} (\mc{N}),
\end{equation}
where $\Ent_{\PPT, A} (\mc{N})$ is the amortized entanglement of a bidirectional channel $\mc{N}$, i.e.,
\begin{multline}\label{eq:ent-ppt-a}
\Ent_{\PPT, A} (\mc{N}) \coloneqq \sup_{\rho_{L_AA'B'L_B}} \left[\Ent_{\PPT}(L_AA;BL_B)_{\sigma}\right. \\
 \left. -\Ent_{\PPT}(L_AA';B'L_B)_{\rho}\right],
\end{multline}
$\rho_{L_AA'B'L_B}\in\mc{D}(\mc{H}_{L_AA'B'L_B})$, and $\sigma_{L_AABL_B}\coloneqq \mc{N}_{A'B'\to AB}(\rho_{L_AA'B'L_B})$.
\end{lemma}

\begin{proof}
From Section~\ref{sec:ent-dist-protocol}, as $\Ent$ is monotonically non-increasing under the action of completely PPT-preserving channels, we get that
\begin{align}
\Ent_{\PPT}(M_A;M_B)_\omega
& \leq \Ent_{\PPT}(L_{A_n}A_n;B_nL_{B_n})_{\sigma^{(n)}}\nonumber\\
&  = \Ent_{\PPT}(L_{A_n}A_n;B_nL_{B_n})_{\sigma^{(n)}}\nonumber\\
&\qquad -\Ent_{\PPT}(L_{A_1}A'_1;B'_1L_{B_1})_{\rho^{(1)}}\nonumber\\
& =\Ent_{\PPT}(L_{A_n}A_n;B_nL_{B_n})_{\sigma^{(n)}}\nonumber\\
&\qquad +\sum_{i=2}^n \left[\Ent_{\PPT}(L_{A_i}A'_i;B'_iL_{B_i})_{\rho^{(i)}}\right.\nonumber\\
&\quad\qquad\qquad\left. -\Ent_{\PPT}(L_{A_i}A'_i;B'_iL_{B_i})_{\rho^{(i)}}\right]\nonumber\\
&\quad\qquad - \Ent_{\PPT}(L_{A_1}A'_1;B'_1L_{B_1})_{\rho^{(1)}}\nonumber\\
& \leq \sum_{i=1}^n\left[ \Ent_{\PPT}(L_{A_i}A_i;B_iL_{B_i})_{\sigma^{(i)}}\right.\nonumber\\ 
&\quad\qquad \left. -\Ent_{\PPT}(L_{A_i}A'_i;B'_iL_{B_i})_{\rho^{(i)}}\right]\nonumber\\
& \leq n\Ent_{\PPT,A}(\mc{N}).
\end{align}
The first equality follows because $\rho^{(1)}_{L_{A_1}A_1'B_1'L_{B_1}}$ is a PPT state with vanishing $\Ent_{\PPT}$. The second equality follows trivially because we add and subtract the same terms. The second inequality follows because  $\Ent_{\PPT}(L_{A_i}A'_i;B'_iL_{B_i})_{\rho^{(i)}}\leq  \Ent_{\PPT}(L_{A_{i-1}}A_{i-1};B_{i-1}L_{B_{i-1}})_{\sigma^{(i-1)}}$ for all $i\in\{2,3,\ldots,n\}$, due to monotonicity of the entanglement measure $\Ent_{\PPT}$ with respect to completely PPT-preserving channels. The final inequality follows by applying the definition in \eqref{eq:ent-ppt-a} to each summand. 
\end{proof}

\subsubsection{Strong converse rate for PPT-assisted bidirectional entanglement generation}

We now establish the following upper bound on the bidirectional entanglement generation rate $\frac{1}{n}\log_2 M$ (qubits per channel use) of any $(n,M,\varepsilon)$ PPT-assisted protocol:
\begin{theorem}\label{thm:rains-ent-dist-strong-converse}
For a fixed $n,\ M\in\mathbb{N},\ \varepsilon\in(0,1)$, the following bound holds for an $(n,M,\varepsilon)$ protocol for PPT-assisted bidirectional entanglement generation over a bidirectional quantum channel $\mc{N}$:
\begin{equation}\label{eq:rains-ent-dist-strong-converse}
\frac{1}{n}
\log_2M\leq R^{2\to 2}_{\max}(\mc{N})+\frac{1}{n}\log_2\!\(\frac{1}{1-\varepsilon}\).
\end{equation}
\end{theorem}
\begin{proof}
From Section~\ref{sec:ent-dist-protocol}, we have that
\begin{equation}
\Tr\{\Phi_{M_AM_B}\omega_{M_AM_B}\}\geq 1-\varepsilon,
\end{equation}
while \cite[Lemma 2]{Rai99} implies that, 
for all $ \sigma_{M_AM_B}\in\PPT'(M_A:M_B)$,
\begin{equation}
 \Tr\{\Phi_{M_AM_B}\sigma_{M_AM_B}\}\leq \frac{1}{M}.
\end{equation}
Under an \textquotedblleft entanglement test\textquotedblright, which is a measurement with POVM $\{\Phi_{M_AM_B},I_{M_AM_B}-\Phi_{M_AM_B}\}$, and applying the data processing inequality for the max-relative entropy, we find that (for details, see (56)--(59) in \cite{BW17})
\begin{equation}
R_{\max}(M_A;M_B)_\omega\geq \log_2[(1-\varepsilon)M]. \label{eq:rains-test-bound}
\end{equation}
Applying Lemma~\ref{thm:ent-ppt-single-letter} and Proposition~\ref{prop:rains-tri-ineq}, we get that
\begin{equation}
R_{\max}(M_A;M_B)_\omega \leq nR^{2\to 2}_{\max}(\mc{N}).\label{eq:rains-single-letter-proof}
\end{equation}
Combining \eqref{eq:rains-test-bound} and \eqref{eq:rains-single-letter-proof}, we arrive at the desired inequality in \eqref{eq:rains-ent-dist-strong-converse}. 
\end{proof}

\begin{remark}
The bound in \eqref{eq:rains-ent-dist-strong-converse} can also be rewritten as
\begin{equation}
1-\varepsilon \leq 2^{-n[Q-R^{2\to 2}_{\max}(\mc{N})]},
\end{equation}
where we set the rate $Q=\frac{1}{n}\log_2 M$. Thus, if the bidirectional communication rate $Q$ is strictly larger than the bidirectional max-Rains information $\mc{R}^{2\to 2}_{\max}(\mc{N})$, then the fidelity of the transmission ($1-\varepsilon$) decays exponentially fast to zero in the number $n$ of channel uses. 
\end{remark}

An immediate corollary of the above remark is the following strong converse statement:
\begin{corollary}
The strong converse PPT-assisted bidirectional quantum capacity of a bidirectional channel $\mc{N}$ is bounded from above by its bidirectional max-Rains information:
\begin{equation}
\widetilde{Q}^{2\to 2}_{\PPT}(\mc{N})\leq R^{2\to 2}_{\max}(\mc{N}).
\end{equation}
\end{corollary}

\section{Secret key distillation from bipartite quantum interactions}\label{sec:priv-key}

In this section, we define the bidirectional max-relative entropy of entanglement $E^{2\to 2}_{\max}(\mc{N})$. The main goal of this section is to derive an upper bound on the rate at which secret key can be distilled from a bipartite quantum interaction. In deriving this bound, we consider private communication protocols that use a bidirectional quantum channel, and we make use of recent techniques developed in quantum information theory for point-to-point private communication protocols \cite{HHHO09,WTB16,CM17,KW17}. 

\subsection{Bidirectional max-relative entropy of entanglement}

The following definition generalizes a channel's max-relative entropy of entanglement from \cite{CM17} to the bidirectional setting:

\begin{definition}\label{def:bi-max-rel}
The bidirectional max-relative entropy of entanglement of a bidirectional channel $\mc{N}_{A'B'\to AB}$ is defined as
\begin{equation}\label{eq:bi-max-rel-opt}
E_{\max}^{2\to 2}(\mc{N})=\sup_{\psi_{S_AA'} \otimes \varphi_{B'S_B}}E_{\max}(S_A A; B S_B)_{\omega},
\end{equation} 
where
$
\omega_{S_A A B S_B}:=\mc{N}_{A'B'\to AB}(\psi_{S_AA'} \otimes \varphi_{B'S_B})
$ and 
$\psi_{S_AA'}$ and  $\varphi_{B'S_B}$ are pure bipartite states
such that $S_A\simeq A'$, and $S_B\simeq B'$.
\end{definition}

\begin{remark}\label{rem:simplify-E-2-to-2-max}
Note that we could define $E_{\max}^{2\to 2}(\mc{N})$ to have an optimization over separable input states 
$\rho_{S_AA'B'S_B}\in \SEP(S_AA'\!:\!B'S_B)$ with finite-dimensional, but arbitrarily large auxiliary systems $S_A$ and $S_B$. However, the quasi-convexity of the max-relative entropy of entanglement \cite{D09,Dat09} and the Schmidt decomposition theorem guarantee that it suffices to restrict the optimization to be as stated in Definition~\ref{def:bi-max-rel}. 
\end{remark}

\begin{proposition}\label{prop:emax-tri-ineq}
Let $\rho_{L_AA'B'L_B}$ be a state and let $\mc{N}_{A'B'\to AB}$ be a bidirectional channel. Then
\begin{multline}
 E_{\max}(L_AA;BL_B)_\omega\\
\leq  E_{\max}(L_AA';B'L_B)_\rho+E^{2\to 2}_{\max}(\mc{N}),
\end{multline}
where $\omega_{L_AABL_B}=\mc{N}_{A'B'\to AB}(\rho_{L_AA'B'L_B})$ and $E^{2\to 2}_{\max}(\mc{N})$ is the bidirectional max-relative entropy of entanglement of $\mc{N}_{A'B'\to AB}$.
\end{proposition}

\begin{proof}
Let us consider states $\sigma_{L_AA'B'L_B}'\in\SEP(L_AA'\!:\!B'L_B)$ and $\sigma_{L_AABL_B}\in\SEP(L_AA\!:\!BL_B)$, where $L_A$ and $L_B$ are finite-dimensional, but  arbitrarily large. With respect to the bipartite cut $L_AA:BL_B$, the following inequality holds
\begin{multline}
 E_{\max}(L_AA;BL_B)_\omega\\ 
 \leq D_{\max}(\mc{N}_{A'B'\to AB}(\rho_{L_AA'B'L_B})\Vert \sigma_{L_AABL_B}).
 \end{multline}
Applying the data-processed triangle inequality \cite[Theorem III.1]{CM17}, we find that
\begin{multline}
D_{\max}(\mc{N}_{A'B'\to AB}(\rho_{L_AA'B'L_B})\Vert \sigma_{L_AABL_B})\\
 \leq D_{\max}(\rho_{L_AA'B'L_B}\Vert \sigma_{L_AA'B'L_B}') \\ 
 +D_{\max}(\mc{N}_{A'B'\to AB}(\sigma_{L_AA'B'L_B}')\Vert \sigma_{L_AABL_B}).
\end{multline}
Since $\sigma_{L_AA'B'L_B}'$ and $\sigma_{L_AABL_B}$ are arbitrary separable states, we arrive at
\begin{multline}
E_{\max}(L_AA;BL_B)_\omega
 \leq  E_{\max}(L_AA';B'L_B)_\rho\\
 + E_{\max}(L_AA; BL_B)_\tau,
\end{multline}
where
\begin{align}
\omega_{L_AABL_B} & =\mc{N}_{A'B'\to AB}(\rho_{L_AA'B'L_B})\\ 
\tau_{L_AABL_B} & = \mc{N}_{A'B'\to AB}(\sigma_{L_AA'B'L_B}').
\end{align}
This implies the desired inequality after applying the observation in Remark~\ref{rem:simplify-E-2-to-2-max}, given that 
$\sigma_{L_AA'B'L_B}' \in \SEP(L_AA'\!:\!B'L_B)$.
\end{proof}
\bigskip

An immediate consequence of Proposition~\ref{prop:emax-tri-ineq} is the following corollary:
\begin{corollary}
\label{cor:amort-max-rel-ent}
Amortization does not enhance the bidirectional max-relative entropy of entanglement of a bidirectional quantum channel $\mc{N}_{A'B'\to AB}$; and the following equality holds
\begin{equation}
E^{2\to 2}_{\max,A}(\mc{N})=E^{2\to 2}_{\max}(\mc{N}),
\end{equation}  
where $E^{2\to 2}_{\max, A} (\mc{N})$ is the amortized entanglement of a bidirectional channel $\mc{N}$, i.e.,
\begin{multline}
\label{eq:ent-locc-b}
E^{2\to 2}_{\max, A} (\mc{N}) \coloneqq \sup_{\rho_{L_AA'B'L_B}} \left[E_{\max}(L_AA;BL_B)_{\sigma}\right.\\ 
 \left. -E_{\max}(L_AA';B'L_B)_{\rho}\right],
\end{multline}
where $\rho_{L_AA'B'L_B}\in\mc{D}(\mc{H}_{L_AA'B'L_B})$ and $\sigma_{L_AABL_B}\coloneqq \mc{N}_{A'B'\to AB}(\rho_{L_AA'B'L_B})$.
\end{corollary}
\begin{proof}
The inequality $E^{2\to 2}_{\max,A}(\mc{N})\geq E^{2\to 2}_{\max}(\mc{N})$ always holds. The other inequality $E^{2\to 2}_{\max,A}(\mc{N})\leq E^{2\to 2}_{\max}(\mc{N})$ is an immediate consequence of Proposition~\ref{prop:emax-tri-ineq} (the argument is similar to that given in the proof of Corollary~\ref{cor:rains-tri-ineq}).
\end{proof}

\subsection{Application to secret key agreement}
\subsubsection{Protocol for LOCC-assisted bidirectional secret key agreement}\label{sec:secret-dist-protocol}

We first introduce an LOCC-assisted secret key agreement protocol that employs a bidirectional quantum channel.

\begin{figure*}
		\centering
		\includegraphics[width=1.0\linewidth]{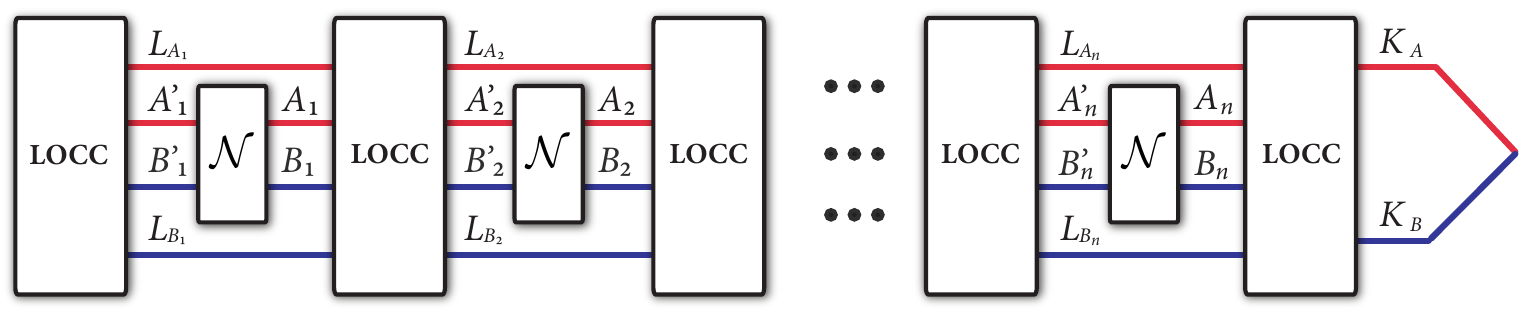}
		\caption{A protocol for LOCC-assisted bidirectional private communication that employs $n$ uses of a bidirectional quantum channel $\mc{N}$. Every channel use is interleaved by an LOCC channel. The goal of such a protocol is to produce an approximate private state in the systems $K_A$ and $K_B$, where Alice possesses system $K_A$ and Bob system $K_B$.}\label{fig:bi-p-com}
	\end{figure*}

In an LOCC-assisted bidirectional secret key agreement protocol, Alice and Bob are spatially separated and they are allowed to make use of a bipartite quantum interaction $\mc{N}_{A'B'\to AB}$, where the bipartite cut is considered between systems associated to Alice and Bob, $L_AA\!:\!L_BB$. Let $\mc{U}^\mc{N}_{A'B'\to ABE}$ be an isometric channel extending $\mc{N}_{A'B'\to AB}$:
\begin{equation}
\mc{U}^\mc{N}_{A'B'\to ABE}(\cdot)=U^{\mc{N}}_{A'B'\to ABE}(\cdot)\(U^{\mc{N}}_{A'B'\to ABE}\)^\dag,
\end{equation}
where $U^{\mc{N}}_{A'B'\to ABE}$ is an isometric extension of $\mc{N}_{A'B'\to AB}$. We assume that the eavesdropper Eve has access to the system $E$, also referred to as the environment, as well as a coherent copy of the classical communication exchanged between Alice and Bob. One could also consider a weaker assumption, in which the eavesdropper has access to only part of $E=E'E''$.

Alice and Bob begin by performing an LOCC channel $\mc{L}^{(1)}_{\oldemptyset\to L_{A_1}A_1'B_1'L_{B_1}}$, which leads to a state $\rho^{(1)}_{L_{A_1}A_1'B_1'L_{B_1}}\in\SEP(L_{A_1}A_1'\!:\!B_1'L_{B_1})$, where $L_{A_1},L_{B_1}$ are finite-dimensional systems of arbitrary size and $A_1',B_1'$ are input systems to the first channel use. Alice and Bob send systems $A_1'$ and $B_1'$, respectively, through the first channel use, that  outputs the state
\begin{equation}
\sigma^{(1)}_{L_{A_1}A_1B_1L_{B_1}}\coloneqq\mc{N}_{A_1'B_1'\to A_1B_1}(\rho^{(1)}_{L_{A_1}A_1'B_1'L_{B_1}}).
\end{equation}
They then perform the LOCC channel $\mc{L}^{(2)}_{L_{A_1}A_1B_1L_{B_1}\to L_{A_2}A_2'B_2'L_{B_2}}$, which leads to the state
\begin{equation}
\rho^{(2)}_{L_{A_2}A_2'B_2'L_{B_2}}\coloneqq \mc{L}^{(2)}_{L_{A_1}A_1B_1L_{B_1}\to L_{A_2}A_2'B_2'L_{B_2}}(\sigma^{(1)}_{L_{A_1}A_1B_1L_{B_1}}).
\end{equation}
Both parties then send systems $A_2',B_2'$ through the second channel use $\mc{N}_{A_2'B_2'\to A_2B_2}$, which yields the state $\sigma^{(2)}_{L_{A_2}A_2B_2L_{B_2}}\coloneqq \mc{N}_{A_2'B_2'\to A_2B_2}(\rho^{(2)}_{L_{A_2}A_2'B_2'L_{B_2}})$. They iterate the process such that the protocol uses the channel $n$ times. In general, we have the following states for the $i$th channel use, for $i\in\{2,3,\ldots,n\}$:
\begin{align}
\rho^{(i)}_{L_{A_i}A_i'B_i'L_{B_i}} &\coloneqq \mc{L}^{(i)}(\sigma^{(i-1)}_{L_{A_{i-1}}A_{i-1}B_{i-1}L_{B_{i-1}}}),\\
\sigma^{(i)}_{L_{A_i}A_iB_iL_{B_i}} &\coloneqq \mc{N}_{A_i'B_i'\to A_iB_i}(\rho^{(i)}_{L_{A_i}A_i'B_i'L_{B_i}}),
\end{align}
where $\mc{L}^{(i)}_{L_{A_{i-1}}A_{i-1}B_{i-1}L_{B_{i-1}}\to L_{A_i}A_i'B_i'L_{B_i}}$ is an LOCC channel corresponding to the bipartite cut $L_{A_{i-1}}A_{i-1}:B_{i-1}L_{B_{i-1}}$. In the final step of the protocol, an LOCC channel $\mc{L}^{(n+1)}_{L_{A_{n}}A_{n}B_{n}L_{B_{n}}\to K_AK_B}$ is applied, which generates the final state:
\begin{equation}
\omega_{K_AK_B}\coloneqq \mc{L}^{(n+1)}_{L_{A_{n}}A_{n}'B_{n}'L_{B_{n}}\to K_AK_B} (\sigma^{(n)}_{L_{A_n}A_n'B_n'L_{B_n}}),
\end{equation}
where the key systems $K_A$ and $K_B$ are held by Alice and Bob, respectively. 

The goal of the protocol is for Alice and Bob to distill a secret key state, such that the systems $K_A$ and $K_B$ are maximally classical correlated and tensor product with all of the systems that Eve possesses (see Section~\ref{sec:rev-priv-states} for a review of tripartite secret key states). See Figure \ref{fig:bi-p-com} for a depiction of the protocol.

\subsubsection{Purifying an LOCC-assisted bidirectional secret key agreement protocol}\label{sec:priv-dist-protocol} 

As observed in \cite{HHHO05,HHHO09} and reviewed in Section~\ref{sec:rev-priv-states}, any protocol of the above form, discussed in Section~\ref{sec:secret-dist-protocol}, can be purified in the following sense. 

The initial state $\rho^{(1)}_{L_{A_1}A_1'B_1'L_{B_1}}\in\SEP(L_{A_1}A_1'\!:\!B_1'L_{B_1})$ is of the following form:
\begin{equation}
\rho^{(1)}_{L_{A_1}A_1'B_1'L_{B_1}}\coloneqq \sum_{y_1}p_{Y_1}(y_1)\tau^{y_1}_{L_{A_1}A_1'}\otimes\varsigma^{y_1}_{L_{B_1}B_1'}.
\end{equation}
The classical random variable $Y_1$ corresponds to a message exchanged between Alice and Bob to establish this state. It can be purified in the following way:
\begin{multline}
\vert \psi^{(1)}\rangle _{Y_1S_{A_1}L_{A_1}A_1'B_1'L_{B_1}S_{B_1}}\coloneqq \\ \sum_{y_1}\sqrt{p_{Y_1}(y_1)}\ket{y_1}_{Y_1}\otimes\ket{\tau^{y_1}}_{S_{A_1}L_{A_1}A_1'}\otimes\ket{\varsigma^{y_1}}_{S_{B_1}L_{B_1}B_1'},
\end{multline}
where $S_{A_1}$ and $S_{B_1}$ are local ``shield" systems that in principle could be held by Alice and Bob, respectively, $\ket{\tau^{y_1}}_{S_{A_1}L_{A_1}A_1'}$ and $\ket{\varsigma^{y_1}}_{S_{B_1}L_{B_1}B_1'}$ purify $\tau^{y_1}_{L_{A_1}A_1'}$ and $\varsigma^{y_1}_{L_{B_1}B_1'}$, respectively, and Eve possesses system $Y_1$, which contains a coherent classical copy of the classical data exchanged between Alice and Bob. Each LOCC channel $\mc{L}^{(i)}_{L_{A_{i-1}}A_{i-1}B_{i-1}L_{B_{i-1}}\to L_{A_{i}}A_{i}'B_{i}'L_{B_{i}}}$ can be written in the following form \cite{Wat15}, for all $i\in {2,3,\ldots,n}$:
\begin{multline}\label{eq:locc}
\mc{L}^{(i)}_{L_{A_{i-1}}A_{i-1}B_{i-1}L_{B_{i-1}}\to L_{A_{i}}A_{i}'B_{i}'L_{B_{i}}}\\
\coloneqq \sum_{y_i}\mc{E}^{y_i}_{L_{A_{i-1}}A_{i-1}\to L_{A_{i}}A_{i}'}\otimes\mc{F}^{y_i}_{B_{i-1}L_{B_{i-1}}\to B_{i}'L_{B_{i}}},
\end{multline}
where $\{\mc{E}^{y_i}_{L_{A_{i-1}}A_{i-1}\to L_{A_{i}}A_{i}'}\}_{y_i}$ and $\{\mc{F}^{y_i}_{B_{i-1}L_{B_{i-1}}\to B_{i}'L_{B_{i}}}\}_{y_i}$ are collections of completely positive, trace non-increasing maps such that the map in~\eqref{eq:locc} is trace preserving. Such an LOCC channel can be purified to an isometry in the following way:
\begin{multline}\label{eq:iso-locc}
U^{\mc{L}^{(i)}}_{L_{A_{i-1}}A_{i-1}B_{i-1}L_{B_{i-1}}\to Y_iS_{A_i}L_{A_{i}}A_{i}'B_{i}'L_{B_{i}}S_{B_i}}\\
 \coloneqq \sum_{y_i}\ket{y_i}_{Y_i}\otimes U^{\mc{E}^{y_i}}_{L_{A_{i-1}}A_{i-1}\to S_{A_i}L_{A_{i}}A_{i}'}\\
  \otimes U^{\mc{F}^{y_i}}_{B_{i-1}L_{B_{i-1}}\to B_{i}'L_{B_{i}}S_{B_i}},
\end{multline}
where $\{U^{\mc{E}^{y_i}}_{L_{A_{i-1}}A_{i-1}\to S_{A_i}L_{A_{i}}A_{i}'}\}_{y_i}$ and $\{U^{\mc{F}^{y_i}}_{B_{i-1}L_{B_{i-1}}\to B_{i}'L_{B_{i}}S_{B_i}}\}_{y_i}$ are collections of linear operators (each of which is a contraction, i.e.,
\begin{multline}
\norm{U^{\mc{E}^{y_i}}_{L_{A_{i-1}}A_{i-1}\to S_{A_i}L_{A_{i}}A_{i}'}}_{\infty},\\
\norm{U^{\mc{F}^{y_i}}_{B_{i-1}L_{B_{i-1}}\to B_{i}'L_{B_{i}}S_{B_i}}}_{\infty}\leq 1
\end{multline}
 for all $y_i$) such that the linear operator $U^{\mc{L}^{(i)}}$ in \eqref{eq:iso-locc} is an isometry, the system $Y_i$ being held by Eve. The final LOCC channel can be written similarly as 
\begin{equation}
\mc{L}^{(n+1)}_{L_{A_{n}}A_{n}'B_{n}'L_{B_{n}}\to K_AK_B}\coloneqq \sum_{y_{n+1}}\mc{E}^{y_{n+1}}_{L_{A_{n}}A_{n}\to K_{A}}\otimes\mc{F}^{y_{n+1}}_{B_{n}L_{B_{n}}\to K_B},
\end{equation} 
and it can be purified to an isometry similarly as
\begin{multline}
 U^{\mc{L}^{(n+1)}}_{L_{A_{n}}A_{n}B_{n}L_{B_{n}}\to Y_{n+1} S_{A_{n+1}} K_A K_B S_{B_{n+1}}}\\ 
 \coloneqq \sum_{y_{n+1}}\ket{y_{n+1}}_{Y_{n+1}}\otimes U^{\mc{E}^{y_{n+1}}}_{L_{A_{n}}A_{n}\to S_{A_{n+1}}K_A}\otimes U^{\mc{F}^{y_{n+1}}}_{K_B S_{B_{n+1}}}.
\end{multline}
Furthermore, each channel use $\mc{N}_{A_i'B_i'\to A_iB_i}$, for all $i\in\{1,2,\ldots,n\}$, is purified by an isometry $U^\mc{N}_{A_i'B_i'\to A_iB_iE_i}$, such that Eve possesses the environment system $E_i$. 

At the end of the purified protocol, Alice possesses the key system $K_A$ and the shield systems $S_A\coloneqq S_{A_1}S_{A_2}\cdots S_{A_{n+1}}$, Bob possesses the key system $K_B$ and the shield systems $S_B\coloneqq S_{B_1}S_{B_2}\cdots S_{B_{n+1}}$, and Eve possesses the environment systems $E^n\coloneqq E_1E_2\cdots E_n$ as well as the coherent copies $Y^{n+1}\coloneqq Y_1Y_2\cdots Y_{n+1}$ of the classical data exchanged between Alice and Bob. The state at the end of the protocol is a pure state $\omega_{Y^{n+1}S_AK_AK_BS_BE^n}$.

For a fixed $n, K\in\mathbb{N},\ \varepsilon\in[0,1]$, the original protocol is an $(n,K,\varepsilon)$ protocol if the channel is used $n$ times as discussed above, $|K_A|=|K_B|=K$, and if 
\begin{align}
F(\omega_{S_AK_AK_BS_B},\gamma_{S_AK_AK_BS_B}) \geq 1-\varepsilon,
\end{align}
where $\gamma_{S_AK_AK_BS_B}$ is a bipartite private state. 

A rate $R$ is achievable for LOCC-assisted bidirectional secret key agreement if for all $\varepsilon\in(0,1]$, $\delta>0$, and sufficiently large $n$, there exists an $(n,2^{n(R-\delta)},\varepsilon)$ protocol. The LOCC-assisted bidirectional secret-key-agreement capacity of a bidirectional channel $\mc{N}$, denoted as $P^{2\to 2}_{\LOCC}(\mc{N})$, is equal to the supremum of all achievable rates. Whereas, a rate $R$ is a strong converse rate for LOCC-assisted bidirectional secret key agreement if for all $\varepsilon\in[0,1)$, $\delta>0$, and sufficiently large $n$, there does not exist an $(n,2^{n(R+\delta)},\varepsilon)$ protocol. The strong converse LOCC-assisted bidirectional secret-key-agreement capacity $\widetilde{P}^{2\to 2}_{\LOCC}(\mc{N})$ is equal to the infimum of all strong converse rates. A bidirectional channel $\mc{N}$ is said to obey the strong converse property for LOCC-assisted bidirectional secret key agreement  if $P^{2\to 2}_{\LOCC}(\mc{N})=\widetilde{P}^{2\to 2}_{\LOCC}(\mc{N})$. 

We note that the identity channel corresponding to no assistance is an LOCC channel. Therefore, one can consider the whole development discussed above for bidirectional private communication without any assistance or feedback instead of LOCC-assisted communication. All the notions discussed above follow when we exempt the employment of any non-trivial LOCC-assistance. It follows that the non-adaptive bidirectional private capacity $P^{2\to 2}_{\textnormal{n-a}}(\mc{N})$ and the strong converse non-adaptive bidirectional private capacity $\widetilde{P}_{\textnormal{n-a}}^{2\to 2}(\mc{N})$ are bounded from above as
\begin{align}
P^{2\to 2}_{\textnormal{n-a}}(\mc{N})&\leq P^{2\to 2}_{\LOCC}(\mc{N}),\\
\widetilde{P}_{\textnormal{n-a}}^{2\to 2}(\mc{N})& \leq \widetilde{P}^{2\to 2}_{\LOCC}(\mc{N}).
\end{align} 

The following lemma is useful in deriving upper bounds on the bidirectional secret-key-agreement capacity of a bidirectional channel. Its proof is very similar to the proof of Lemma~\ref{thm:ent-ppt-single-letter}, and so we omit it.
\begin{lemma}\label{thm:ent-locc-single-letter}
Let $\Ent_{\LOCC}(A;B)_{\rho}$ be a bipartite entanglement measure for an arbitrary bipartite state $\rho_{AB}$. Suppose that $\Ent_{\LOCC}(A;B)_{\rho}$ vanishes for all $\rho_{AB}\in \SEP(A\!:\!B)$ and is monotone non-increasing under LOCC channels. Consider an $(n,K,\varepsilon)$ protocol for LOCC-assisted secret key agreement over a bidirectional quantum channel $\mc{N}_{A'B'\to AB}$ as described in Section~\ref{sec:priv-dist-protocol}. Then the following bound holds
\begin{equation}
\Ent_{\LOCC}(S_A K_A;K_B S_B)_{\omega}\leq n \Ent_{\LOCC, A} (\mc{N}),
\end{equation}
where $\Ent_{\LOCC, A} (\mc{N})$ is the amortized entanglement of a bidirectional channel $\mc{N}$, i.e.,
\begin{multline}
\label{eq:ent-locc-c}
\Ent_{\LOCC, A} (\mc{N})\coloneqq \sup_{\rho_{L_AA'B'L_B}} \left[\Ent_{\LOCC}(L_AA;BL_B)_{\sigma}\right.\\
\left. -\Ent_{\LOCC}(L_AA';B'L_B)_{\rho}\right],
\end{multline}
and $\sigma_{L_AABL_B}\coloneqq \mc{N}_{A'B'\to AB}(\rho_{L_AA'B'L_B})$.
\end{lemma}

\subsubsection{Strong converse rate for LOCC-assisted bidirectional secret key agreement}
We now prove the following upper bound on the bidirectional secret key agreement rate $\frac{1}{n}\log_2 K$ (secret bits per channel use) of any $(n,K,\varepsilon)$ LOCC-assisted secret-key-agreement protocol:

\begin{theorem}\label{thm:emax-ent-dist-strong-converse}
For a fixed $n,\ K\in\mathbb{N},\ \varepsilon\in(0,1)$, the following bound holds for an $(n,K,\varepsilon)$ protocol for LOCC-assisted bidirectional secret key agreement over a bidirectional quantum channel $\mc{N}$:
\begin{equation}\label{eq:emax-ent-dist-strong-converse}
\frac{1}{n}\log_2K\leq E^{2\to 2}_{\max}(\mc{N})+
\frac{1}{n}\log_2\!\(\frac{1}{1-\varepsilon}\).
\end{equation}
\end{theorem}
\begin{proof}
From Section~\ref{sec:priv-dist-protocol}, the following inequality holds for an $(n,K,\varepsilon)$ protocol:
\begin{equation}
F(\omega_{S_AK_AK_BS_B},\gamma_{S_AK_AK_BS_B})\geq 1-\varepsilon,
\end{equation}
for some bipartite private state $\gamma_{S_AK_AK_BS_B}$ with key dimension $K$. From Section~\ref{sec:rev-priv-states}, $\omega_{S_AK_AK_BS_B}$ passes a $\gamma$-privacy test with probability  at least $1-\varepsilon$, whereas any $\tau_{S_AK_AK_BS_B}\in\SEP(S_AK_A:K_BS_B)$ does not pass with probability greater than $\frac{1}{K}$ \cite{HHHO09} (see also \cite{WTB16}). Making use of the discussion in \cite[Sections III \& IV]{CM17} (i.e., from the monotonicity of the max-relative entropy of entanglement under the $\gamma$-privacy test), we conclude that
\begin{equation}\label{eq:emax-test-bound}
\log_2 K \leq
E_{\max}(S_AK_A;K_BS_B)_{\omega}+\log_2\!\(\frac{1}{1-\varepsilon}\).
\end{equation}
Applying Lemma~\ref{thm:ent-locc-single-letter} and Corollary~\ref{cor:amort-max-rel-ent}, we get that
\begin{equation}
E_{\max}(S_AK_A;K_BS_B)_\omega \leq nE^{2\to 2}_{\max}(\mc{N}).\label{eq:emax-single-letter-proof}
\end{equation}
Combining \eqref{eq:emax-test-bound} and \eqref{eq:emax-single-letter-proof}, we get the desired inequality in \eqref{eq:emax-ent-dist-strong-converse}. 
\end{proof}

\begin{remark}
The bound in \eqref{eq:emax-ent-dist-strong-converse} can also be rewritten as
\begin{equation}
1-\varepsilon \leq 2^{-n[P-E^{2\to 2}_{\max}(\mc{N})]},
\end{equation}
where we set the rate $P=\frac{1}{n}\log_2 K$. Thus, if the bidirectional secret-key-agreement rate $P$ is strictly larger than the bidirectional max-relative entropy of entanglement $\mc{E}^{2\to 2}_{\max}(\mc{N})$, then the reliability and security  of the transmission $(1-\varepsilon)$ decays exponentially fast to zero in the number $n$ of channel uses. 
\end{remark}

An immediate corollary of the above remark is the following strong converse statement:
\begin{corollary}
The strong converse LOCC-assisted bidirectional secret-key-agreement capacity of a bidirectional channel $\mc{N}$ is bounded from above by its bidirectional max-relative entropy of entanglement:
\begin{equation}
\widetilde{P}^{2\to 2}_{\LOCC}(\mc{N})\leq E^{2\to 2}_{\max}(\mc{N}).
\end{equation}
\end{corollary}

\section{Bidirectional channels with symmetry}\label{sec:ent-mes-sim}

Channels obeying particular symmetries have played an important role in several quantum information processing tasks in the context of quantum communication protocols \cite{BDSW96,HHH99,Hol02}, quantum computing and quantum metrology \cite{DP05,JWD+08,DM14}, and resource theories \cite{Fri15,BG15}, etc. 

In this section, we define bidirectional PPT- and teleportation-simulable channels by adapting the definitions of point-to-point PPT- and LOCC-simulable channels \cite{BDSW96,HHH99,KW17} to the bidirectional setting. Then, we give upper bounds on the entanglement and secret-key-agreement capacities for communication protocols that employ bidirectional PPT- and teleportation-simulable channels, respectively. These bounds are generally tighter than those given in the previous section, because they exploit the symmetry inherent in bidirectional PPT- and teleportation-simulable channels.

\begin{definition}[Bidirectional PPT-simulable]\label{def:bi-ppt-sim}
A bidirectional channel $\mc{N}_{A'B'\to AB}$ is PPT-simulable
with associated resource state 
$\theta_{D_AD_B}\in\mc{D}\(\mc{H}_{D_A}\otimes\mc{H}_{D_B}\)$
if for all input states $\rho_{A'B'}\in\mc{D}\(\mc{H}_{A'}\otimes\mc{H}_{B'}\)$ the following equality holds
\begin{multline}
\mc{N}_{A'B'\to AB}\(\rho_{A'B'}\)\\ =\mc{P}_{D_AA'B'D_B\to AB}\(\rho_{A'B'}\otimes\theta_{D_AD_B}\),
\end{multline}
with $\mc{P}_{D_AA'B'D_B\to AB}$ being a completely PPT-preserving channel acting on $D_AA'\!:\!D_BB'$, where the partial transposition acts on the composite system $D_BB'$.
\end{definition}    

The following definition was given in 
\cite{STM11} for the special case of bipartite unitary channels:
\begin{definition}[Bidirectional teleportation-simulable]\label{def:bi-tel-sim}
A bidirectional channel $\mc{N}_{A'B'\to AB}$ is teleportation-simulable
with associated resource state 
$\theta_{D_AD_B}\in\mc{D}\(\mc{H}_{D_A}\otimes\mc{H}_{D_B}\)$
if for all input states $\rho_{A'B'}\in\mc{D}\(\mc{H}_{A'}\otimes\mc{H}_{B'}\)$ the following equality holds
\begin{multline}
\mc{N}_{A'B'\to AB}\(\rho_{A'B'}\) \\ =\mc{L}_{D_AA'B'D_B\to AB}\(\rho_{A'B'}\otimes\theta_{D_AD_B}\),
\end{multline}
where $\mc{L}_{D_AA'B'D_B\to AB}$ is an LOCC channel acting on $D_AA':D_BB'$.
\end{definition} 

Let $G$ and $H$ be finite groups, and for $g\in G$ and $h\in H$, let
$g\rightarrow U_{A^{\prime}}(g)$ and $h\rightarrow V_{B^{\prime}}(h)$ be
unitary representations. Also, let $(g,h)\rightarrow W_{A}(g,h)$ and
$(g,h)\rightarrow T_{B}(g,h)$ be unitary representations. A bidirectional
quantum channel $\mathcal{N}_{A^{\prime}B^{\prime}\rightarrow AB}$ is
\textit{bicovariant} with respect to these representations if the following relation
holds for all input density operators $\rho_{A^{\prime}B^{\prime}}$ and group
elements $g\in G$ and $h\in H$:%
\begin{multline}
\mathcal{N}_{A^{\prime}B^{\prime}\rightarrow AB}((\mathcal{U}_{A^{\prime}%
}(g)\otimes\mathcal{V}_{B^{\prime}}(h))(\rho_{A^{\prime}B^{\prime}%
}))\\ =(\mathcal{W}_{A}(g,h)\otimes\mathcal{T}_{B}(g,h))(\mathcal{N}_{A^{\prime
}B^{\prime}\rightarrow AB}(\rho_{A^{\prime}B^{\prime}})),
\end{multline}
where $\mathcal{U}(g)(\cdot)\coloneqq U(g)(\cdot)\left(
U(g)\right)^{\dag}$ denotes the unitary channel associated with a unitary operator $U(g)$, with a similar convention for the other unitary channels above.

\begin{definition}[Bicovariant channel]\label{def:bicov}
We define a bidirectional channel to be bicovariant if it is bicovariant with
respect to groups that have representations as unitary one-designs, i.e.,
$\frac{1}{\left\vert G\right\vert }\sum_{g}\mathcal{U}_{A^{\prime}}%
(g)(\rho_{A^{\prime}})=\pi_{A^{\prime}}$ and $\frac{1}{\left\vert H\right\vert
}\sum_{h}\mathcal{V}_{B^{\prime}}(h)(\rho_{B^{\prime}})=\pi_{B^{\prime}}$.
\end{definition}

An example of a bidirectional channel that is bicovariant is the
controlled-NOT (CNOT) gate \cite{BDEJ95}, for which we have the following covariances \cite{G99,GC99}:
\begin{align}
\text{CNOT}(X\otimes I)  & =(X\otimes X)\text{CNOT},\\
\text{CNOT}(Z\otimes I)  & =(Z\otimes I)\text{CNOT},\\
\text{CNOT}(Y\otimes I)  & =(Y\otimes X)\text{CNOT},\\
\text{CNOT}(I\otimes X)  & =(I\otimes X)\text{CNOT},\\
\text{CNOT}(I\otimes Z)  & =(Z\otimes Z)\text{CNOT},\\
\text{CNOT}(I\otimes Y)  & =(Z\otimes Y)\text{CNOT},
\end{align}
where $\{I,X,Y,Z\}$ is the Pauli group with the identity element $I$. A more general example of a bicovariant channel is one that applies a CNOT with some probability and,  with the complementary probability, replaces the input with the maximally mixed state.

In \cite{GC99}, the prominent idea of gate teleportation was developed, wherein one can generate the Choi state for the CNOT gate by sending in shares of maximally entangled states and then  simulate the CNOT gate's action on any input state by using teleportation through the Choi state (see also \cite{NC97} for earlier related developments). This idea generalized the notion of teleportation simulation of channels \cite{BDSW96,HHH99} from the single-sender single-receiver setting to the bidirectional setting. After these developments, \cite{CDKL01,DBB08} generalized the idea of gate teleportation to bipartite quantum channels that are not necessarily unitary channels. 

The following result slightly generalizes the developments in \cite{GC99,CDKL01,DBB08}:
\begin{proposition}\label{prop:bicov}
If a bidirectional channel $\mathcal{N}_{A^{\prime}B^{\prime}\rightarrow AB}$
is bicovariant, Definition~\ref{def:bicov}, then it is teleportation-simulable with resource state
$\theta_{L_{A}ABL_{B}}=\mathcal{N}_{A^{\prime}B^{\prime}\rightarrow AB}%
(\Phi_{L_{A}A^{\prime}}\otimes\Phi_{B^{\prime}L_{B}})$ (Definition~\ref{def:bi-tel-sim}).
\end{proposition}
We give a proof of Proposition~\ref{prop:bicov}  in Appendix~\ref{app:bicov}. 
\bigskip

We now establish an upper bound on the  entanglement generation rate of any $(n,M,\varepsilon)$ PPT-assisted protocol that employs a bidirectional PPT-simulable channel.

\begin{theorem}\label{thm:rains-ent-dist-strong-converse-ppt}
For a fixed $n,\ M\in\mathbb{N},\ \varepsilon\in(0,1)$, the following strong converse bound holds for an $(n,M,\varepsilon)$ protocol for PPT-assisted bidirectional entanglement generation  over a bidirectional PPT-simulable quantum channel $\mc{N}$ with associated resource state $\theta_{D_AD_B}$, Definition~\ref{def:bi-ppt-sim}, $\forall \alpha>1$,
\begin{multline}\label{eq:rains-ent-dist-strong-converse-ppt}
\frac{1}{n} \log_2 M
 \leq \\ \widetilde{R}_{\alpha} (D_A;D_B)_{\theta} 
+\frac{\alpha}{n(\alpha-1)}\log_2 \!\(
\frac{1}{1-\varepsilon}\),
\end{multline} 
where $\widetilde{R}_{\alpha}(D_A;D_B)_{\theta}$ is the sandwiched Rains information \eqref{eq:alpha-rains-inf-state} of the resource state $\theta_{D_A D_B}$.
\end{theorem}

\begin{proof}
The first few steps are similar to those in the proof of Theorem~\ref{thm:rains-ent-dist-strong-converse}. From Section~\ref{sec:ent-dist-protocol}, we have that
\begin{equation}
\Tr\{\Phi_{M_AM_B}\omega_{M_AM_B}\}\geq 1-\varepsilon,
\end{equation}
while \cite[Lemma 2]{Rai99} implies that, $\forall \sigma_{M_AM_B}\in\PPT'(M_A\!:\!M_B)$, 
\begin{equation}
 \Tr\{\Phi_{M_AM_B}\sigma_{M_AM_B}\}\leq \frac{1}{M}.
\end{equation}
Under an \textquotedblleft entanglement test\textquotedblright, which is a measurement with POVM $\{\Phi_{M_AM_B},I_{M_AM_B}-\Phi_{M_AM_B}\}$, and applying the data processing inequality for the sandwiched R\'enyi relative entropy, we find that  (for details, see Lemma 5 of \cite{CMW14}), for all $\alpha>1$,
\begin{equation}
\log_2 M\leq \widetilde{R}_{\alpha}(M_A;M_B)_{\omega}+\frac{\alpha}{\alpha-1}\log_2\!\(\frac{1}{1-\varepsilon}\). \label{eq:rains-test-bound-ppt}
\end{equation}
The sandwiched Rains relative entropy is monotonically non-increasing under the action of completely PPT-preserving channels and vanishing for a PPT state. Applying Lemma~\ref{thm:ent-ppt-single-letter}, we find that  
\begin{align}
&\frac{1}{n}\widetilde{R}_{\alpha}(M_A;M_B)_\omega\leq\nonumber\\
&  \sup_{\rho_{L_AA'B'L_B}} \left[ \widetilde{R}_{\alpha}(L_{A}A;BL_{B})_{\mc{N}(\rho)}-\widetilde{R}_{\alpha}(L_{A}A';B'L_{B})_{\rho}\right].\label{eq:summand-ppt}
\end{align}

As stated in Definition~\ref{def:bi-ppt-sim}, a PPT-simulable bidirectional channel $\mc{N}_{A'B'\to AB}$ with associated resource state
$\theta_{D_AD_B}$
is such that, for any input state $\rho'_{A'B'}$, 
\begin{multline}
\mc{N}_{A'B'\to AB}\(\rho'_{A'B'}\)\\ =\mc{P}_{D_AA'B'D_B\to AB}
\(\rho'_{A'B'}\otimes\theta_{D_AD_B}\).
\end{multline}
Then, for any input state $\omega'_{L_AA'B'L_B}$, 
\begin{align}
&\widetilde{R}_{\alpha}(L_AA;BL_B)_{\mc{P}(\omega'\otimes\theta)}-\widetilde{R}_{\alpha}(L_AA';B'L_B)_{\omega'}\nonumber\\ 
&\leq \widetilde{R}_{\alpha}(D_AL_AA';B'L_BD_B)_{\omega'\otimes\theta}-\widetilde{R}_{\alpha}(L_AA';B'L_B)_{\omega'} \nonumber \\
&\leq \widetilde{R}_{\alpha}(L_AA';B'L_B)_{\omega'}+\widetilde{R}_{\alpha}(D_A;D_B)_\theta \nonumber \\
&\qquad -\widetilde{R}_{\alpha}(L_AA';B'L_B)_{\omega'}\nonumber \\
&=\widetilde{R}_{\alpha}(D_A;D_B)_\theta. \label{eq:reduce-bound-ppt}
\end{align}
The first inequality follows from monotonicity of $\widetilde{R}_{\alpha}$ with respect to completely PPT-preserving channels. The second inequality follows because $\widetilde{R}_{\alpha}$ is sub-additive with respect to tensor-product states.

Applying the bound in \eqref{eq:reduce-bound-ppt} to \eqref{eq:summand-ppt}, we find that
\begin{equation}
\widetilde{R}_{\alpha}(M_A;M_B)_\omega\leq n\widetilde{R}_{\alpha}(D_A;D_B)_\theta.\label{eq:rains-single-letter-proof-ppt}
\end{equation}
 Combining \eqref{eq:rains-test-bound-ppt} and \eqref{eq:rains-single-letter-proof-ppt}, we get the desired inequality in \eqref{eq:rains-ent-dist-strong-converse-ppt}. 
\end{proof}

\bigskip
Now we establish an upper bound on the  secret key rate of an $(n,K,\varepsilon)$ secret-key-agreement protocol that employs a bidirectional teleportation-simulable channel.  

\begin{theorem}\label{thm:rel-ent-dist-strong-converse-tel}
For a fixed $n,\ K\in\mathbb{N},\ \varepsilon\in(0,1)$, the following strong converse bound holds for an $(n,K,\varepsilon)$ protocol for secret key agreement  over a bidirectional teleportation-simulable quantum channel $\mc{N}$ with associated resource state $\theta_{D_AD_B}$: $\forall \alpha>1$,
\begin{equation}\label{eq:rel-ent-dist-strong-converse-tel}
\frac{1}{n} \log_2 K
\leq \widetilde{E}_{\alpha} (D_A;D_B)_{\theta}
+\frac{\alpha}{n(\alpha-1)}\log_2\!\( \frac{1}{1-\varepsilon}\)
,
\end{equation} 
where $\widetilde{E}_{\alpha}(D_A;D_B)_{\theta}$ is the sandwiched relative entropy of entanglement \eqref{eq:rel-ent-state} of the resource state $\theta_{D_AD_B}$.
\end{theorem}

\begin{proof}
As stated in Definition~\ref{def:bi-ppt-sim}, a bidirectional teleportation-simulable  channel $\mc{N}_{A'B'\to AB}$ is such that, for any input state $\rho'_{A'B'}$,
\begin{multline}
\mc{N}_{A'B'\to AB}\(\rho'_{A'B'}\)\\ =\mc{L}_{D_AA'B'D_B\to AB}\(\rho'_{A'B'}\otimes\theta_{D_AD_B}\).
\end{multline}
Then, for any input state $\omega'_{L'_AA'B'L'_B}$, 
\begin{align}
&\widetilde{E}_{\alpha}(L'_AA;BL'_B)_{\mc{L}(\omega'\otimes\theta)}-\widetilde{E}_{\alpha}(L'_AA';B'L'_B)_{\omega'}\nonumber\\ 
&\leq \widetilde{E}_{\alpha}(D_AL'_AA';B'L'_BD_B)_{\omega'\otimes\theta}-\widetilde{E}_{\alpha}(L'_AA';B'L'_B)_{\omega'}\nonumber\\
&\leq \widetilde{E}_{\alpha}(L'_AA';B'L'_B)_{\omega'}+\widetilde{E}_{\alpha}(D_A;D_B)_\theta \nonumber\\ 
&\qquad  -\widetilde{E}_{\alpha}(L'_AA';B'L'_B)_{\omega'}\nonumber\\
&=\widetilde{E}_{\alpha}(D_A;D_B)_\theta. \label{eq:reduce-bound-tel}
\end{align}
The first inequality follows from monotonicity of $\widetilde{E}_{\alpha}$ with respect to LOCC channels. The second inequality follows because $\widetilde{E}_{\alpha}$ is sub-additive.

From Section~\ref{sec:priv-dist-protocol}, the following inequality holds for an $(n,K,\varepsilon)$ protocol:
\begin{equation}
F(\omega_{S_AK_AK_BS_B},\gamma_{S_AK_AK_BS_B})\geq 1-\varepsilon,
\end{equation}
for some bipartite private state $\gamma_{S_AK_AK_BS_B}$ with key dimension $K$. From Section~\ref{sec:rev-priv-states}, $\omega_{S_AK_AK_BS_B}$ passes a $\gamma$-privacy test with probability at least $1-\varepsilon$, whereas any $\tau_{S_AK_AK_BS_B}\in\SEP(S_AK_A:K_BS_B)$ does not pass with probability greater than $\frac{1}{K}$ \cite{HHHO09}. Making use of the results in \cite[Section 5.2]{WTB16}, we conclude that
\begin{align}\label{eq:emax-test-bound-tel}
\log_2 K \leq \widetilde{E}_{\alpha}(S_AK_A;K_BS_B)_{\omega}+\frac{\alpha}{\alpha-1}\log_2\!\(\frac{1}{1-\varepsilon}\).
\end{align}
Now we can follow steps similar to those in the proof of Theorem~\ref{thm:rains-ent-dist-strong-converse-ppt} in order to arrive at \eqref{eq:rel-ent-dist-strong-converse-tel}. 
\end{proof}

\bigskip 

We can also establish the following weak converse bounds, by combining the above approach with that in \cite[Section~3.5]{KW17}:

\begin{remark}
The following weak converse bound holds for an $(n,M,\varepsilon)$ PPT-assisted bidirectional quantum communication protocol (Section~\ref{sec:ent-dist-protocol}) that employs a bidirectional PPT-simulable quantum channel $\mc{N}$ with associated resource state $\theta_{L_AL_B}$ 
\begin{equation}\label{eq:rel-ent-dist-strong-converse-ppt-2}
(1-\varepsilon)\frac{\log_2 M}{n} \leq R(L_A;L_B)_{\theta}+\frac{1}{n}h_2(\varepsilon),
\end{equation} 
where $R(L_A;L_B)_{\theta}$ is defined in \eqref{eq:rains-inf-state} and $h_2(\varepsilon)\coloneqq -\varepsilon\log_2\varepsilon-(1-\varepsilon)\log_2(1-\varepsilon)$. 
\end{remark}

\begin{remark}\label{rem:tel-sim-priv-dist}
The following weak converse bound holds for an $(n,K,\varepsilon)$ LOCC-assisted bidirectional secret key agreement protocol (Section~\ref{sec:priv-dist-protocol}) that employs a bidirectional teleportation-simulable quantum channel $\mc{N}$ with associated resource state $\theta_{D_AD_B}$
\begin{equation}\label{eq:rel-ent-dist-strong-converse-tel-2}
(1-\varepsilon)\frac{\log_2 K}{n} \leq E(D_A;D_B)_{\theta}+\frac{1}{n}h_2(\varepsilon),
\end{equation} 
where $E(D_A;D_B)_{\theta}$ is defined in \eqref{eq:rel-ent-state-1}. 
\end{remark}

Since every LOCC channel $\mc{L}_{D_AA'B'D_B\to AB}$ acting with respect to the bipartite cut $D_AA':D_BB'$ is also a completely PPT-preserving channel with the partial transposition action on $D_BB'$, it follows that bidirectional teleportation-simulable channels are also bidirectional PPT-simulable channels. Based on Proposition~\ref{prop:bicov}, Theorem~\ref{thm:rains-ent-dist-strong-converse-ppt}, Theorem~\ref{thm:rel-ent-dist-strong-converse-tel}, and the limits $n \to \infty$ and then $\alpha\to 1$ (in this order),\footnote{One could also set $\alpha = 1+1/\sqrt{n}$ and then take the limit $n \to \infty$.} we can then  conclude the following strong converse bounds:
\begin{corollary}
\label{cor:str-conv-TP-simul}
If a bidirectional quantum channel $\mc{N}$ is bicovariant (Definition~\ref{def:bicov}), then 
\begin{align}
\widetilde{Q}^{2\to2}_{\PPT}(\mc{N})& \leq R(L_A A;BL_B)_{\theta},\\
\widetilde{P}^{2\to2}_{\LOCC}(\mc{N})&\leq E(L_AA;BL_B)_{\theta},
\end{align}
where $\theta_{L_{A}ABL_{B}}=\mathcal{N}_{A^{\prime}B^{\prime}\rightarrow AB}
(\Phi_{L_{A}A^{\prime}}\otimes\Phi_{B^{\prime}L_{B}})$, and $\widetilde{Q}^{2\to2}_{\PPT}(\mc{N})$ and $\widetilde{P}^{2\to2}_{\LOCC}(\mc{N})$ denote the strong converse PPT-assisted bidirectional quantum capacity and strong converse LOCC-assisted bidirectional secret-key-agreement capacity, respectively, of a bidirectional channel $\mc{N}$. 
\end{corollary}

\section{Private reading of a read-only memory device}\label{sec:priv-read}
Devising a communication or information processing protocol that is secure against an eavesdropper is an area of primary interest in information theory. In this section, we introduce the task of private reading of information stored in a memory device. A secret message can either be  encrypted in a computer program with circuit gates or in a physical storage device, such as a CD-ROM, DVD, etc. Here we limit ourselves to the case in which these computer programs or physical storage devices are used for read-only tasks; for simplicity, we refer to such media as memory devices. 

In \cite{BRV00}, a communication setting was considered in which a memory cell consists of unitary operations that encode a classical message. This model  was generalized and studied under the name \textquotedblleft quantum reading\textquotedblright\ in \cite{Pir11}, and it was  applied to the setting of an optical memory. In subsequent works \cite{PLG+11,LP16,DW17}, the model was extended to a memory cell consisting of arbitrary quantum channels.
In \cite{DW17}, the most  natural and general definition of the reading capacity of a memory cell was given, and this work also determined the reading capacities for some broad classes of memory cells. 
Quantum reading can be understood as a direct application of quantum channel discrimination \cite{Kit97,Fuj01,DPP01,Aci01,WY06,DFY09,HHLW10,CMW14,DGLL16}. In many cases, one can achieve performance better than what can be achieved when using a classical strategy \cite{PLG+11,GDN+11,GW12,WGTL12,LP16}. In \cite{Spe15}, the author discussed the security of a message encoded using a particular class of optical memory cells against readers employing classical strategies. 

In a reading protocol, it is assumed that the reader has a description of a memory cell, which is a set of quantum channels. The memory cell is used to encode a classical message in a memory device. The memory device containing the encoded message is then delivered to the interested reader, whose task is to read out the message stored in it. To decode the message,
the reader can transmit a quantum state to the memory device and perform a quantum measurement on the output state. In general, since quantum channels are noisy, there is a loss of information to the environment, and there is a limitation on how well information can be read out from the memory device.

To motivate the task of private reading, consider that once reading devices equipped with quantum systems are built, the readers can use these devices to transmit quantum states as a probe and then perform a joint measurement for reading the memory device. There could be a circumstance in which an individual would have to access a reading device in a public library under the surveillance of a librarian or other parties, whom we suppose to be a passive eavesdropper Eve. In such a situation, an individual would want information in a memory device not to be leaked to Eve, who has access to the environment, for security and privacy reasons. This naturally gives rise to the question of whether there exists a protocol for reading out a classical message that is secure from a passive eavesdropper.

In what follows, we introduce the details of private reading: briefly, it is the task of reading out a classical message (key) stored in a memory device, encoded with a memory cell, by the reader such that the message is not leaked to Eve. We also mention here that private reading can be understood as a particular kind of secret-key-agreement protocol that employs a particular kind of bipartite interaction, and thus, there is a strong link between the developments in Section~\ref{sec:priv-key} and what follows (we elaborate on this point in what follows).

\subsection{Private reading protocol}

\label{sec:q-r-p}

\begin{figure*}
		\centering
		\includegraphics[scale=0.75]{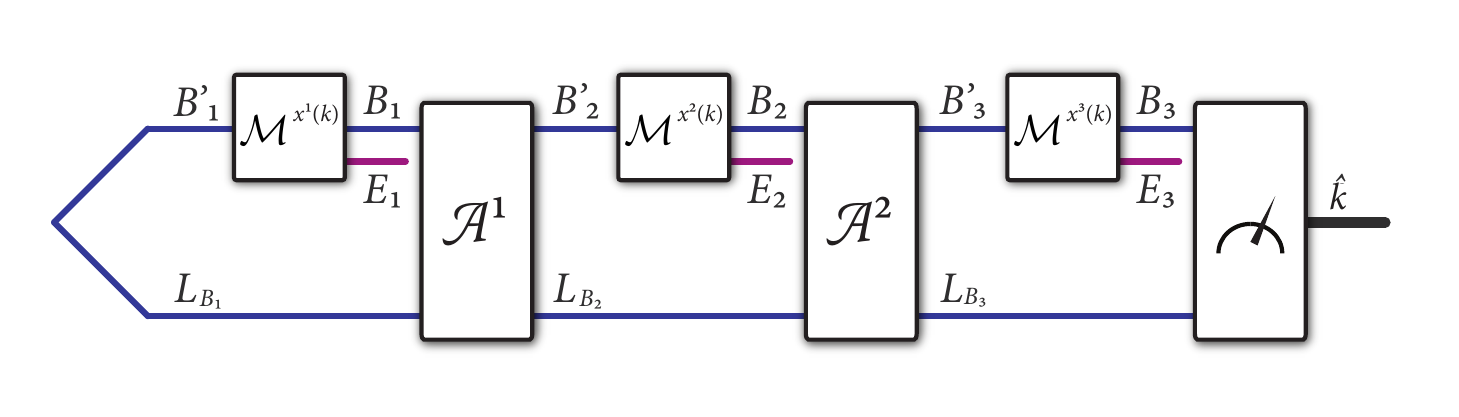}
		\caption{The figure depicts a private reading protocol that calls a memory cell three times to
decode the key $k$ as $\hat{k}$. See the discussion in Section~\ref{sec:q-r-p} for a detailed description of a private
reading protocol.}\label{fig:priv-read}
	\end{figure*}

In a private reading protocol, we consider an encoder and a reader (decoder). Alice, an encoder, is one who encodes a secret classical message onto a read-only memory device that is delivered to Bob, a receiver, whose task is to read the message. We also refer to Bob as the reader. The private reading task comprises the estimation of the secret message encoded in the form of a sequence of quantum wiretap channels chosen from a given set $\{\mc{M}^x_{B'\to BE}\}_{x\in\mc{X}}$ of quantum wiretap channels (called a wiretap memory cell), where $\mathcal{X}$ is an alphabet, such that there is negligible leakage of information to Eve, who has access to the system $E$.
A special case of this is when each wiretap channel $\mc{M}^x_{B'\to BE}$ is an isometric channel.
In the most natural and general setting, the reader can use an adaptive strategy when decoding, as considered in \cite{DW17}.

Consider a set $\{\mc{M}^x_{B'\to BE}\}_{x\in\mc{X}}$ of wiretap quantum channels, where the size of $B'$, $B$, and $E$ are fixed and independent of $x$. The memory cell from the encoder Alice to the reader Bob is as follows: $\overline{\mathcal{M}}_{\mathcal{X}}=\{\mathcal{M}^x_{B'\to B}\}_{x}$, where
\begin{equation}
\forall x\in\mc{X}:\ \mathcal{M}^x_{B'\to B}(\cdot)\coloneqq \Tr_E\{\mc{M}^x_{B'\to BE}(\cdot)\},
\end{equation}
 which may also be known to Eve, before executing the reading protocol. We assume only the systems $E$ are accessible to Eve for all channels $\mc{M}^{x}$ in a memory cell. Thus, Eve is a passive eavesdropper in the sense that  all she can do is to access the output of the channels \begin{equation}
\forall x\in\mc{X}:\  \mc{M}^x_{B'\to E}(\cdot)=\Tr_{B}\left\{\mc{M}^x_{B'\to B E}(\cdot)\right\}.
 \end{equation}
 
 We consider a classical message set $\mathcal{K}=\{1,2,\ldots,K\}$, and let $K_A$ be an associated system denoting a classical register for the secret message. In general, Alice encodes a message $k\in\mathcal{K}$ using a codeword
$x^n(k)=x_1(k)x_2(k)\cdots x_n(k)$
of length $n$, where $x_i(k)\in\mathcal{X}$ for all $i\in\{1,2,\ldots,n\}$. Each codeword identifies with a corresponding sequence of quantum channels chosen from the wiretap memory cell $\overline{\mathcal{M}}_{\mathcal{X}}$:
\begin{equation} 
\(\mathcal{M}^{x_1(k)}_{{B}_1'\to B_1 E_1}, \mathcal{M}^{x_2(k)}_{B_2'\to B_2 E_2},\ldots,\mathcal{M}^{x_n(k)}_{B_n'\to B_n E_n}\).
\end{equation}

An adaptive decoding strategy makes $n$ calls to the memory cell, as depicted in Figure~\ref{fig:priv-read}. It is specified in terms of a transmitter state $\rho_{L_{B_1}B_1'}$, a set of adaptive, interleaved channels $\{\mc{A}^i_{L_{B_i}B_i\to L_{B_{i+1}}B'_{i+1}}\}_{i=1}^{n-1}$, and a final quantum measurement $\{\Lambda^{(\hat{k})}_{L_{B_n}B_n}\}_{\hat{k}}$ that outputs an estimate~$\hat{k}$ of the message~$k$. The strategy begins with Bob preparing the input state $\rho_{L_{B_1}B'_1}$ and sending the $B'_1$ system into the channel $\mc{M}^{x_1(k)}_{B'_1\to B_1 E_1}$. The channel outputs the system $B_1$ for Bob. He adjoins the system $B_1$ to the system $L_{B_1}$ and applies the channel $\mc{A}^1_{L_{B_1}B_1\to L_{B_2}B'_2}$. The channel $\mc{A}^i_{L_{B_i}B_i\to L_{B_{i+1}}B'_{i+1}}$ is called adaptive because it can take an action conditioned on the information in the system $B_i$, which itself might contain partial information about the message~$k$. Then, he sends the system $B'_2$ into the  channel $\mc{M}^{x_2(k)}_{B'_2\to B_2 E_2}$, which outputs systems $B_2$ and $E_2$. The process of successively using the channels interleaved by the adaptive channels  continues $n-2$ more times, which results in the final output systems $L_{B_n}$ and $B_n$ with Bob. Next, he performs a measurement $\{\Lambda^{(\hat{k})}_{L_{B_n}B_n}\}_{\hat{k}}$ on the output state $\rho_{L_{B_n} B_n}$, and the measurement outputs an estimate $\hat{k}$ of the original message $k$. It is natural to assume that the outputs of the adaptive channels and their complementary channels are inaccessible to Eve and are instead held securely by Bob.

The physical model that we assume, as is standard in QKD protocols, is that Bob's local laboratory is secure. So Bob can perform whatever local operations that he would like to in his lab. Furthermore, without loss of generality, Bob can perform all of these local steps as isometric channels, sending the original output as output and keeping the former environment to himself, thus ensuring that the new complement of each isometric channel is trivial so that Eve gets no information from these steps. So the task does not change even if we assume that Eve has access to the complements of each of the adaptive channels since it is possible to do things in this way without loss of generality.

It is apparent that a non-adaptive strategy is a special case of an adaptive strategy. In a non-adaptive strategy, the reader does not perform any adaptive channels and instead uses $\rho_{L_B{B'}^n}$ as the transmitter state with each $B'_i$ system passing through the corresponding channel $\mc{M}^{x_i(k)}_{B'_i\to B_i E_i}$ and $L_B$ being a reference system. The final step in such a non-adaptive strategy is to perform a decoding measurement on the  joint system $L_BB^n$. 

As argued in \cite{DW17}, based on the physical setup of quantum reading, in which the reader assumes the role of both a transmitter and receiver, it is natural to consider the use of an adaptive strategy when defining the private reading capacity of a memory cell. 

\begin{definition}[Private reading protocol]\label{def:QR}
An $(n,K,\varepsilon,\delta)$ private reading protocol for a wiretap memory cell $\overline{\mc{M}}_{\mathcal{X}}$ is defined by an encoding map $\mc{K}\to \mc{X}^{\otimes n}$, an adaptive strategy with measurement $\{\Lambda_{L_{B_n} B_n}^{(\hat{k})}\}_{\hat{k}}$, such that the average success probability is at least $1-\varepsilon$ where $\varepsilon\in(0,1)$:
\begin{equation}
1- \varepsilon \leq 1 - p_{\operatorname{err}} :=\frac{1}{K} \sum_{k}\Tr\left\{\Lambda^{(k)}_{L_{B_n}B_n}\rho^{(k)}_{L_{B_n}B_n}\right\},
\end{equation} 
where
\begin{multline}
\label{eq3.3}
\rho^{(k)}_{L_{B_n}B_nE^n}=\Big(\mc{M}^{x_n(k)}_{B'_n\to B_nE_n}\circ\mc{A}^{{n-1}}_{L_{B_{n-1}}B_{n-1}\to L_{B_n}B'_n}\circ\\
\cdots\circ\mc{A}^{1}_{L_{B_1}B_1\to L_{B_2}B'_2}\circ\mc{M}^{x_1(k)}_{B'_1\to B_1E_1}\Big)\(\rho_{L_{B_1}B'_1}\).
\end{multline}
Furthermore, the security condition is that
\begin{equation}
\frac{1}{K}\sum_{ k\in\mc{K}} \frac{1}{2}\left\|\rho^{(k)}_{E^n}- \tau_{E^n}\right\|_1\leq\delta,
\end{equation}
where $\rho^{(k)}_{E^n}$ denotes the state accessible to the passive eavesdropper when message $k$ is encoded. Also, $\tau_{E^n}$ is some fixed state. The rate $P\coloneqq\frac{1}{n}\log_2 K$ of a given $(n,K,\varepsilon,\delta)$ private reading protocol is equal to the number of secret bits read per channel use.
\end{definition}

Based on the discussions in \cite[Appendix~B]{WTB16}, there are connections between the notions of private communication given in Section~\ref{sec:priv-dist-protocol} and Definition~\ref{def:QR}, and we exploit these in what follows. 

To arrive at a definition of the private reading capacity, we demand that there exists a sequence of private reading protocols, indexed by $n$, for which the error probability $p_{\operatorname{err}}\to 0$ and security parameter $\delta\to 0$ as $n\to \infty$ at a fixed rate~$P$.

A rate $P$ is called achievable if for all $\varepsilon,\delta \in (0,1]$, $\delta^\prime >0$, and sufficiently large $n$, there exists an $(n,2^{n(P-\delta^\prime)},\varepsilon,\delta)$ private reading protocol. 
The private reading capacity $P^{\textnormal{read}}(\overline{\mc{M}}_{\mc{X}})$ of a wiretap memory cell $\overline{\mc{M}}_{\mc{X}}$ is defined as the supremum of all achievable rates.

An $(n,K,\varepsilon,\delta)$ private reading protocol for a wiretap memory cell $\overline{\mc{M}}_{\mathcal{X}}$ is a non-adaptive private reading protocol when the reader abstains from employing any adaptive strategy for decoding. 
The non-adaptive private reading capacity $P^{\textnormal{read}}_{\textnormal{n-a}}(\overline{\mc{M}}_{\mc{X}})$ of a wiretap  memory cell $\overline{\mc{M}}_{\mc{X}}$ is defined as the supremum of all achievable rates for a private reading protocol that is limited to non-adaptive strategies.

\subsection{Non-adaptive private reading capacity}\label{sec:na-priv-read}

In what follows we restrict our attention to reading protocols that employ a non-adaptive strategy, and we now derive a regularized expression  for the non-adaptive private reading capacity of a general wiretap memory cell.

\begin{theorem}\label{thm:n-a-priv-read}
The non-adaptive private reading capacity of a wiretap memory cell $\overline{\mc{M}}_{\mc{X}}$ is given by
\begin{multline}
P^{\textnormal{read}}_{\textnormal{n-a}}\(\overline{\mc{M}}_{\mc{X}}\) =
\\ \sup_{n}\max_{p_{X^n},\sigma_{L_B{B'}^n}}\frac{1}{n}\[I(X^n;L_BB^n)_\tau-I(X^n;E^n)_\tau\],
\end{multline}
where
\begin{multline}\label{eq:cq-na-read}
\tau_{X^nL_BB^nE^n} \coloneqq
\\\sum_{x^n}p_{X^n}(x^n)\ket{x^n}\!\bra{x^n}_{X^n}\otimes \mc{M}^{x^n}_{{B'}^n\to B^nE^n}(\sigma_{L_B{B'}^n}),
\end{multline}
and it suffices for $\sigma_{L_B{B'}^n}$ to be  a pure state such that $L_B \simeq {B'}^n$.
\end{theorem}
\begin{proof}
Let us begin by defining a cq-state corresponding to the task of private reading. Consider a wiretap memory cell $\overline{\mc{M}}_{\mc{X}}=\{\mc{M}^x_{B'\to BE}\}_{x\in\mc{X}}$. The initial  state $\rho_{K_AL_B{B'}^n}$ of a non-adaptive private reading protocol takes  the form
\begin{equation}
\rho_{K_AL_B{B'}^n}\coloneqq\frac{1}{K}\sum_{k}|k\>\<k|_{K_A}\otimes\rho_{L_B{B'}^n}.
\end{equation}
The action of the encoding is to apply an instrument that  measures the $K_A$ register and, conditioned on the outcome, presents Bob with a channel codeword sequence $\mathcal{M}^{x^n(k)}_{{B'}^n\to B^n E^n}\coloneqq \bigotimes_{i=1}^n\mc{M}^{x_i(k)}_{{B'}_i\to B_i E_i}$. Bob then passes the transmitter state $\rho_{L_B{B'}^n}$ through $\mathcal{M}^{x^n(k)}_{{B'}^n\to B^n E^n}$. 
Then the resulting state is
\begin{equation}
\rho_{K_AL_BB^nE^n}=\frac{1}{K}\sum_{k}|k\>\<k|_{K_A}\otimes\mc{M}^{x^n(k)}_{{B'}^n\to B^n E^n}\left(\rho_{L_B{B'}^n}\right).
\end{equation}
Let $\rho_{K_AK_B}=\mc{D}_{L_BB^n\to K_B}\(\rho_{K_AL_BB^n}\)$ be the output state at the end of the protocol after the decoding channel $\mc{D}_{L_BB^n\to K_B}$ is performed by Bob. The privacy criterion introduced in Definition~\ref{def:QR} requires that
\begin{equation}
\frac{1}{K} \sum_{ k\in\mc{K}}  \frac{1}{2}\left\Vert \rho^{x^n(k)}_{E^n}-\tau_{E^n}\right\V_1\leq\delta,
\end{equation}
where 
$\rho^{x^n(k)}_{E^n}\coloneqq \operatorname{Tr}_{L_BB^n} \{{\mc{M}^{x^n(k)}}_{{B'}^n\to B^n E^n}\left(\rho_{L_B{B'}^n}\right)\}$ and 
$\tau_{E^n}$
is some arbitrary constant state. Hence
\begin{align}
\delta & \geq \frac{1}{2}\sum_{k}\frac{1}{K}\left\Vert \rho^{x^n(k)}_{E^n}-\tau_{E^n}\right\Vert_1 \\ 
& =\frac{1}{2}\Vert \rho_{K_AE^n}-\pi_{K_A}\otimes\tau_{E^n}\Vert_1,
\end{align}
where $\pi_{K_A}$ denotes maximally mixed state, i.e., $\pi_{K_A}\coloneqq\frac{1}{K}\sum_{k}|k\>\<k|_{K_A}$. We note that
\begin{align}
I(K_A;E^n)_\rho &= S(K_A)_\rho-S(K_A|E^n)_\rho\\
&= S(K_A|E^n)_{\pi\otimes\tau}-S(K_A|E^n)_\rho\\
&\leq \delta\log_2 K+g(\delta),
\end{align}
which follows from an application of Lemma~\ref{thm:AFW}.

We are now ready to derive a weak converse bound on the private reading rate:
\begin{align}
& \log_2 K \nonumber \\
& = S(K_A)_\rho \nonumber\\ 
&=I(K_A;K_B)_\rho+S(K_A|K_B)_{\rho}\nonumber\\
&\leq I(K_A;K_B)_\rho+\varepsilon\log_2K+h_2(\varepsilon)\nonumber\\
&\leq I(K_A;L_BB^n)_\rho+\varepsilon\log_2K+h_2(\varepsilon)\nonumber\\ 
&\leq I(K_A;L_BB^n)_\rho-I(K_A;E^n)_\rho+\varepsilon\log_2K\nonumber\\ &\qquad +h_2(\varepsilon)+\delta\log_2K+g(\delta)\nonumber\\ 
&\leq \max_{p_{X^n},\sigma_{L_B{B'}^n}}\[I(X^n;L_BB^n)_\tau-I(X^n;E^n)_\tau\]\nonumber\\ &\qquad +\varepsilon\log_2K+h_2(\varepsilon)+\delta\log_2K+g(\delta),
\end{align}
where $\tau_{X^nL_BB^nE^n}$ is a state of the form  in \eqref{eq:cq-na-read}.
The first inequality follows from Fano's inequality \cite{F08}. The second inequality follows from the monotonicity of mutual information under the action of a local quantum channel by Bob (Holevo bound). The final inequality follows because the maximization is over all possible probability distributions and input states. Then,
\begin{multline}
\frac{\log_2K}{n} (1-\varepsilon-\delta)\\ \leq \max_{p_{X^n},\sigma_{L_B{B'}^n}}\frac{1}{n}\[I(X^n;L_BB^n)_\tau-I(X^n;E^n)_\tau\]\\ +\frac{h_2(\varepsilon)+g(\delta)}{n}.
\end{multline}
Now considering a sequence of non-adaptive $(n,K_n,\varepsilon_n,\delta_n)$ protocols with $\lim_{n \to \infty} \frac{\log_2 K_n}{n} = P$,
$\lim_{n \to \infty} \varepsilon_n = 0$, and $\lim_{n \to \infty} \delta_n = 0$, 
the converse bound on non-adaptive private reading capacity of memory cell $\overline{\mc{M}}_{\mc{X}}$ is given by
\begin{equation}\label{eq:priv-na-rad-up}
P\leq\sup_{n}\max_{p_{X^n},\sigma_{L_B{B'}^n}}\frac{1}{n}\[I(X^n;L_BB^n)_\tau-I(X^n;E^n)_\tau\],
\end{equation} 
which follows by taking the limit as $n\to \infty$.

It follows from the results of \cite{D05,DW05} that right-hand side of \eqref{eq:priv-na-rad-up} is also an achievable rate in the limit $n\to\infty$. Indeed, the encoder and reader can induce the cq wiretap channel $x \to \mc{M}^{x}_{{B'}\to B E}(\sigma_{L_BB'})$, to which the results of \cite{D05,DW05} apply. A regularized coding strategy then gives the general achievability statement. Therefore, the non-adaptive private reading capacity is given as stated in the theorem.
\end{proof}

\subsection{Purifying private reading protocols}\label{sec:n-a-priv-read-coherent}

As observed in \cite{HHHO05,HHHO09} and reviewed in Section~\ref{sec:rev-priv-states}, any protocol of the above form, discussed in Section~\ref{sec:na-priv-read}, can be purified in the following sense. In this section, we assume that each wiretap memory cell consists of a set of isometric channels, written as
$\{\mc{U}^{\mc{M}^x}_{B'\to BE}\}_x$. 
Thus, Eve has access to system $E$, which is the output of a particular isometric extension of the channel $\mc{M}^x_{B'\to B}$, i.e., $\widehat{\mc{M}}^x_{B'\to E}(\cdot) = 
\Tr_B\{\mc{U}^{\mc{M}^x}_{B'\to BE}(\cdot)\}$, for all $x\in\mc{X}$. We refer to such memory cell as an isometric wiretap memory cell.

We begin by considering non-adaptive private reading protocols. A non-adaptive purified secret-key-agreement protocol that uses an isometric wiretap memory cell begins with Alice preparing a purification of the maximally classically correlated state: 
\begin{equation}
\frac{1}{\sqrt{K}}\sum_{k\in\mc{K}}\ket{k}_{K_A}\ket{k}_{\hat{K}}\ket{k}_{C},
\end{equation}
where $\mc{K}=\{1,2,\ldots,K\}$, and $K_A$, $\hat{K}$, and $C$ are classical registers. Alice coherently encodes the value of the register $C$ using the memory cell, the codebook $\{x^n(k)\}_k$,
and the isometric mapping $\ket{k}_C \to \ket{x^n(k)}_{X^n}$. Alice makes two coherent copies of the codeword $x^n(k)$ and stores them safely in coherent  classical registers $X^n$ and $\hat{X}^n$. At the same time, she acts on Bob's input state $\rho_{L_B{B'}^n}$ with the following isometry:
\begin{equation}\label{eq:iso-v-m}
 \sum_{x^n}\ket{x^n}\!\bra{x^n}_{X^n}\otimes{U}^{\mc{M}^{x^n}}_{{B'}^n\to B^nE^n}\otimes\ket{x^n}_{\hat{X}^n}.
\end{equation} 
For the task of reading, Bob inputs the state $\rho_{L_B{B'}^n}$ to the channel sequence $\mc{M}^{x^n(k)}$, with the goal of decoding $k$. In the purified setting, the resulting output state is $\psi_{K_A\hat{K}X^nL_B'L_BB^nE^n\hat{X}^n}$, which  includes all concerned coherent classical registers or quantum systems accessible by Alice, Bob and Eve:
\begin{multline}
\label{eq:coh-mem-state}
\ket{\psi}_{K_A\hat{K}X^nL_B'L_BB^nE^n\hat{X}^n} \coloneqq \frac{1}{\sqrt{K}}\sum_{k}\ket{k}_{K_A}\ket{k}_{\hat{K}}\otimes 
\\\ket{x^n(k)}_{X^n}
U^{\mc{M}^{x^n}}_{{B'}^n\to B^nE^n}\ket{\psi}_{L_B'L_B{B'}^n}\ket{x^n(k)}_{\hat{X}^n},
\end{multline}
where $\psi_{L_B'L_B{B'}^n}$ is a purification of $\rho_{L_B{B'}^n}$ and the systems $L_B'$, $L_B$, and $B^n$ are held by Bob, whereas Eve has access only to $E^n$. The final global state is $\psi_{K_A\hat{K}X^nL_B'K_BE^n\hat{X}^n}$  after Bob applies the decoding channel $\mc{D}_{L_BB^n\to K_B}$, where 
\begin{multline}
\ket{\psi}_{K_A\hat{K}X^nL_B'L_B''K_BE^n\hat{X}^n}\\ \coloneqq U^{\mc{D}}_{L_BB^n\to L_B'' K_B}\ket{\psi}_{K_A\hat{K}X^nL_B'L_BB^nE^n\hat{X}^n},
\end{multline}
$U^{\mc{D}}$ is an isometric extension of the decoding channel $\mc{D}$, and $L_B''$ is part of the shield system of Bob. 

At the end of the purified protocol, Alice possesses the key system $K_A$ and the shield systems $\hat{K}X^n\hat{X}^n$, Bob possesses the key system $K_B$ and the shield systems $L_B'L_B''$, and Eve possesses the environment system $E^n$. The state $\psi_{K_A\hat{K}X^nL_B'L_B''K_B\hat{X}^nE^n}$ at the end of the protocol is a pure state.

For a fixed $n,\ K\in\mathbb{N},\ \varepsilon\in[0,1]$, the original protocol is an $(n,2^{nP},\sqrt{\varepsilon},\sqrt{\varepsilon})$ private reading protocol if the memory cell is called $n$ times as discussed above, and if 
\begin{equation}\label{eq:coh-key-aprox-priv}
F(\psi_{K_A\hat{K}X^nL_B'L_B''K_B\hat{X}^n},\gamma_{S_AK_AK_BS_B})\geq 1-\varepsilon,
\end{equation}
where $\gamma$ is a private state such that $S_A=\hat{K}X^n\hat{X}^n,\ K_A=K_A, \ K_B=K_B,\ S_B=L_B'L_B''$. See \cite[Appendix~B]{WTB16} for further details.

Similarly, it is possible to purify a general adaptive private reading protocol, but we omit the details.

\subsection{Converse bounds on private reading capacities}\label{sec:priv-read-sc}

In this section, we derive different upper bounds on the private reading capacity of an isometric wiretap memory cell. 
The first is a weak converse upper bound on the non-adaptive private reading capacity in terms of the squashed entanglement. The second is a strong converse upper bound on the (adaptive) private reading capacity in terms of the bidirectional max-relative entropy of entanglement. Finally, we evaluate the private reading capacity for an example:  a qudit erasure memory cell. 


We derive the first converse bound on non-adaptive private reading capacity by making the following observation, related to the development in \cite[Appendix~B]{WTB16}: any non-adaptive $(n,2^{nP},\varepsilon,\delta)$ private reading protocol of an isometric wiretap memory cell $\overline{\mc{M}}_{\mc{X}}$, for reading out a secret key, can be realized by  an $(n,2^{nP},\varepsilon'(2-\varepsilon'))$ non-adaptive purified secret-key-agreement reading protocol, where $\varepsilon'\coloneqq\varepsilon+2\delta$. As such, a converse bound for the latter protocol implies a converse bound for the former.

First, we derive an upper bound on the non-adaptive private reading capacity in terms of the squashed entanglement \cite{CW04}: 
\begin{proposition}\label{prop:EsqBound}
The non-adaptive private reading capacity $P^{\textnormal{read}}_{\textnormal{n-a}}(\overline{\mc{M}}_{\mc{X}})$ of an isometric wiretap memory cell $\overline{\mc{M}}_{\mc{X}}=\{\mc{U}^{\mc{M}^x}_{B'\to BE}\}_{x\in\mc{X}}$ is bounded from above as
\begin{equation}
P^{\textnormal{read}}_{\textnormal{n-a}}(\overline{\mc{M}}_{\mc{X}})\leq \sup_{p_X,\psi_{LB'}}E_{\sq}(XL_B;B)_\omega,
\end{equation} 
where $\omega_{XL_BB}=\Tr_{E}\{\omega_{XL_BBE}\}$, such that $\psi_{L_BB'}$ is a pure state and
\begin{equation}
\vert \omega\rangle_{XLBE}=\sum_{x\in\mc{X}}\sqrt{p_X(x)}|x\>_X\otimes{U}^{\mc{M}^{x}}_{B'\to BE}\ket{\psi}_{L_BB'} \label{eq:opt-form-sq-bnd}.
\end{equation}
\end{proposition}

\begin{proof}
For the discussed purified non-adaptive secret-key-agreement reading protocol, when \eqref{eq:coh-key-aprox-priv} holds, the dimension of the secret key system is upper bounded as \cite[Theorem 2]{Wil16}:
\begin{equation}
\log_2 K\leq E_{\sq}(\hat{K}X^n\hat{X}^nK_A;K_BL_BL_B'')_{\psi}+f_1(\sqrt{\varepsilon},K),
\end{equation}
where 
\begin{equation}
f_1(\varepsilon,K_A)\coloneqq 2\varepsilon\log_2K+2g(\varepsilon).
\end{equation}
We can then proceed as follows:
\begin{align}
\log_2 K&\leq E_{\sq}(\hat{K}X^n\hat{X}^nK_A;K_BL_B''L_B')_{\psi}+f_1(\sqrt{\varepsilon},K)\\
&= E_{\sq}(\hat{K}X^n\hat{X}^nK_A;B^nL_BL_B')_{\psi}+f_1(\sqrt{\varepsilon},K).\label{eq:priv-sq-bound}
\end{align}
where the first equality is due to the invariance of $E_{\sq}$ under isometries. 

For any five-partite pure state $\phi_{B'B_1B_2E_1E_2}$, the following inequality holds \cite[Theorem~7]{TGW14}:
\begin{equation}\label{eq:tri-sq-ent}
E_{\sq}(B';B_1B_2)_\phi\leq E_{\sq}(B'B_2E_2;B_1)_{\phi}+E_{\sq}(B'B_1E_1;B_2)_{\phi}.
\end{equation}
Choosing $B'=\hat{K}X^n\hat{X}^nK_A$, $B_1=B_n$, $B_2=L_BL_B'B^{n-1}$, $E_1=E_n$ and $E_2=E^{n-1}$, this implies that
\begin{align}
& E_{\sq}(\hat{K}X^n\hat{X}^nK_A;B^nL_BL_B')_{\psi}\nonumber\\
&\leq E_{\sq}(\hat{K}X^n\hat{X}^{n}K_AL_BL_B'B^{n-1}E^{n-1};B_n)_{\psi}\nonumber\\ &\qquad +E_{\sq}(\hat{K}X^n\hat{X}^{n}K_AB_nE_n;L_BL_B'B^{n-1})_{\psi}\nonumber\\
&= E_{\sq}(\hat{K}X^n\hat{X}^nK_AL_BL_B'B^{n-1}E^{n-1};B_n)_{\psi}\nonumber\\ &\qquad +E_{\sq}(\hat{K}X^n\hat{X}^{n-1}K_AB'_n;L_BL_B'B^{n-1})_{\psi}.\label{eq:sq-ine-n-pre}
\end{align}
where the equality holds by considering an isometry with the following uncomputing action:
\begin{widetext}
\begin{align}
&\ket{k}_{K_A}\ket{k}_{\hat{K}}\ket{x^n(k)}_{X^n}U^{\mc{M}^{x^n}}_{{B'}^n\to B^nE^n}\ket{\psi}_{L_B'L_B{B'}^n}\ket{x^n(k)}_{\hat{X}^n}\nonumber\\ &\qquad\qquad \qquad
\to \ket{k}_{K_A}\ket{k}_{\hat{K}}\ket{x^n(k)}_{X^n}U^{\mc{M}^{x^{n-1}}}_{{B'}^{n-1}\to B^{n-1}E^{n-1}}\ket{\psi}_{L_B'L_B{B'}^n}\ket{x^{n-1}(k)}_{\hat{X}^{n-1}}.
\end{align}
\end{widetext}

 Applying the inequality in \eqref{eq:tri-sq-ent} and uncomputing isometries like the above repeatedly to \eqref{eq:sq-ine-n-pre}, we find that 
\begin{multline}\label{eq:single-let-sq}
E_{\sq}(\hat{K}X^n\hat{X}^nK_A;B^nL_BL_B')_{\psi}\\ \leq \sum_{i=1}^n E_{\sq}(\hat{K}X^n\hat{X}_iK_AL_BL_B'B'^{n\setminus \{i\}};B_i),
\end{multline}
where the notation ${B'}^{n\setminus\{i\}}$ indicates the composite system $B'_1B'_2\cdots B'_{i-1}B'_{i+1}\cdots B'_n$, i.e. all $n-1$\ $B'$-labeled systems except $B'_i$.
Each summand above is equal to the squashed entanglement of some state of the following form: a bipartite state is prepared on some auxiliary system $Z$ and a control system $X$, a bipartite state is prepared on systems $L_B$ and $B'$, a controlled isometry $\sum_x \vert x\rangle \langle x\vert_X \otimes {U}^{\mc{M}^{x}}_{B'\to BE}$ is performed from $X$ to $B'$, and then $E$ is traced out. By applying the development in \cite[Appendix~A]{CY16}, we conclude that the auxiliary system $Z$ is not necessary. Thus, the state of
systems $X$, $L_B$, $B'$, and $E$ can be taken to have the form in \eqref{eq:opt-form-sq-bnd}.
From \eqref{eq:priv-sq-bound} and the above reasoning, since $\lim_{\varepsilon\to 0}\lim_{n\to \infty} \frac{f_1(\sqrt{\varepsilon},K)}{n}=0$, we conclude that 
\begin{equation}
\widetilde{P}^{\textnormal{read}}_{\textnormal{n-a}}(\overline{\mc{M}}_{\mc{X}})\leq \sup_{p_X,\psi_{L_BB'}}E_{\sq}(XL;B)_\omega,
\end{equation} 
where $\omega_{XL_BB}=\Tr_{E}\{\omega_{XL_BBE}\}$, such that $\psi_{L_BB'}$ is a pure state and
\begin{equation}
\vert \omega\rangle_{XL_BBE}=\sum_{x\in\mc{X}}\sqrt{p_X(x)}|x\>_X\otimes{U}^{\mc{M}^{x}}_{B'\to BE}\ket{\psi}_{L_BB'}.
\end{equation}
This concludes the proof.
\end{proof}
\bigskip

We now bound the strong converse private reading capacity of an isometric wiretap memory cell in terms of the bidirectional max-relative entropy.
\begin{theorem}\label{thm:priv-read-strong-converse}
The strong converse private reading capacity $\widetilde{P}^{\textnormal{read}}(\overline{\mc{M}}_{\mc{X}})$ of an isometric wiretap memory cell $\overline{\mc{M}}_{\mc{X}}=\{\mc{U}^{\mc{M}^x}_{B'\to BE}\}_{x\in\mc{X}}$ is  bounded from above by the bidirectional max-relative entropy of entanglement $E^{2\to 2}_{\max}(\mc{N}^{\overline{\mc{M}}_{\mc{X}}}_{X'B'\to XB})$ of the bidirectional channel $\mc{N}^{\overline{\mc{M}}_{\mc{X}}}_{X'B'\to XB}$, i.e.,
\begin{equation}\label{eq:priv-read-strong-converse}
\widetilde{P}^{\textnormal{read}}(\overline{\mc{M}}_{\mc{X}})\leq E^{2\to 2}_{\max}(\mc{N}^{\overline{\mc{M}}_{\mc{X}}}_{XB'\to XB}),
\end{equation}
where 
\begin{equation}\label{eq:memory-cell-bidir}
\mc{N}^{\overline{\mc{M}}_{\mc{X}}}_{XB'\to XB}(\cdot)\coloneqq \Tr_{E}\left\{U^{\overline{\mc{M}}_{\mc{X}}}_{XB'\to XBE}(\cdot)\(U^{\overline{\mc{M}}_{\mc{X}}}_{XB'\to XBE}\)^\dag\right\},
\end{equation}
such that
\begin{equation}
U^{\overline{\mc{M}}_{\mc{X}}}_{XB'\to XBE}\coloneqq \sum_{x\in\mc{X}}\ket{x}\!\bra{x}_{X}\otimes U^{\mc{M}^x}_{B'\to BE}.
\end{equation}
\end{theorem}

\begin{proof}
First we recall, as stated previously,
that a $(n,2^{nP},\varepsilon,\delta)$ (adaptive) private reading protocol of a memory cell $\overline{\mc{M}}_{\mc{X}}$, for reading out a secret key, can be realized by an $(n,2^{nP},\varepsilon'(2-\varepsilon'))$  purified secret-key-agreement reading protocol, where $\varepsilon'\coloneqq \varepsilon+2 \delta$.
Given that a purified secret-key-agreement reading protocol can be understood as particular case of a bidirectional secret-key-agreement protocol (as discussed in Section~\ref{sec:priv-dist-protocol}), we conclude that the strong converse private reading capacity is bounded from above by 
\begin{equation}
\widetilde{P}^{\textnormal{read}}_{\textnormal{n-a}}(\overline{\mc{M}}_{\mc{X}}) \leq E^{2\to 2}_{\max}(\mc{N}^{\overline{\mc{M}}_{\mc{X}}}_{XB'\to XB}),
\end{equation}
where the bidirectional channel is
\begin{equation}
\mc{N}^{\overline{\mc{M}}_{\mc{X}}}_{XB'\to XB}(\cdot)=\Tr_{E}\left\{U^{\overline{\mc{M}}_{\mc{X}}}_{XB'\to XBE}(\cdot)\(U^{\overline{\mc{M}}_{\mc{X}}}_{XB'\to XBE}\)^\dag\right\},
\end{equation}
such that
\begin{equation}
U^{\overline{\mc{M}}_{\mc{X}}}_{XB'\to XBE}\coloneqq \sum_{x\in\mc{X}}\ket{x}\!\bra{x}_{X}\otimes U^{\mc{M}^x}_{B'\to BE}. 
\end{equation}
The reading protocol is a particular instance of an LOCC-assisted bidirectional secret-key-agreement protocol in which classical communication between Alice and Bob does not occur. The local operations of Bob in the bidirectional secret-key-agreement protocol are equivalent to adaptive operations by Bob in reading. Therefore, applying Theorem~\ref{thm:emax-ent-dist-strong-converse}, we conclude that \eqref{eq:priv-read-strong-converse} holds, where the strong converse in this context means that $\varepsilon+2 \delta \to 1$ in the limit as $n\to \infty$ if the reading rate exceeds $E^{2\to 2}_{\max}(\mc{N}^{\overline{\mc{M}}_{\mc{X}}}_{XB'\to XB})$.\footnote{Such a bound might be called a ``pretty strong converse,'' in the sense of \cite{MW13}. However, we could have alternatively defined a private reading protocol to have a single parameter characterizing reliability and security, as in \cite{WTB16}, and with such a definition, we would get a true strong converse.}
\end{proof}

\subsubsection{Qudit erasure wiretap memory cell}

The main goal of this section is to evaluate the private reading capacity of the qudit erasure wiretap memory cell \cite{DW17}. 

\begin{definition}[Erasure wiretap memory cell]\label{def:qudit-erasure}
The qudit erasure wiretap memory cell $\overline{\mc{Q}}^q_{\mc{X}}=\left\{\mc{Q}^{q,x}_{B'\to BE}\right\}_{x\in\mc{X}}$, $\vert\mc{X}\vert=d^2$, consists of the following qudit channels:
\begin{equation}
\mc{Q}^{q,x}(\cdot)=\mc{Q}^q(\sigma^x(\cdot)\(\sigma^x\)^\dag) ,
\end{equation}
where $\mc{Q}^q$ is an isometric channel extending the qudit erasure channel \cite{GBP97}:
\begin{align}
\mc{Q}^q(\rho_{B'})& =U^{q} \rho_{B'} (U^q)^\dag,\\
U^q \vert \psi\rangle_{B'}
&= \sqrt{1-q}\vert \psi \rangle_B \vert e \>_E + \sqrt{q}|e\>_B \vert \psi\>_E,
\end{align}
such that $q\in [0,1]$, $\dim(\mc{H}_{B'})=d$, $|e\>\<e|$ is some state orthogonal to the support of input state $\rho$, and
$\forall x\in\mc{X}: \sigma^x\in\mathbf{H}$ are the Heisenberg--Weyl operators as reviewed in \eqref{eq:HW-op} of Appendix~\ref{app:qudit}. Observe that $\mc{Q}^q_{\mc{X}}$ is jointly covariant with respect to the Heisenberg--Weyl group $\mathbf{H}$ because the qudit erasure channel $\mc{Q}^q$ is covariant with respect to $\mathbf{H}$.
\end{definition}

Now we establish the private reading capacity of the qudit erasure wiretap memory cell. 
 
\begin{proposition}
The private reading capacity and strong converse private reading capacity of the qudit erasure wiretap memory cell $\overline{\mc{Q}}^q_{\mc{X}}$ are given by 
\begin{equation}
P^{\operatorname{read}}(\overline{\mc{Q}}^q_{\mc{X}})=
\widetilde{P}^{\operatorname{read}}(\overline{\mc{Q}}^q_{\mc{X}})=
2(1-q)\log_2 d.
\end{equation}
\end{proposition}

\begin{proof}
To prove the proposition, consider that
$\mc{N}^{\overline{\mc{Q}}^q_{\mc{X}}}$ as defined in \eqref{eq:memory-cell-bidir} is bicovariant and $\mc{Q}^q_{B'\to B}$ is covariant. Thus, to get an upper bound on the strong converse private reading capacity, it is sufficient to consider the action of a coherent use of the memory cell on a maximally entangled state (see Corollary~\ref{cor:str-conv-TP-simul}). 
We furthermore apply the development in \cite[Appendix~A]{CY16} to restrict to the following state:
\begin{align}
&\phi_{XL_BBE}\nonumber\\ 
&\coloneqq\frac{1}{\sqrt{|\mc{X}|}}\sum_{x\in\mc{X}}\ket{x}_{X}\otimes U^{\mc{Q}^{q,x}}_{B'\to BE}\ket{\Phi}_{L_B{B'}}\nonumber\\
&=\sqrt{\frac{1-q}{d\mc{|\mc{X}|}}}\sum_{i=0}^d\sum_{x}\ket{x}_{X}\otimes\sigma^x\ket{i}_B\ket{i}_{L_B}\ket{e}_{E} \nonumber\\ 
&\qquad +\sqrt{\frac{q}{d\mc{|\mc{X}|}}}\sum_{i=0}^d\sum_{x}\ket{x}_{X}\otimes\ket{e}_B\ket{i}_{L_B}\otimes \sigma^x\ket{i}_{E}.
\end{align}
Observe that
$\sum_{i=0}^{d-1}\sum_{x}\ket{x}_{X}\otimes\ket{e}_B\ket{i}_{L_B}\otimes \sigma^x\ket{i}_{E}$ and $\sum_{i=0}^{d-1}\sum_{x}\ket{x}_{X}\otimes\sigma^x\ket{i}_B\ket{i}_{L_B}\ket{e}_{E}$ are orthogonal. Also, since, $\ket{e}$ is orthogonal to the input Hilbert space, the only term contributing to the relative entropy of entanglement is $\sqrt{1-q}\frac{1}{d}\sum_{i=0}^d\sum_{x}\ket{x}_{X}\otimes\sigma^x\ket{i}_B\ket{i}_{L_B}$. Let
\begin{equation}
\ket{\psi}_{XL_BB}=\frac{1}{\sqrt{|\mc{X}|}}\sum_{x=0}^{d^2-1}\ket{x}_{X}\otimes\sigma^x\ket{\Phi}_{BL_B}.
\end{equation}
$\{\sigma^x\ket{\Phi}_{BL_B}\}_{x\in\mc{X}}$ forms an orthonormal basis in $\mc{H}_{B}\otimes\mc{H}_{L_B}$ (see Appendix~\ref{app:qudit}), so 
\begin{equation}
\ket{\psi}_{XL_BB}=\ket{\Phi}_{X:BL_B}=\frac{1}{d}\sum_{x=0}^{d^2-1}\ket{x}_{X}\otimes\ket{x}_{BL_B},
\end{equation} 
and $E(X;LB)_\Phi =2\log_2 d$. Applying Corollary~\ref{cor:str-conv-TP-simul} and convexity of relative entropy of entanglement, we conclude that
\begin{equation}\label{eq:e-cell-up}
\widetilde{P}^{\text{read}}(\overline{\mc{Q}}^q_{\mc{X}})\leq 2(1-q)\log_2 d.
\end{equation}
From Theorem~\ref{thm:n-a-priv-read}, the following bound holds
\begin{align}
P^{\text{read}}(\overline{\mc{Q}}^q_{\mc{X}}) &\geq P^{\text{read}}_{\text{n-a}}(\overline{\mc{Q}}^q_{\mc{X}})\label{eq:e-cell-down}\\
&\geq I(X;L_BB)_{\rho}-I(X;E)_{\rho},
\end{align}
where 
\begin{equation}
\rho_{XL_BBE}=\frac{1}{d^2}\sum_{x=0}^{d^2-1}\ket{x}\!\bra{x}_X\otimes \mc{U}^{\mc{Q}^{q,x}}_{B'\to BE}(\Phi_{X:L_BB'}).
\end{equation}
After a calculation, we find that $I(X;E)_{\rho}=0$ and $I(X;L_BB)_{\rho}=2(1-q)\log_2 d$. Therefore, from \eqref{eq:e-cell-up} and the above, we  conclude the statement of the theorem.
\end{proof}

\bigskip
From the above and \cite[Corollary 4]{DW17}, we conclude that there is no difference between the private reading capacity of the qudit erasure memory cell  and its reading capacity.

\section{Entanglement generation from a coherent memory cell or controlled isometry}\label{sec:coh-read}

In this section, we consider an entanglement distillation task between two parties Alice and Bob holding systems $X$ and $B$, respectively. The set up is similar to purified secret key generation when using a memory cell (see Section~\ref{sec:n-a-priv-read-coherent}). The goal of the protocol is as follows: Alice and Bob, who are spatially separated, try to generate a maximally entangled state between them by making coherent use of an isometric wiretap memory cell $\overline{\mc{M}}_{\mc{X}}=\{\mc{U}^{\mc{M}^x}_{B'\to BE}\}_{x\in\mc{X}}$ known to both parties. That is, Alice and Bob have access to the following controlled isometry:
\begin{equation}
U^{\overline{\mc{M}}_{\mc{X}}}_{XB'\to XBE}\coloneqq \sum_{x\in\mc{X}}\ket{x}\!\bra{x}_{X}\otimes U^{\mc{M}^x}_{B'\to BE},\label{eq:contrl-iso-coherent-mem-cell}
\end{equation}
such that $X$ and $E$ are inaccessible to Bob. Using techniques from \cite{DW05}, we can state an achievable rate of entanglement generation by coherently  using the memory cell.  

\begin{theorem}
The following
rate is achievable for entanglement generation when using the controlled isometry in \eqref{eq:contrl-iso-coherent-mem-cell}:
\begin{equation}
I(X\rangle L_BB)_{\omega},
\end{equation}
where $I(X\rangle L_BB)_{\omega}$ is the coherent information of state $\omega_{XL_BB}$ \eqref{eq:coh-info} such that 
\begin{equation}
\ket{\omega}_{XL_BBE}=\sum_{x}\sqrt{p_{X}(x)}|x\rangle_{X}\otimes U^{\mc{M}^x}_{B'\rightarrow
BE}|\psi\rangle_{L_BB'}.
\end{equation}
\end{theorem}

\begin{proof}
Let $\{x^{n}(m,k)\}_{m,k}$ denote a codebook for private reading, as discussed in Section~\ref{sec:na-priv-read}, and let
$\psi_{L_BB'}$ denote a pure state that can be fed in to each coherent use of the memory cell. The
codebook is such that for each $m$ and $k$, the codeword $x^{n}(m,k)$ is
unique. The rate of private reading is given by
\begin{equation}
I(X;L_BB)_{\rho}-I(X;E)_{\rho},
\end{equation}
where%
\begin{equation}
\rho_{XB'BE}=\sum_{x}p_{X}(x)|x\rangle\langle x|_{X}\otimes\mathcal{U}
_{B'\rightarrow BE}^{\mc{M}^x}(\psi_{L_BB'}).
\end{equation}
Note that the following equality holds%
\begin{equation}
I(X;L_BB)_{\rho}-I(X;E)_{\rho}=I(X\rangle L_BB)_{\omega},
\end{equation}
where%
\begin{equation}
\ket{\omega}_{XL_BBE}=\sum_{x}\sqrt{p_{X}(x)}|x\rangle_{X}\otimes U_{B'\rightarrow
BE}^{\mc{M}^x}|\psi\rangle_{L_BB'}. 
\end{equation} 
The code is such that there is a measurement $\Lambda
_{L_B^{n}B^{n}}^{m,k}$ for all $m,k$, for which%
\begin{equation}
\operatorname{Tr}\{\Lambda_{L_B^{n}B^{n}}^{m,k}\mathcal{M}_{{B'}^{n}\rightarrow
B^{n}}^{x^{n}(m,k)}(\psi_{L_BB'}^{\otimes n})\}\geq1-\varepsilon,
\end{equation}
and%
\begin{equation}
\frac{1}{2}\left\Vert \frac{1}{K}\sum_{k}\widehat{\mathcal{M}}_{{B'}^{n}%
\rightarrow E^{n}}^{x^{n}(m,k)}(\psi_{B'}^{\otimes n})-\sigma_{E^{n}%
}\right\Vert _{1}\leq\delta.\label{eq:security-condition}%
\end{equation}

From this private reading code, we construct a coherent reading code as
follows. Alice begins by preparing the state%
\begin{equation}
\frac{1}{\sqrt{MK}}\sum_{m,k}|m\rangle_{M_A}|k\rangle_{K_A}.
\end{equation}
Alice performs a unitary that implements the following isometry:
\begin{equation}
|m\rangle_{M_A}|k\rangle_{K_A}\rightarrow|m\rangle_{M_A}
|k\rangle_{K_A}|x^{n}(m,k)\rangle_{X^{n}},\label{eq:encoding-unitary}%
\end{equation}
so that the state above becomes%
\begin{equation}
\frac{1}{\sqrt{MK}}\sum_{m,k}|m\rangle_{M_A}|k\rangle_{K_A}|x^{n}(m,k)\rangle
_{X^{n}}.
\end{equation}
Bob prepares the state $|\psi\rangle_{L_BB'}^{\otimes n}$, so that the overall
state is%
\begin{equation}
\frac{1}{\sqrt{MK}}\sum_{m,k}|m\rangle_{M_A}|k\rangle_{K_A}|x^{n}(m,k)\rangle
_{X^{n}}|\psi\rangle_{L_BB'}^{\otimes n}.
\end{equation}
Now Alice and Bob are allowed to access $n$ instances of the controlled
isometry%
\begin{equation}
\sum_{x}|x\rangle\langle x|_{X}\otimes U_{B'\rightarrow BE}^{\mc{M}^x}%
,\label{eq:controlled-isometry-coherent-memory-cell}%
\end{equation}
and the state becomes%
\begin{equation}
\frac{1}{\sqrt{MK}}\sum_{m,k}|m\rangle_{M_A}|k\rangle_{K_A}|x^{n}(m,k)\rangle
_{X^{n}}U_{{B'}^{n}\rightarrow B^{n}E^{n}}^{\mc{M}^{x^{n}(m,k)}}|\psi\rangle_{L_BB'}^{\otimes
n}.
\end{equation}
Bob now performs the isometry%
\begin{equation}
\sum_{m,k}\sqrt{\Lambda_{L_B^{n}B^{n}}^{m,k}}\otimes|m\rangle_{M_{1}}%
|k\rangle_{K_{1}},
\end{equation}
and the resulting state is close to%
\begin{multline}
\frac{1}{\sqrt{MK}}\sum_{m,k}|m\rangle_{M_A}|k\rangle_{K_A}|x^{n}(m,k)\rangle
_{X^{n}} \\ \otimes U_{{B'}^{n}\rightarrow B^{n}E^{n}}^{x^{n}(m,k)}|\psi\rangle_{L_BB'}^{\otimes
n}|m\rangle_{M_{1}}|k\rangle_{K_{1}}.
\end{multline}
At this point, Alice locally uncomputes the unitary from
\eqref{eq:encoding-unitary} and discards the $X^{n}$ register, leaving the
following state:%
\begin{multline}
\frac{1}{\sqrt{MK}}\sum_{m,k}|m\rangle_{M_A}|k\rangle_{K_A}U_{{B'}^{n}\rightarrow
B^{n}E^{n}}^{\mc{M_A}^{x^{n}(m,k)}}|\psi\rangle_{L_BB'}^{\otimes n}\\ \otimes|m\rangle_{M_{1}%
}|k\rangle_{K_{1}}.
\end{multline}
Following the scheme of \cite{DW05} for entanglement distillation, she then performs a Fourier transform on
the register $K_A$ and measures it, obtaining an outcome $k^{\prime}%
\in\{0,\ldots,K-1\}$, leaving the following state:%
\begin{multline}
\frac{1}{\sqrt{MK}}\sum_{m,k}e^{2\pi ik^{\prime}k/K}|m\rangle_{M_A}%
\otimes U_{{B'}^{n}\rightarrow B^{n}E^{n}}^{\mc{M_A}^{x^{n}(m,k)}}|\psi\rangle_{L_BB'}^{\otimes
n}\\ \otimes|m\rangle_{M_{1}}|k\rangle_{K_{1}}.
\end{multline}
She communicates the outcome to Bob, who can then perform a local unitary on system $K_1$ to
bring the state to%
\begin{equation}
\frac{1}{\sqrt{MK}}\sum_{m,k}|m\rangle_{M_A}U_{{B'}^{n}\rightarrow B^{n}E^{n}%
}^{\mc{M}^{x^{n}(m,k)}}|\psi\rangle_{L_BB'}^{\otimes n}|m\rangle_{M_{1}}|k\rangle_{K_{1}}.
\end{equation}
Now consider that, conditioned on a value $m$ in register $M$, the local state
of Eve's register $E^n$ is given by%
\begin{equation}
\frac{1}{K_A}\sum_{k}\widehat{\mathcal{M}}_{{B'}^{n}\rightarrow E^{n}}^{x^{n}%
(m,k)}(\psi_{B'}^{\otimes n}).
\end{equation}
Thus, by invoking the security condition in \eqref{eq:security-condition} and
Uhlmann's theorem \cite{U76}, there exists a isometry $V_{L_B^{n}B^{n}K_{1}\rightarrow
\widetilde{B}}^{m}$ such that%
\begin{multline}
V_{L_B^{n}B^{n}K_{1}\rightarrow\widetilde{B}}^{m}\left[  \frac{1}{\sqrt{K_A}}%
\sum_{k}U_{{B'}^{n}\rightarrow B^{n}E^{n}}^{\mc{M}^{x^{n}(m,k)}}|\psi\rangle_{L_BB'}^{\otimes
n}|k\rangle_{K_{1}}\right] \\ \approx|\varphi^{\sigma}\rangle_{E^{n}%
\widetilde{B}}.
\end{multline}
Thus, Bob applies the controlled isometry%
\begin{equation}
\sum_{m}|m\rangle\langle m|_{M_{1}}\otimes V_{L_B^{n}B^{n}K_{1}\rightarrow
\widetilde{B}}^{m},
\end{equation}
and then the overall state is close to%
\begin{equation}
\frac{1}{\sqrt{M}}\sum_{m}|m\rangle_{M_A}|\varphi^{\sigma}\rangle_{E^{n}%
\widetilde{B}}|m\rangle_{M_{1}}.
\end{equation}
Bob now discards the register $\widetilde{B}$ and Alice and Bob are left with
a maximally entangled state that is locally equivalent to approximately $ n[I(X;L_BB)_{\rho
}-I(X;E)_{\rho}] = nI(X\rangle L_BB)_{\omega}$ ebits.
\end{proof}


\section{Discussion}\label{sec:dis}

In this work, we mainly focused on two different information processing tasks: entanglement distillation and secret key distillation using bipartite quantum interactions or bidirectional channels. We determined several bounds on the entanglement and secret-key-agreement capacities of bipartite quantum interactions. In deriving these bounds, we described communication protocols in the bidirectional setting, related to those discussed in \cite{BHLS03} and which generalize related point-to-point communication protocols. We introduced an entanglement measure called the bidirectional max-Rains information of a bidirectional channel and showed that it is a strong converse upper bound on the PPT-assisted quantum capacity of the given bidirectional channel. We also introduced a related entanglement measure called the bidirectional max-relative entropy of entanglement and showed that it is a strong converse bound on the LOCC-assisted secret-key-agreement capacity of a given bidirectional channel. When the bidirectional channels are either teleportation- or PPT-simulable, the upper bounds on the bidirectional quantum and bidirectional secret-key-agreement capacities depend only on the entanglement of an underlying resource state. If a bidirectional channel is bicovariant, then the underlying resource state can be taken to be the Choi state of the bidirectional channel.

Next, we introduced a private communication task called private reading. This task allows for secret key agreement between an encoder and a reader in the presence of a passive eavesdropper. Observing that access to an isometric wiretap memory cell by an encoder and the reader is a particular kind of bipartite quantum interaction, we were able to leverage our bounds on the LOCC-assisted bidirectional secret-key-agreement capacity to determine bounds on its private reading capacity. We also determined a regularized expression for the non-adaptive private reading capacity of an arbitrary wiretap memory cell. For particular classes of memory cells obeying certain symmetries, such that there is an adaptive-to-non-adaptive reduction in a reading protocol, as in \cite{DW17}, the private reading capacity and the non-adaptive private reading capacity are equal. We derived a single-letter, weak converse upper bound on the non-adaptive private reading capacity of an isometric wiretap memory cell in terms of the squashed entanglement. We also proved a strong converse upper bound on the private reading capacity of an isometric wiretap memory cell in terms of  the bidirectional max-relative entropy of entanglement. We applied our results to show that the private reading capacity and the reading capacity of the qudit erasure memory cell are equal. Finally, we determined an achievable rate at which entanglement can be generated between two parties who have coherent access to a memory cell. 
 

We have left open the question of determining a relation between the bidirectional max-Rains information and the bidirectional max-relative entropy of entanglement for an arbitrary bidirectional channel. However, we  strongly suspect that the bidirectional max-Rains information can never exceed the bidirectional max-relative entropy of entanglement. It would also be interesting to derive an upper bound on the bidirectional secret-key-agreement capacity in terms of the squashed entanglement. Another future direction would be to determine classes of memory cells for which the regularized expressions of the non-adaptive private reading capacities reduce to single-letter expressions. For this, one could consider memory cells consisting of degradable channels \cite{DS05,S08}. More generally, determining the private reading capacity of an arbitrary wiretap memory cell is an important open question.  



\bigskip
\textbf{Acknowledgements.}
We thank Koji Azuma, Aram Harrow, Cosmo Lupo, Bill Munro, Mio Murao, and George Siopsis for helpful discussions. SD acknowledges support from
the LSU Graduate School Economic Development Assistantship and the LSU Coates Conference Travel Award. MMW acknowledges support from the US Office of Naval Research and the National Science Foundation. Part of this work was completed during the workshop \textquotedblleft Beyond i.i.d.~in Information Theory,\textquotedblright\  hosted by the Institute for Mathematical Sciences, NUS Singapore,
24-28 July 2017.

\begin{appendix}

\section{Covariant channel}\label{app:cov-lemma}
\begin{proof}[Proof of Lemma~\ref{thm:cov-hol}]
Given is a group $G$  and a quantum channel $\mathcal{M}_{A\rightarrow B}$ that is covariant in the
following sense:
\begin{equation}
\mathcal{M}_{A\rightarrow B}(U_{A}^{g}\rho_{A}U_{A}^{g\dag})=V_{B}%
^{g}\mathcal{M}_{A\rightarrow B}(\rho_{A})V_{B}^{g\dag},\label{eq:cov-sym}
\end{equation}
for a set of unitaries $\{U_{A}^{g}\}_{g\in G}$ and $\{ V_{B}^{g} \}_{g \in G}$.

Let a Kraus representation of $\mathcal{M}_{A\rightarrow B}$ be given as%
\begin{equation}
\mathcal{M}_{A\rightarrow B}(\rho_{A})=\sum_{j}L^{j}\rho_{A}L^{j\dag}.
\end{equation}
We can rewrite \eqref{eq:cov-sym} as%
\begin{equation}
V_{B}^{g\dag}\mathcal{M}_{A\rightarrow B}(U_{A}^{g}\rho_{A}U_{A}^{g\dag}%
)V_{B}^{g}=\mathcal{M}_{A\rightarrow B}(\rho_{A}),
\end{equation}
which means that for all $g$, the following equality holds%
\begin{equation}
\sum_{j}L^{j}\rho_{A}L^{j\dag}=\sum_{j}V_{B}^{g\dag}L^{j}U_{A}^{g}\rho
_{A}\left(  V_{B}^{g\dag}L^{j}U_{A}^{g}\right)  ^{\dag}.
\end{equation}
Thus, the channel has two different Kraus representations $\{L^{j}\}_{j}$ and
$\{V_{B}^{g\dag}L^{j}U_{A}^{g}\}_{j}$, and these are necessarily related by a
unitary with matrix elements $w_{jk}^{g}$ \cite{Wbook17,Wat15}:
\begin{equation}
V_{B}^{g\dag}L^{j}U_{A}^{g}=\sum_{k}w_{jk}^{g}L^{k}.
\end{equation}
A canonical isometric extension $U_{A\rightarrow BE}^{\mathcal{M}}$ of
$\mathcal{M}_{A\rightarrow B}$ is given as%
\begin{equation}
U_{A\rightarrow BE}^{\mathcal{M}}=\sum_{j}L^{j}\otimes|j\rangle_{E},
\end{equation}
where $\{|j\rangle_{E}\}_j$ is an orthonormal basis.
Defining $W_{E}^{g}$ as the following unitary%
\begin{equation}
W_{E}^{g}|k\rangle_{E}=\sum_{j}w_{jk}^{g}|j\rangle_{E},
\end{equation}
where the states $|k\rangle_{E}$ are chosen from $\{|j\rangle_{E}\}_j$,
consider that%
\begin{align}
U_{A\rightarrow BE}^{\mathcal{M}}U_{A}^{g}  & =\sum_{j}L^{j}U_{A}^{g}%
\otimes|j\rangle_{E}\\
& =\sum_{j}V_{B}^{g}V_{B}^{g\dag}L^{j}U_{A}^{g}\otimes|j\rangle_{E}\\
& =\sum_{j}V_{B}^{g}\left[  \sum_{k}w_{jk}^{g}L^{k}\right]  \otimes
|j\rangle_{E}\\
& =V_{B}^{g}\sum_{k}L^{k}\otimes\sum_{j}w_{jk}^{g}|j\rangle_{E}\\
& =V_{B}^{g}\sum_{k}L^{k}\otimes W_{E}^{g}|k\rangle_{E}\\
& =\left(  V_{B}^{g}\otimes W_{E}^{g}\right)  U_{A\rightarrow BE}%
^{\mathcal{M}}.
\end{align}
This concludes the proof.
\end{proof}

\section{Bicovariant channels and teleportation simulation}\label{app:bicov}

\begin{proof}[Proof of Proposition~\ref{prop:bicov}]
Let $\mathcal{N}_{A^{\prime}B^{\prime}\rightarrow AB}$ be a bidirectional
quantum channel, and let $G$ and $H$ be groups with unitary representations
$g\rightarrow\mathcal{U}_{A^{\prime}}(g)$ and $h\rightarrow V_{B^{\prime}}(h)$
and $(g,h)\rightarrow W_{A}(g,h)$ and $(g,h)\rightarrow T_{B}(g,h)$,
such that%
\begin{align}
\frac{1}{\left\vert G\right\vert }\sum_{g}\mathcal{U}_{A^{\prime}%
}(g)(X_{A^{\prime}}) &  =\operatorname{Tr}\{X_{A^{\prime}}\}\pi_{A^{\prime}%
},\label{eq:one-design-cond}\\
\frac{1}{\left\vert H\right\vert }\sum_{h}\mathcal{V}_{B^{\prime}%
}(h)(Y_{B^{\prime}}) &  =\operatorname{Tr}\{Y_{B^{\prime}}\}\pi_{B^{\prime}%
},
\end{align}
and
\begin{multline}
\mathcal{N}_{A^{\prime}B^{\prime}\rightarrow AB}((\mathcal{U}_{A^{\prime}%
}(g)\otimes\mathcal{V}_{B^{\prime}}(h))(\rho_{A^{\prime}B^{\prime}})) \\
=(\mathcal{W}_{A}(g,h)\otimes\mathcal{T}_{B}(g,h))(\mathcal{N}_{A^{\prime
}B^{\prime}\rightarrow AB}(\rho_{A^{\prime}B^{\prime}})),
\end{multline}
where $X_{A^{\prime}}\in\mathcal{B}(\mathcal{H}_{A^{\prime}})$, $Y_{B^{\prime
}}\in\mathcal{B}(\mathcal{H}_{B^{\prime}})$, and $\pi$ denotes the maximally
mixed state. Consider that%
\begin{equation}
\frac{1}{\left\vert G\right\vert }\sum_{g}\mathcal{U}_{A^{\prime\prime}%
}(g)(\Phi_{A^{\prime\prime}A^{\prime}})=\pi_{A^{\prime\prime}}\otimes
\pi_{A^{\prime}},\label{eq:max-ent-cov-action}%
\end{equation}
where $\Phi$ denotes a maximally entangled state and $A^{\prime\prime}$ is a
system isomorphic to $A^{\prime}$. Similarly,
\begin{equation}
\frac{1}{\left\vert H\right\vert }\sum_{h}\mathcal{V}_{B^{\prime\prime}%
}(h)(\Phi_{B^{\prime\prime}B^{\prime}})=\pi_{B^{\prime\prime}}\otimes
\pi_{B^{\prime}}.
\end{equation}
Note that in order for $\{U_{A^{\prime}}^{g}\}$ to satisfy
\eqref{eq:one-design-cond}, it is necessary that $\left\vert A^{\prime
}\right\vert ^{2}\leq\left\vert G\right\vert $ \cite{AMTW00}. Similarly, it is
necessary that $\left\vert B^{\prime}\right\vert ^{2}\leq\left\vert
H\right\vert $. Consider the POVM $\{E_{A^{\prime\prime}L_{A}}^{g}\}_{g}$,
with each element $E_{A^{\prime\prime}L_{A}}^{g}$ defined as%
\begin{equation}
E_{A^{\prime\prime}L_{A}}^{g}:=\frac{\left\vert A^{\prime}\right\vert
^{2}}{\left\vert G\right\vert }U_{A^{\prime\prime}}^{g}\Phi_{A^{\prime\prime
}L_{A}}\left(  U_{A^{\prime\prime}}^{g}\right)  ^{\dag}.
\end{equation}
It follows from the fact that $\left\vert A^{\prime}\right\vert ^{2}%
\leq\left\vert G\right\vert $ and \eqref{eq:max-ent-cov-action}\ that
$\{E_{A^{\prime\prime}L_{A}}^{g}\}_{g}$ is a valid POVM. Similarly, we define
the POVM\ $\{F_{B^{\prime\prime}L_{B}}^{h}\}_{h}$ as%
\begin{equation}
F_{B^{\prime\prime}L_{B}}^{h}:=\frac{\left\vert B^{\prime}\right\vert
^{2}}{\left\vert H\right\vert }V_{B^{\prime\prime}}^{h}\Phi_{B^{\prime\prime
}L_{B}}\left(  V_{B^{\prime\prime}}^{h}\right)  ^{\dag}%
\end{equation}

The simulation of the channel $\mathcal{N}_{A^{\prime}B^{\prime}\rightarrow
AB}$ via teleportation begins with a state $\rho_{A^{\prime\prime}%
B^{\prime\prime}}$ and a shared resource $\theta_{L_{A}ABL_{B}}=\mathcal{N}%
_{A^{\prime}B^{\prime}\rightarrow AB}(\Phi_{L_{A}A^{\prime}}\otimes
\Phi_{B^{\prime}L_{B}})$. The desired outcome is for the receivers to receive
the state $\mathcal{N}_{A^{\prime}B^{\prime}\rightarrow AB}(\rho_{A^{\prime
}B^{\prime}})$ and for the protocol to work independently of the input state
$\rho_{A^{\prime}B^{\prime}}$. The first step is for the senders to locally
perform the measurement $\{E_{A^{\prime\prime}L_{A}}^{g}\otimes F_{B^{\prime
\prime}L_{B}}^{h}\}_{g,h}$ and then send the outcomes $g$\ and $h$ to the
receivers. Based on the outcomes $g$ and $h$, the receivers then perform
$W_{A}^{g,h}$ and $T_{B}^{g,h}$. The following analysis demonstrates that this
protocol works, by simplifying the form of the post-measurement state:
\begin{widetext}
\begin{align}
&  \left\vert G\right\vert \left\vert H\right\vert \operatorname{Tr}%
_{A^{\prime\prime}L_{A}B^{\prime\prime}L_{B}}\{(E_{A^{\prime\prime}L_{A}}%
^{g}\otimes F_{B^{\prime\prime}L_{B}}^{h})(\rho_{A^{\prime\prime}%
B^{\prime\prime}}\otimes\theta_{L_{A}ABL_{B}})\}\nonumber\\
&  =\left\vert A^{\prime}\right\vert ^{2}\left\vert B^{\prime}\right\vert
^{2}\operatorname{Tr}_{A^{\prime\prime}L_{A}B^{\prime\prime}L_{B}%
}\{[U_{A^{\prime\prime}}^{g}\Phi_{A^{\prime\prime}L_{A}}\left(  U_{A^{\prime
\prime}}^{g}\right)  ^{\dag}\otimes V_{B^{\prime\prime}}^{h}\Phi
_{B^{\prime\prime}L_{B}}\left(  V_{B^{\prime\prime}}^{h}\right)  ^{\dag}%
](\rho_{A^{\prime\prime}B^{\prime\prime}}\otimes\theta_{L_{A}ABL_{B}})\}\\
&  =\left\vert A^{\prime}\right\vert ^{2}\left\vert B^{\prime}\right\vert
^{2}\langle\Phi|_{A^{\prime\prime}L_{A}}\otimes\langle\Phi|_{B^{\prime\prime
}L_{B}}\left(  U_{A^{\prime\prime}}^{g}\otimes V_{B^{\prime\prime}}%
^{h}\right)  ^{\dag}(\rho_{A^{\prime\prime}B^{\prime\prime}}\otimes
\theta_{L_{A}ABL_{B}})(U_{A^{\prime\prime}}^{g}\otimes V_{B^{\prime\prime}%
}^{h})|\Phi\rangle_{A^{\prime\prime}L_{A}}\otimes|\Phi\rangle_{B^{\prime
\prime}L_{B}}\\
&  =\left\vert A^{\prime}\right\vert ^{2}\left\vert B^{\prime}\right\vert
^{2}\langle\Phi|_{A^{\prime\prime}L_{A}}\otimes\langle\Phi|_{B^{\prime\prime
}L_{B}}\left[  \left(  U_{A^{\prime\prime}}^{g}\otimes V_{B^{\prime\prime}%
}^{h}\right)  ^{\dag}\rho_{A^{\prime\prime}B^{\prime\prime}}(U_{A^{\prime
\prime}}^{g}\otimes V_{B^{\prime\prime}}^{h})\right]  \nonumber\\
&  \qquad\qquad\otimes\mathcal{N}_{A^{\prime}B^{\prime}\rightarrow AB}%
(\Phi_{L_{A}A^{\prime}}\otimes\Phi_{B^{\prime}L_{B}}))|\Phi\rangle
_{A^{\prime\prime}L_{A}}\otimes|\Phi\rangle_{B^{\prime\prime}L_{B}}\\
&  =\left\vert A^{\prime}\right\vert ^{2}\left\vert B^{\prime}\right\vert
^{2}\langle\Phi|_{A^{\prime\prime}L_{A}}\otimes\langle\Phi|_{B^{\prime\prime
}L_{B}}\left[  \left(  U_{L_{A}}^{g}\otimes V_{L_{B}}^{h}\right)  ^{\dag}%
\rho_{L_{A}L_{B}}(U_{L_{A}}^{g}\otimes V_{L_{B}}^{h})\right]  ^{\ast
}\nonumber\\
&  \qquad\qquad\mathcal{N}_{A^{\prime}B^{\prime}\rightarrow AB}(\Phi
_{L_{A}A^{\prime}}\otimes\Phi_{B^{\prime}L_{B}}))|\Phi\rangle_{A^{\prime
\prime}L_{A}}\otimes|\Phi\rangle_{B^{\prime\prime}L_{B}}%
.\label{eq:cov-tp-simul-block-1}%
\end{align}
\end{widetext}
The first three equalities follow by substitution and some rewriting. The
fourth equality follows from the fact that%
\begin{equation}
\langle\Phi|_{A^{\prime}A}M_{A^{\prime}}=\langle\Phi|_{A^{\prime}A}M_{A}%
^{\ast}\label{eq:ricochet-prop}%
\end{equation}
for any operator $M$ and where $\ast$ denotes the complex conjugate, taken
with respect to the basis in which $|\Phi\rangle_{A^{\prime}A}$ is defined.
Continuing, we have that%
\begin{widetext}
\begin{align}
\text{Eq.~}\eqref{eq:cov-tp-simul-block-1} &  =\left\vert A^{\prime}\right\vert
\left\vert B^{\prime}\right\vert \operatorname{Tr}_{L_{A}L_{B}}\left\{
\left[  \left(  U_{L_{A}}^{g}\otimes V_{L_{B}}^{h}\right)  ^{\dag}\rho
_{L_{A}L_{B}}(U_{L_{A}}^{g}\otimes V_{L_{B}}^{h})\right]  ^{\ast}%
\mathcal{N}_{A^{\prime}B^{\prime}\rightarrow AB}(\Phi_{L_{A}A^{\prime}}%
\otimes\Phi_{B^{\prime}L_{B}}))\right\}  \\
&  =\left\vert A^{\prime}\right\vert \left\vert B^{\prime}\right\vert
\operatorname{Tr}_{L_{A}L_{B}}\left\{  \mathcal{N}_{A^{\prime}B^{\prime
}\rightarrow AB}\left(  \left[  \left(  U_{A^{\prime}}^{g}\otimes
V_{B^{\prime}}^{h}\right)  ^{\dag}\rho_{A^{\prime}B^{\prime}}(U_{A^{\prime}%
}^{g}\otimes V_{B^{\prime}}^{h})\right]  ^{\dag}\left(  \Phi_{L_{A}A^{\prime}%
}\otimes\Phi_{B^{\prime}L_{B}}\right)  \right)  \right\}  \\
&  =\mathcal{N}_{A^{\prime}B^{\prime}\rightarrow AB}\left(  \left[  \left(
U_{A^{\prime}}^{g}\otimes V_{B^{\prime}}^{h}\right)  ^{\dag}\rho_{A^{\prime
}B^{\prime}}(U_{A^{\prime}}^{g}\otimes V_{B^{\prime}}^{h})\right]  ^{\dag
}\right)  \\
&  =\mathcal{N}_{A^{\prime}B^{\prime}\rightarrow AB}\left(  \left(
U_{A^{\prime}}^{g}\otimes V_{B^{\prime}}^{h}\right)  ^{\dag}\rho_{A^{\prime
}B^{\prime}}(U_{A^{\prime}}^{g}\otimes V_{B^{\prime}}^{h})\right)  \\
&  =\left(  W_{A}^{g,h}\otimes T_{B}^{g,h}\right)  ^{\dag}\mathcal{N}%
_{A^{\prime}B^{\prime}\rightarrow AB}\left(  \rho_{A^{\prime}B^{\prime}%
}\right)  (W_{A}^{g,h}\otimes T_{B}^{g,h})
\end{align}
\end{widetext}
The first equality follows because $\left\vert A\right\vert \langle
\Phi|_{A^{\prime}A}\left(  I_{A^{\prime}}\otimes M_{AB}\right)  |\Phi
\rangle_{A^{\prime}A}=\operatorname{Tr}_{A}\{M_{AB}\}$ for any operator
$M_{AB}$. The second equality follows by applying the conjugate transpose of
\eqref{eq:ricochet-prop}. The final equality follows from the covariance
property of the channel.

Thus, if the receivers finally perform the unitaries $W_{A}^{g,h}\otimes
T_{B}^{g,h}$ upon receiving $g$ and $h$ via a classical channel from the
senders, then the output of the protocol is $\mathcal{N}_{A^{\prime}B^{\prime
}\rightarrow AB}\left(  \rho_{A^{\prime}B^{\prime}}\right)  $, so that this
protocol simulates the action of the channel $\mathcal{N}$ on the state $\rho$.
\end{proof}
 
\section{Qudit system and Heisenberg--Weyl group}
\label{app:qudit}

Here we introduce some basic notations and definitions related to qudit systems. A system represented with a $d$-dimensional Hilbert space is called a qu$d$it system. Let $J_{B'}=\{|j\>_{B'}\}_{j\in \{0,\ldots,d-1\}}$ be a computational orthonormal basis of $\mc{H}_{B'}$ such that $\dim(\mc{H}_{B'})=d$. There exists a unitary operator called \textit{cyclic shift operator} $X(k)$ that acts on the orthonormal states as follows:
\begin{equation}
\forall |j\>_{B'}\in J_{B'}:\ \ X(k)|j\>=|k\oplus j\>,
\end{equation}
where $\oplus$ is a cyclic addition operator, i.e., $k\oplus j:= (k+j)\ \textnormal{mod}\ d$. There also exists another unitary operator called the \textit{phase operator} $Z(l)$ that acts on the qudit computational basis states as
\begin{equation}
\forall |j\>_{B'}\in J_{B'}:\ \ Z(l)|j\>=\exp\(\frac{\iota 2\pi lj}{d}\)|j\>. 
\end{equation}
The $d^2$ operators $\{X(k)Z(l)\}_{k,l\in\{0,\ldots,d-1\}}$ are known as the Heisenberg--Weyl operators. Let $\sigma(k,l):=X(k)Z(l)$. 
The maximally entangled state $\Phi_{R:B'}$ of qudit systems $RB'$ is given as
$
|\Phi\>_{RB'}:=\frac{1}{\sqrt{d}}\sum_{j=0}^{d-1}|j\>_R|j\>_{B'},
$
and we define 
$
|\Phi^{k,l}\>_{RB'}:=(I_R\otimes\sigma^{k,l}_{B'})|\Phi\>_{R:B'}.
$
The $d^2$ states $\{|\Phi^{k,l}\>_{RB'}\}_{k,l\in\{0,\ldots,d-1\}}$ form a complete, orthonormal basis:
\begin{align}
\<\Phi^{k_1,l_1}|\Phi^{k_2,l_2}\>&=\delta_{k_1,k_2}\delta_{l_1,l_2},\\
\sum_{k,l=0}^{d-1}|\Phi^{k,l}\>\<\Phi^{k,l}|_{RB'}&=I_{RB'}.
\end{align}

Let $\mc{W}$ be a discrete set such that $|\mc{W}|=d^2$. There exists one-to-one mapping $\{(k,l)\}_{k,l\in\{0,d-1\}}\leftrightarrow \{w\}_{w\in\mc{W}}$. For example, we can use the following map: $w=k+d\cdot l$ for $\mc{W}=\{0,\ldots,d^2-1\}$. This allows us to define $\sigma^w:=\sigma(k,l)$ and $\Phi^w_{RB'}:=\Phi^{k,l}
_{RB'}$.  Let the set of $d^2$ Heisenberg--Weyl operators be denoted as
\begin{equation}\label{eq:HW-op}
\mathbf{H}:=\{ \sigma^w\}_{w\in\mc{W}}=\{X(k)Z(l)\}_{k,l\in\{0,\ldots,d-1\}},
\end{equation}
and we refer to $\mathbf{H}$  as the Heisenberg--Weyl group. 

\end{appendix}

\bibliography{smallbib}

\end{document}